\documentclass[11pt,letterpaper]{article}
\usepackage[letterpaper, left=1in, right=1in, top=0.9in, bottom=0.9in]{geometry}
\usepackage[american]{babel}
\usepackage[normalem]{ulem}
\usepackage{amsmath, amssymb, cases, amsthm}
\usepackage{thmtools}
\usepackage[shortlabels]{enumitem}
\usepackage{mdframed}
\usepackage{bbm}
\usepackage{bm}
\usepackage{microtype}
\usepackage{xcolor}
\usepackage{makecell}
\usepackage{mathtools}
\usepackage{algorithmic}
\usepackage[procnumbered,ruled,vlined,linesnumbered]{algorithm2e}
\usepackage{float}
\usepackage{varwidth}
\usepackage{modletter}
\usepackage{tcolorbox}
\usepackage{tikz}
\usetikzlibrary{arrows.meta,positioning,fit,calc}

\newtcolorbox{construction}[2][]
{
	colframe = gray!50,
	colback  = gray!10,
	coltitle = gray!10!black,
	left*=0mm, 
	before skip = 10pt,
	after skip = 10pt,
	title    = \textbf{\space\space #2},
	#1,
}

\SetKwInput{KwData}{Input}
\SetKwInput{KwResult}{Output}
\SetKwInput{KwGlobalVar}{Global variables}

\usepackage{xcolor}
\definecolor{ForestGreen}{rgb}{0.1333,0.5451,0.1333}
\definecolor{DarkRed}{rgb}{0.80,0,0}
\definecolor{Red}{rgb}{1,0,0}
\usepackage[linktocpage=true,
pagebackref=true,colorlinks,
linkcolor=DarkRed,citecolor=ForestGreen,
bookmarks,bookmarksopen,bookmarksnumbered]
{hyperref}

\usepackage[capitalize,nosort,nameinlink]{cleveref}

\declaretheorem[numberwithin=section,refname={Theorem,Theorems},Refname={Theorem,Theorems}]{theorem}
\declaretheorem[numberwithin=section,name=Theorem,refname={Theorem,Theorems},Refname={Theorem,Theorems}]{thm}
\declaretheorem[numberlike=theorem]{lemma}

\declaretheorem[numberlike=theorem]{corollary}
\declaretheorem[numberlike=theorem,name=Corollary,refname={Corollary,Corollaries},Refname={Corollary,Corollaries}]{cor}
\declaretheorem[numberlike=theorem,style=definition]{definition}

\theoremstyle{definition}

\def\final{1}  %
\ifnum\final=0  %
\newcommand{\todo}[1]{{\color{red}[{\tiny TODO: \bf #1}]\marginpar{\color{red}*}}}
\newcommand{\yonggang}[1]{{\color{blue}[{\tiny Yonggang: \bf #1}]\marginpar{\color{blue}*}}}
\newcommand{\thatchaphol}[1]{{\color{purple}[{\tiny Thatchaphol: \bf #1}]\marginpar{\color{purple}*}}}
\else %
\newcommand{\yonggang}[1]{}
\newcommand{\thatchaphol}[1]{}
\newcommand{\todo}[1]{}
\fi

\newcommand{\eps}{\epsilon}
\newcommand{\poly}{\mathrm{poly}}
\newcommand{\polylog}{\mathrm{polylog}}
\newcommand{\dist}{\mathrm{dist}}
\newcommand{\Val}{\mathrm{val}}

\newcommand{\vol}{\mathrm{vol}}

\newcommand{\Start}{\mathrm{Start}}
\newcommand{\End}{\mathrm{End}}
\newcommand{\Supp}{\mathrm{Supp}}
\newcommand{\Ohs}{\cO_{\mathrm{HS}}}
\newcommand{\Odp}{\cO_{\mathrm{DP}}}
\newcommand{\mincut}{\mathrm{MinCut}}

\newcommand{\amfmc}[1]{#1\text{-}\mathsf{ApxMFMC}}
\newcommand{\amfmcDAG}[1]{#1\text{-}\mathsf{ApxMFMC}\text{-}\mathsf{DAG}}

\newcommand{\Dem}{\mathrm{Dem}}

\newcommand{\Cong}{\mathrm{Cong}}

\renewcommand{\path}{\mathrm{Path}}

\newcommand{\Pre}{\mathrm{Pre}}
\newcommand{\Dec}{\mathrm{Suf}}
\newcommand{\HL}{\hat{\pi}}

\newcommand{\Kap}{\gamma}
\newcommand{\dd}{\mathbf{d}}
\newcommand{\apxDSSSP}[1]{#1\text{-}\mathsf{ApxSSSP}\text{-}\mathsf{DAG}}
\newcommand{\apxSSSP}[1]{#1\text{-}\mathsf{ApxSSSP}}
\newcommand{\Out}[1]{{#1}^{\mathrm{last}}}
\newcommand{\In}[1]{{#1}^{\mathrm{first}}}

\newcommand\ff{\mathit{f}}
\newcommand\Erev{E^{\mathrm{rev}}}

\newcommand\DDelta{\boldsymbol{\mathit{\Delta}}}
\newcommand\nnabla{\boldsymbol{\mathit{\nabla}}}

\newcommand{\DDE}{\textsc{DistDAGProj}}
\newcommand{\To}{T^{\mathrm{out}}}
\newcommand{\Ti}{T^{\mathrm{in}}}

\newcommand{\Erem}{{E^{\mathrm{rem}}\xspace}}

\newcommand{\SL}{s_{\mathrm{LDD}}}

\global\long\def\maxflow{\mathrm{maxflow}}

 \usepackage{tikz}
\usetikzlibrary{arrows.meta,calc,positioning}

\newcommand{\CellW}{2.1}   %
\newcommand{\CellH}{1.30}  %

\newcommand{\GridSep}{1.10}   %
\newcommand{\LabelShift}{0.00}

\newcommand{\HeaderOff}{0.20}   %
\newcommand{\TitleDrop}{0.45}   %

\newcommand{\RowSep}{3.25}

\newcommand{\DownStartY}{1.13}
\newcommand{\DownEndY}{0.78}

\newcommand{\LongStartX}{0.90}
\newcommand{\LongEndX}{2.15}
\newcommand{\LongCtrlX}{1.50}
\newcommand{\LongTopY}{1.70}
\newcommand{\LongBotY}{0.25}
\newcommand{\LongBendTop}{0.20}   %
\newcommand{\LongBendBot}{-0.20}  %

\newcommand{\ShortStartX}{1.90}
\newcommand{\ShortEndX}{2.15}
\newcommand{\ShortCtrlX}{2.00}
\newcommand{\ShortTopY}{1.35}
\newcommand{\ShortBotY}{0.68}
\newcommand{\ShortBendTop}{-0.00}
\newcommand{\ShortBendBot}{0.00}

\newcommand{\HXXLeftX}{0.70}
\newcommand{\HXXRightX}{1.10}
\newcommand{\HXXY}{0.50}
\newcommand{\HXXLabelX}{0.80}
\newcommand{\HXXLabelY}{0.50}

\newcommand{\RozX}{2.60}
\newcommand{\RozYBot}{0.90}
\newcommand{\RozYTop}{1.40}
\newcommand{\RozLabelX}{2.62}
\newcommand{\RozLabelY}{1.25}

\newcommand{\VijXLeft}{1.60}   %
\newcommand{\VijXRight}{2.035}  %
\newcommand{\VijY}{0.50}       %
\newcommand{\VijLabelX}{1.7}  %
\newcommand{\VijLabelY}{0.45}  %

\newcommand{\MadryStartX}{0.85}   %
\newcommand{\MadryStartY}{0.15}
\newcommand{\MadryCtrlAX}{1.2}
\newcommand{\MadryCtrlAY}{-0.15}
\newcommand{\MadryCtrlBX}{2}
\newcommand{\MadryCtrlBY}{-0.15}
\newcommand{\MadryEndX}{2.15}     %
\newcommand{\MadryEndY}{0.15}
\newcommand{\MadryLabelX}{1.20}
\newcommand{\MadryLabelY}{-0.10}

\newcommand{\HeadSize}{\scriptsize}
\newcommand{\CellSize}{\scriptsize}

\definecolor{cellgreen}{HTML}{DDF3D5}
\definecolor{cellpeach}{HTML}{F5D9CC}
\colorlet{LiText}{green!50!black}
\colorlet{RozText}{red!70!black}
\colorlet{GridLine}{black!70}

\tikzset{
	grid/.style={line width=.7pt, draw=GridLine},
	subgrid/.style={line width=.5pt, draw=black!60},
	header/.style={font=\HeadSize\bfseries, align=center},
	celltext/.style={font=\CellSize, align=center},
	midlabel/.style={font=\scriptsize, align=center},
	trivial/.style={densely dotted, line width=0.85pt, draw=gray!75, -{Latex[length=2.0mm]}},
	strong/.style={line width=1.1pt, draw=black, -{Latex[length=2.3mm]}}
}

\newcommand{\OneGrid}[5]{%
	\begin{scope}[shift={(#1,0)}]
		#2
		\draw[grid] (0,0) rectangle (3,2);
		\draw[subgrid] (1,0) -- (1,2);
		\draw[subgrid] (2,0) -- (2,2);
		\draw[subgrid] (0,1) -- (3,1);
		\node[header] at (.5,2+\HeaderOff) {Undirected};
		\node[header] at (1.5,2+\HeaderOff) {DAG};
		\node[header] at (2.5,2+\HeaderOff) {Directed};
		
		\foreach \X in {0.5,1.5,2.5} { \draw[trivial] (\X,\DownStartY) -- (\X,\DownEndY); }
		
		\draw[trivial] (\LongStartX,\LongTopY)
		.. controls (\LongCtrlX,\LongTopY+\LongBendTop) .. (\LongEndX,\LongTopY);
		\draw[trivial] (\LongStartX,\LongBotY)
		.. controls (\LongCtrlX,\LongBotY+\LongBendBot) .. (\LongEndX,\LongBotY);
		
		\draw[trivial] (\ShortStartX,\ShortTopY)
		.. controls (\ShortCtrlX,\ShortTopY+\ShortBendTop) .. (\ShortEndX,\ShortTopY);
		\draw[trivial] (\ShortStartX,\ShortBotY)
		.. controls (\ShortCtrlX,\ShortBotY+\ShortBendBot) .. (\ShortEndX,\ShortBotY);
		
		#3
		#4
		\node[font=\small\bfseries] at (1.5,-\TitleDrop) {#5};
	\end{scope}%
}

\begin{document}
\sloppy

\title{
DAG Projections:\\ Reducing Distance and Flow Problems to DAGs
}
\author{Bernhard Haeupler\thanks{
        INSAIT, Sofia University ``St.~Kliment Ohridski'' and ETH Zürich,
        \texttt{bernhard.haeupler@insait.ai}.
        Partially funded by the Ministry of Education and Science of Bulgaria's support for INSAIT as part of the Bulgarian National Roadmap for Research Infrastructure and through the European Research Council (ERC) under the European Union's Horizon 2020 research and innovation program (ERC grant agreement 949272).} \and 
Yonggang Jiang\thanks{MPI-INF and Saarland University, Germany, \texttt{yjiang@mpi-inf.mpg.de}. Part of this work was done while visiting INSAIT.  Supported by Google PhD fellowship.} \and   
Thatchaphol Saranurak\thanks{
        University of Michigan,
        \texttt{thsa@umich.edu}.
        Supported by NSF Grant CCF-2238138 and a Sloan Fellowship. Part of this work was done at INSAIT. Partially funded by the Ministry of Education and Science of Bulgaria's support for INSAIT as part of the Bulgarian National Roadmap for Research Infrastructure. }
}
\date{}
\maketitle
\pagenumbering{gobble}
\begin{abstract}
We show that every directed graph $G$ with $n$ vertices and $m$ edges admits a directed acyclic graph (DAG) with $m^{1+o(1)}$ edges, called a \emph{DAG projection}, that can either $(1+1/\polylog (n))$-approximate distances between all pairs of vertices $(s,t)$ in $G$, or $n^{o(1)}$-approximate  maximum flow between all pairs of vertex subsets $(S,T)$ in $G$. Previous similar results suffer a $\Omega(\log n)$ approximation factor for distances \cite{assadi2025covering,filtser2025stochastic} and, for maximum flow, no prior result of this type is known. 

Our DAG projections admit $m^{1+o(1)}$-time constructions. Further, they admit almost-optimal parallel constructions, i.e., algorithms with $m^{1+o(1)}$ work and $m^{o(1)}$ depth, assuming the ones for approximate shortest path or maximum flow \emph{on DAGs}, even when the input $G$ is not a DAG. 

DAG projections immediately transfer results on DAGs, usually simpler and more efficient, to directed graphs. As examples, we improve the state-of-the-art of $(1+\eps)$-approximate distance preservers \cite{HoppenworthXX25} and single-source minimum cut  \cite{cheung2013graph}, and obtain simpler construction of $(n^{1/3},\epsilon)$-hop-set  \cite{kogan2022new,bernstein2023closing} and combinatorial max flow algorithms \cite{bernstein2024maximum,BBLST25}.

Finally, via DAG projections, we reduce major open problems on almost-optimal parallel algorithms for \emph{exact} single-source shortest paths (SSSP) and maximum flow to easier settings:
\begin{itemize}
\item From exact directed SSSP to exact \emph{undirected} ones, {\small \par}
\item From exact directed SSSP to $(1+1/\polylog(n))$-approximation on DAGs, and{\small \par}
\item From exact directed maximum flow to $n^{o(1)}$-approximation on DAGs.{\small \par}
\end{itemize}

\end{abstract}

\clearpage
\tableofcontents
\clearpage
\pagenumbering{arabic}
\section{Introduction}

Distance and maximum flow structures of \emph{undirected} graphs can be approximated with trees, using \emph{tree covers} \cite{mendel2007ramsey} and \emph{tree cut sparsifier} \cite{racke2002minimizing}, respectively. Both objects are highly influential across many areas including data structures \cite{mendel2007ramsey,thorup2005approximate}, fast algorithms \cite{racke2014computing,van2024almost}, and online algorithm \cite{bartal1996probabilistic,alon2006general}.\footnote{There are probabilistic version of tree covers and tree cut sparsifiers called \emph{probabilistic/stochastic tree embedding} \cite{bartal1996probabilistic,fakcharoenphol2003tight} and \emph{probabilistic tree cut sparsifiers} \cite{racke2008optimal}, respectively. } This motivates the research program on showing analogous results in directed graphs. 

In this paper, we make significant progress in this direction and obtain many applications. Below, we first review tree covers and tree cut sparsifier, as well as the state-of-the-art of their directed analogous objects summarized in \Cref{tab:history}. 

For any undirected or directed graph $G$ with $n$ vertices and $m$ edges, vertex pair $s,t\in V(G)$, and vertex subset pair $S,T\subseteq V(G)$, let $\dist_{G}(s,t)$ denote the distance from $s$ to $t$ and $\maxflow_{G}(S,T)$ denote the maximum flow value from $S$ to $T$.

\paragraph{Approximating Distances.}

A \emph{tree cover }of a graph $G$ is a collection ${\cal T}$ of trees where $\dist_{G}(s,t)\le\min_{T\in{\cal T}}\dist_{T}(s,t)\le\alpha\cdot\dist_{G}(s,t)$ for every $s,t\in V(G)$ and $\alpha$ is called the approximation factor. Every\emph{ undirected} graph admits, for any $k\ge1$, a tree cover ${\cal T}$ containing $O(kn^{1/k})$ trees with $O(k)$ approximation and total size $|\sum_{T\in{\cal T}}E(T)|=O(kn^{1+1/k})$ \cite{mendel2007ramsey}. 

Motivated by the fact that directed cyclic graphs (DAGs) are usually more algorithmic friendly than general directed graphs,  Assadi, Hoppenworth, and Wein \cite{assadi2025covering} 
recently introduced a \emph{DAG cover} as a directed analog of a tree cover. A DAG cover of a graph $G$ is a collection ${\cal D}$ of DAGs where $\dist_{G}(s,t)\le\min_{D\in{\cal D}}\dist_{D}(s,t)\le\alpha\cdot\dist_{G}(s,t).$ They showed that every directed graph with edge weights from $\{1,2,\dots,\poly(n)\}$ admits a DAG cover containing $O(\log n)$ DAGs with $O(\log^{3}n\log\log n)$ approximation and total size $|\sum_{D\in{\cal D}}E(D)|=O(m\log^{3}n)$. They also gave a $\tilde{O}(m)$-time algorithm to construct it. Later, Filtser \cite{filtser2025stochastic} strictly improved the approximation factor to $O(\log n\log\log n)$, while using the same size and construction time.%

Both \cite{assadi2025covering,filtser2025stochastic} left as an open problem whether their approximation factor can be improved further. 

\paragraph{Approximating Maximum Flow.}

A \emph{tree cut sparsifier} of $G$  is a tree $T$ where, for every $S,T\subseteq V(G)$, $\maxflow_{G}(S,T)\le\maxflow_{T}(S,T)\le\alpha\cdot\maxflow_{G}(S,T)$. Every undirected graph admits a tree cut sparsifier $T$ with $O(\log n\log\log n)$ approximation and size $|E(T)|\le2n$ \cite{racke2014improved}. However, it remains unknown if there exists any set of DAGs that could approximate maximum flows of a directed graph.

\begin{table}
\centering{

\footnotesize{

\begin{tabular}{|c|c|c|c|}
\hline 
\textbf{Distance} & From $\rightarrow$ To & Approximation & Total size\tabularnewline
\hline 
\hline 
\cite{mendel2007ramsey} & \textcolor{brown}{undirected $\rightarrow$ tree} & $O(k)$  &  $O(kn^{1+1/k})$ \tabularnewline
\hline 
\cite{bartal2019covering} & \textcolor{brown}{undirected $\rightarrow$ tree} & $O(n^{1/k}\log^{1-1/k}n)$ & $nk$\tabularnewline
\hline 
\cite{assadi2025covering} & directed $\rightarrow$ DAG & \textcolor{black}{$O(\log^{3}n\log\log n)$} & $O(m\log^{3}n)$\tabularnewline
\hline 
\cite{filtser2025stochastic} & directed $\rightarrow$ DAG & \textcolor{black}{$O(\log n\log\log n)$} & $O(m\log^{3}n)$\tabularnewline
\hline 
\textbf{Ours} (\Cref{thm:main distance }) & directed $\rightarrow$ DAG & $(1+1/\polylog(n))$ & $m^{1+o(1)}$\tabularnewline
\hline 
\end{tabular}

~\\
~\\

\begin{tabular}{|c|c|c|c|}
\hline 
\textbf{Maximum flow} & From $\rightarrow$ To & Approximation & Total size\tabularnewline
\hline 
\hline 
\cite{racke2002minimizing} & \textcolor{brown}{undirected $\rightarrow$ tree} & $O(\log^{3}n)$  & $2n$\tabularnewline
\hline 
\cite{harrelson2003polynomial} & \textcolor{brown}{undirected $\rightarrow$ tree} & $O(\log^{2}n\log\log n)$ & $2n$\tabularnewline
\hline 
\cite{racke2008optimal} & \textcolor{brown}{undirected $\rightarrow$ tree} & $O(\log n)$ & \textcolor{brown}{$O(mn)$}\tabularnewline
\hline 
\cite{racke2014improved} & \textcolor{brown}{undirected $\rightarrow$ tree} & $O(\log n\log\log n)$ & $2n$\tabularnewline
\hline 
\textbf{Ours} (\Cref{thm:main flow}) & directed $\rightarrow$ DAG & $n^{o(1)}$ & $m^{1+o(1)}$\tabularnewline
\hline 
\end{tabular}

}}

\caption{Summary of results on approximating undirected and directed graphs by trees and DAGs, respectively. For distances, we trade $O(m \log^3 n)$ size and $\Omega(\log n)$ approximate for $m^{1+o(1)}$ size and $(1+o(1))$ approximation. For maximum flow, we give the first such DAG.}
\label{tab:history}
\end{table}

\subsection{Our Structural Results}

We make significant progress on both sides. First, for the distance side, we improve the approximation factor of $O(\log n\log\log n)$ \cite{filtser2025stochastic} down to $(1+1/\polylog(n))$ using a DAG with slightly larger $m^{1+o(1)}$ size. Second, for the flow side, we give the first directed analog of tree cut sparsifiers using a DAG of size $m^{1+o(1)}$. 

To describe our DAGs more precisely, we need a notion of \emph{DAG projection.} A \emph{partial projection} to $G$ is a graph $G'$ whose vertices of $G'$ are either copies of vertices in $G$ or Steiner vertices. Formally, there exists a mapping $\pi:V(G')\rightarrow V(G)\cup\{\bot\}$. We say $u'\in\pi^{-1}(u)$ is a copy of $u$ and $u'\in\pi^{-1}(\bot)$ is a Steiner vertex.\footnote{A \emph{projection} has a stricter requirement that $\pi : V(G')\rightarrow V(G)$ is a graph homomorphism. That is, there is no Steiner vertex and if $(u',v')\in E(G')$, then $(u,v)\in E(G)$ where $u=\pi(u')$ and $v=\pi(v')$ (and they have the same edge weight). See \Cref{def:inherited}, which was also used in \cite{negSSSP}. We usually omit the term ``partial'' and distinguish the two concepts only when the difference is crucial.} The \emph{width} of $G'$ is the maximum number of copies, i.e., $\max_{u\in V(G)}|\pi^{-1}(u)|$. If $G'$ is a DAG, then we say $G'$ is a \emph{DAG projection}.

\paragraph{$(1+\epsilon)$-Distance-Preserving DAG.}

Our first result is a DAG projection with almost-linear size that can $(1+1/\polylog(n))$-approximate distances. 
\begin{thm}
\label{thm:main distance }For any directed graph $G$ with edge weights from $\{0,1,2,\dots,\poly(n)\}$ and $\epsilon\ge1/\polylog(n)$, there exists a DAG projection $G'$ to $G$ of size $|E(G')|=m^{1+o(1)}$ and width $n^{o(1)}$  such that, for every $s,t\in V(G)$, 
\[
\dist_{G}(s,t)\le\dist_{G'}(\pi^{-1}(s),\pi^{-1}(t))\le(1+\epsilon)\dist_{G}(s,t).
\]
Moreover, there is a randomized algorithm for computing $G'$ in $m^{1+o(1)}$ time.
\end{thm}

Recall that $\dist_{G'}(\pi^{-1}(s),\pi^{-1}(t))$ is the distance from any copy of $s$ to any copy of $t$ in $G'$. It can be easily computed, for example, by adding to $G'$ zero-weight edges from dummy source $s_{0}$ to $\pi^{-1}(s)$ and from $\pi^{-1}(t)$ to dummy sink $t_{0}$, and then computing the distance from $s_{0}$ to $t_{0}$. 

Compared with \cite{assadi2025covering,filtser2025stochastic}, their DAG cover ${\cal D}$ only guarantees $O(\log n\log\log n)$ approximation albeit with slightly better $O(m\log^{3}n)$ total size. 

\paragraph{$n^{o(1)}$-Congestion-Preserving DAG.}

Our next result is a DAG projection with almost-linear size that can $n^{o(1)}$-approximate maximum flows. This is the first directed analog of tree cut sparsifier in undirected graphs.
\begin{thm}
\label{thm:main flow}For any directed graph $G$ with edge capacity from $\{1,2,\dots,\poly(n)\}$, there exists a DAG projection $G'$ to $G$ of size $|E(G')|=m^{1+o(1)}$ and width $n^{o(1)}$ such that, for every $S,T\subseteq V(G)$, 
\[
\maxflow_{G}(S,T)\le\maxflow_{G'}(\pi^{-1}(S),\pi^{-1}(T))\le n^{o(1)}\maxflow_{G}(S,T).
\]
Moreover, there is a randomized algorithm for computing $G'$ in $m^{1+o(1)}$ time.
\end{thm}

Observe that $\maxflow_{G'}(\pi^{-1}(S),\pi^{-1}(T))$ can be computed, for example, by adding to $G'$ infinite-capacity edges from dummy source $s_{0}$ to $\pi^{-1}(S)$ and from $\pi^{-1}(T)$ to dummy sink $t_{0}$, and then computing the maximum flow from $s_{0}$ to $t_{0}$. 

We emphasize that there are exponentially many $\maxflow_{G}(S,T)$ values that $G'$ preserves. So, it is unclear that $G'$ exists even if we allow $|E(G')|=O(n^{2})$. Note that even though our approximation is $n^{o(1)}$, the best approximation for tree cut sparsifiers in undirected graphs is $\Omega(\log n)$.

\paragraph{Almost-Optimal Size.}

The size of $m^{1+o(1)}$ in both \Cref{thm:main distance ,thm:main flow} are optimal up to $m^{o(1)}$ factor. Indeed, since reachability information of $m$-edge directed bipartite graphs can encode $\Omega(m)$ bits of information, the $\Omega(m)$ size lower bound holds even for arbitrary data structures that can answer reachability only, let alone approximating distances or flow. In contrast, in undirected graphs, tree covers or tree cut sparsifiers may have total size $o(m)$. 

\paragraph{Parallel Construction: Reductions to Approximations on DAGs.}

Not only we can construct the DAG projections from \Cref{thm:main distance ,thm:main flow} in almost-optimal $m^{1+o(1)}$ time in the sequential setting, our constructions are also almost optimal in the parallel settings (i.e., they take $m^{1+o(1)}$ work and $m^{o(1)}$ depth), assuming almost-optimal parallel algorithms for \emph{approximate} shortest paths or maximum flow \emph{on DAGs}. 

To formalize this, we say that there is an \emph{efficient parallel reduction} from problem ${\cal A}$ to problem $\cO$ if, given that $\cO$ can be solved in $w$ work and $d$ depth on a graph with $m$ edges, then $\cA$ can be solved in $w\cdot m^{o(1)}$ work and $d\cdot m^{o(1)}$ depth on a graph with $m$ edges. 
\begin{thm}
[Efficient parallel reductions]\label{thm:parallel construction intro}There are efficient parallel reductions from constructing DAG projections in \Cref{thm:main distance ,thm:main flow} to, respectively, computing $(1+1/\polylog(n))$-approximate single-source shortest path and $n^{o(1)}$-approximate maximum flow on a DAG whose topological order is given.
\end{thm}

\subsection{New Landscape of Parallel Shortest Paths and Maximum Flow}

Via the efficient parallel reduction in \Cref{thm:parallel construction intro}, we reduce major open problems of finding almost-optimal parallel algorithms for \emph{exact} single-source shortest paths (SSSP) and maximum flow to easier settings, leading to a clean landscape of both problems. See \Cref{fig:landscape}.

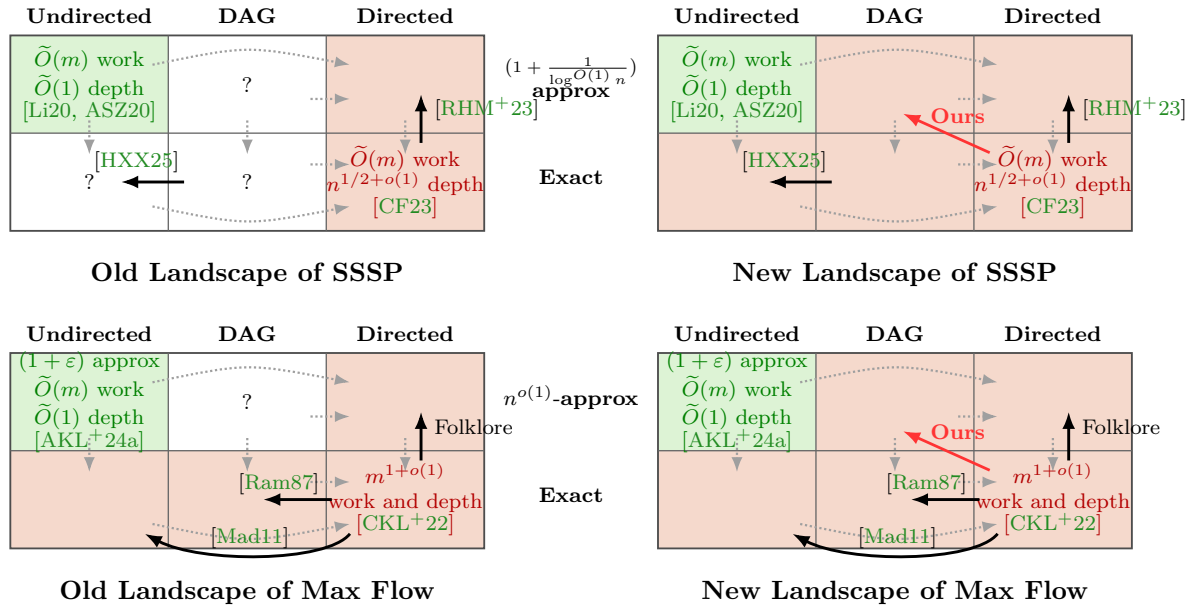
\begin{figure}

	\begin{center}
		\begin{tikzpicture}[x=\CellW cm, y=\CellH cm]
			\path[use as bounding box] (-0.25,-\RowSep-1.2*\TitleDrop) rectangle (6+\GridSep+0.35,2.6);

			\OneGrid{0}{%
				\fill[cellgreen] (0,1) rectangle (1,2);
				\fill[cellpeach] (2,0) rectangle (3,2);
			}{%
				\node[celltext,text=LiText]  at (.5,1.5) {$\widetilde{O}(m)$ work\\ $\widetilde{O}(1)$ depth\\[-1pt] {\scriptsize\cite{li2020faster,andoni2020parallel}}};
				\node[celltext] at (1.5,1.5) {?};
				\node[celltext] at (.5,0.5)  {?};
				\node[celltext] at (1.5,0.5) {?};
				\node[celltext,text=RozText] at (2.5,0.5) {$\widetilde{O}(m)$ work\\ $n^{1/2+o(1)}$ depth\\[-1pt] {\scriptsize\cite{cao2023parallel}}};
			}{%
				\draw[strong] (\HXXRightX,\HXXY) -- (\HXXLeftX,\HXXY);
				\node[font=\scriptsize, anchor=south] at (\HXXLabelX,\HXXLabelY) {\cite{HoppenworthXX25}};
				\draw[strong] (\RozX,\RozYBot) -- (\RozX,\RozYTop);
				\node[font=\scriptsize, anchor=west] at (\RozLabelX,\RozLabelY) {\cite{RozhonHMGZ23}};
			}{Old Landscape of SSSP}
			
			\OneGrid{3+\GridSep}{%
				\fill[cellpeach] (0,0) rectangle (3,2);
				\fill[cellgreen] (0,1) rectangle (1,2);
			}{%
				\node[celltext,text=LiText]  at (.5,1.5) {$\widetilde{O}(m)$ work\\ $\widetilde{O}(1)$ depth\\[-1pt] {\scriptsize\cite{li2020faster,andoni2020parallel}}};
				\node[celltext,text=RozText] at (2.5,0.5) {$\widetilde{O}(m)$ work\\ $n^{1/2+o(1)}$ depth\\[-1pt] {\scriptsize\cite{cao2023parallel}}};
			}{%
				\draw[strong] (\HXXRightX,\HXXY) -- (\HXXLeftX,\HXXY);
				\node[font=\scriptsize, anchor=south] at (\HXXLabelX,\HXXLabelY) {\cite{HoppenworthXX25}};
				\draw[line width=1.05pt, draw=red!80, -{Latex[length=2.3mm]}] (2.10,0.8) -- (1.55,1.20);
				\node[font=\scriptsize\bfseries, text=red!80, anchor=south west] at (1.66,1.00) {Ours};
				\draw[strong] (\RozX,\RozYBot) -- (\RozX,\RozYTop);
				\node[font=\scriptsize, anchor=west] at (\RozLabelX,\RozLabelY) {\cite{RozhonHMGZ23}};
			}{New Landscape of SSSP}
			
			\pgfmathsetmacro{\LabelMidX}{3 + 0.5*\GridSep + \LabelShift}
			\node[midlabel] at (\LabelMidX,1.55)
			{\tiny{$\bigl(1+\tfrac{1}{\log^{O(1)} n}\bigr)$}\\\textbf{approx}};
			\node[midlabel] at (\LabelMidX,0.55) {\textbf{Exact}};
			
			\begin{scope}[shift={(0,-\RowSep)}]
				
				\OneGrid{0}{%
					\fill[cellgreen] (0,1) rectangle (1,2);
					\fill[cellpeach] (0,0) rectangle (3,1);   %
					\fill[cellpeach] (2,1) rectangle (3,2);   %
				}{%
					\node[celltext,text=LiText] at (.5,1.5)
					{$(1+\varepsilon)$ approx\\ $\widetilde{O}(m)$ work\\ $\widetilde{O}(1)$ depth\\[-1pt] {\scriptsize\cite{AgarwalKLPWWZ24}}};
					\node[celltext] at (1.5,1.5) {?};
					\node[celltext,text=RozText] at (2.5,0.5)
					{$m^{1+o(1)}$\\ work and depth\\[-1pt] {\scriptsize\cite{ChenKLPGS22}}};
				}{%
					\draw[strong] (2.6,0.9) -- (2.6,1.4);
					\node[font=\scriptsize, anchor=west] at (2.62,1.25) {Folklore};
					\draw[strong] (\VijXRight,\VijY) -- (\VijXLeft,\VijY);
					\node[font=\scriptsize, anchor=south] at (\VijLabelX,\VijLabelY) {\cite{ramachandran1987complexity}};
					\draw[strong] (\MadryEndX,\MadryEndY)
					.. controls (\MadryCtrlBX,\MadryCtrlBY) and (\MadryCtrlAX,\MadryCtrlAY) ..
					(\MadryStartX,\MadryStartY);
					\node[font=\scriptsize, anchor=south west] at (\MadryLabelX,\MadryLabelY) {\cite{madry2011graphs}};
				}{Old Landscape of Max Flow}
				
				\OneGrid{3+\GridSep}{%
					\fill[cellpeach] (0,0) rectangle (3,2);
					\fill[cellgreen] (0,1) rectangle (1,2);
				}{%
					\node[celltext,text=LiText] at (.5,1.5)
					{$(1+\varepsilon)$ approx\\ $\widetilde{O}(m)$ work\\ $\widetilde{O}(1)$ depth\\[-1pt] {\scriptsize\cite{AgarwalKLPWWZ24}}};
					\node[celltext,text=RozText] at (2.5,0.5)
					{$m^{1+o(1)}$\\ work and depth\\[-1pt] {\scriptsize\cite{ChenKLPGS22}}};
				}{%
					\draw[strong] (2.6,0.9) -- (2.6,1.4);
					\node[font=\scriptsize, anchor=west] at (2.62,1.25) {Folklore};
					\draw[line width=1.05pt, draw=red!80, -{Latex[length=2.3mm]}]
					(2.10,0.80) -- (1.55,1.20);
					\node[font=\scriptsize\bfseries, text=red!80, anchor=south west]
					at (1.66,1.02) {Ours};
					\draw[strong] (\VijXRight,\VijY) -- (\VijXLeft,\VijY);
					\node[font=\scriptsize, anchor=south] at (\VijLabelX,\VijLabelY) {\cite{ramachandran1987complexity}};
					\draw[strong] (\MadryEndX,\MadryEndY)
					.. controls (\MadryCtrlBX,\MadryCtrlBY) and (\MadryCtrlAX,\MadryCtrlAY) ..
					(\MadryStartX,\MadryStartY);
					\node[font=\scriptsize, anchor=south west] at (\MadryLabelX,\MadryLabelY) {\cite{madry2011graphs}};
				}{New Landscape of Max Flow}
				
				\pgfmathsetmacro{\LabelMidXTwo}{3 + 0.5*\GridSep + \LabelShift}
				\node[midlabel] at (\LabelMidXTwo,1.55)
				{$n^{o(1)}$-\textbf{approx}};
				\node[midlabel] at (\LabelMidXTwo,0.55) {\textbf{Exact}};
				
			\end{scope}
		\end{tikzpicture}
	\end{center}

\caption{Old and new landscapes of parallel SSSP and maximum flow algorithms. The green area highlights the settings where near-optimal parallel algorithms are known. The red area highlights the settings as hard as the exact directed setting. The solid arrows represent non-trivial efficient parallel reductions, while the dotted arrows represent trivial ones. \label{fig:landscape}}

\end{figure}

More precisely, using our DAG projections, we obtain efficient parallel reductions
\begin{enumerate}
\item From exact directed SSSP to $(1+1/\polylog(n))$-approximation on DAGs, 
\item From exact directed SSSP to exact \emph{undirected} ones (using \cite{HoppenworthXX25}), and
\item From exact directed maximum flow to $n^{o(1)}$-approximation on DAGs. This reduction 
was not known even in the classical \emph{sequential} setting.\footnote{It  confirms the informal statement in \cite{bernstein2024maximum} about computing maximum flow, saying that  ``the main bottleneck seems to be a fast $n^{o(1)}$-approximation for DAGs''.}
\end{enumerate}
Consider algorithms in the six settings based on the following combinations: (1) exact or approximate, and (2) on directed graphs, DAGs, or undirected graphs. Our reductions categorize these six settings for both SSSP and maximum flow into only two regimes. The \emph{easy} regime consists of only the approximate undirected setting, and the \emph{hard} regime contains the remaining five settings. These five settings are equivalent in the sense they all admit algorithms with the same work and depth up to subpolynomial factor. 

This significantly cleans up the landscape of both problems. With this new landscape, to improve the state-of-the-art of either exact directed SSSP or maximum flow, it suffices to improve either (1) approximation algorithms on DAGs, or (2) exact algorithms on undirected graphs. 

We note that we crucially exploit our approximation guarantees. If they were worse, our reductions between exact and approximate settings above would not work with current techniques. For example, they would fail if the distance-preserving DAG projection only guaranteed $1.01$-approximation.\footnote{It fails because the reduction in \Cref{lem:boostingSSSP} by \cite{RozhonHMGZ23} requires $(1+o(1/\log n))$-approximate oracle.} Thus, the previous results with $\Omega(\log n)$ approximate \cite{assadi2025covering,filtser2025stochastic} indeed do not work. Similarly, if the congestion-preserving DAG projection guaranteed $n^{0.1}$-approximation, then the reduction would have been too inefficient.

\subsection{Transferring Results from DAGs to Directed Graphs}

Lastly, our DAG projections immediately transfer algorithms on DAGs, usually simpler and more efficient, to directed graphs. Below, we list problems that we improve the state-of-the-art or obtain simpler algorithms via DAG projections.

\subsubsection*{Applications of Distance-Preserving DAG Projections.}

\paragraph{Distance preservers: improved.}

Given a directed graph $G=(V,E)$ with edge lengths and a set of demand pairs $P\subseteq V\times V$, a \emph{$(1+\eps)$-approximate distance preserver} ($(1+\eps)$-DP) is a subgraph $H\subseteq G$ such that for every $(s,t)\in P$, we have $\dist_{H}(s,t)\le(1+\eps)\cdot\dist_{G}(s,t)$. A long line of work (e.g.~\cite{CoppersmithE05,bodwin2017linear,HoppenworthXX25}) have been trying to find $(1+\eps)$-DPs of smallest size $|E(H)|$. 

Currently, the best bounds for DAGs are strictly better than the ones for general directed graphs. To simplify the discussion, we focus on the $O(n)$-size regime. There exists a $(1+\eps)$-DP for DAGs of size $O(n+p\sqrt{n})$, which gives $O(n)$ as long as $p=O(\sqrt{n})$ \cite{CoppersmithE05,HoppenworthXX25}.\footnote{This holds even in the exact setting by applying Lemma 4.6 of \cite{HoppenworthXX25} to the exact distance preservers on undirected graphs by \cite{CoppersmithE05}.} However, the best bound for general directed graphs achieves $O(n)$ size only when $p=O(n^{1/3})$ \cite{bodwin2017linear}.\footnote{The bound of $O(n+n^{2/3}p)$ by Bodwin \cite{bodwin2017linear} works even in the exact setting.} Our distance-preserving DAG projection \Cref{thm:main distance } closes this gap up to $n^{o(1)}$ factor.
\begin{cor}
\label{cor:distancepreserver} For $\eps\ge1/\polylog(n)$, there exists a $(1+\eps)$-approximate distance preserver for any $n$-node directed graph with polynomial edge lengths and $p$ demand pairs of size $(n+p\sqrt{n})\cdot n^{o(1)}$, which is $n^{1+o(1)}$ for $p=O(\sqrt{n})$.
\end{cor}

\Cref{cor:distancepreserver} exploits that the width in \Cref{thm:main distance } is $n^{o(1)}$ so that we do not blow up the number of vertices, and also that $G'$ is actually a \emph{projection} (see \Cref{def:inherited}), not only a partial projection.

\paragraph{Hop-set: simplified.}

Given a directed graph $G=(V,E)$ with edge lengths, a \emph{$(\beta,\eps)$-hop-set} of $G$ is a set of additional edges $H\subseteq V\times V$ added to the graph so that for every $s,t\in V$ we have (i) $\dist_{G}(s,t)\le\dist_{G\cup H}(s,t)$, and (2) $\dist_{G\cup H}^{(\beta)}(s,t)\le(1+\eps)\cdot\dist_{G}(s,t)$ where $\dist_{G\cup H}^{(\beta)}(s,t)$ denotes the lengths of the shortest path from $s$ to $t$ \emph{using at most $\beta$ edges}.

The breakthrough result of Kogan and Parter \cite{kogan2022new} showed the existence of linear-sized $(O(n^{1/3}),\eps)$-hop-sets for DAGs, but only linear-sized $(O(n^{2/5}),\eps)$-hop-sets for general graphs. With significant effort, Bernstein and Wein \cite{bernstein2023closing} closed this gap by showing linear-sized $(O(n^{1/3}),\eps)$-hop-sets for any graph.

\Cref{thm:main distance } can bypass this significant effort and close the gap up to $n^{o(1)}$ factor in a black-box way. Indeed, by applying the construction of \cite{kogan2022new} for DAGs on top of our distance-preserving DAG projection, we immediately obtain $n^{1+o(1)}$-sized $(O(n^{1/3}),\eps)$-hop-sets.

\paragraph{Potential applications.}

There are gaps between DAGs and general directed graphs in several variants of \emph{min-distance} problems \cite{dalirrooyfard2021approximation,chechik2022constant}. Also, the recent approximate restricted SSSP algorithm \cite{ashvinkumar2025faster} is much simpler on DAGs. Thus, DAG projections could potentially help close these gaps or simplify the algorithms.

\subsubsection*{Applications of Congestion-Preserving DAG Projections.}

\paragraph{Combinatorial Max flow: simplified.}

The  combinatorial maximum flow algorithms by \cite{bernstein2024maximum,BBLST25} runs in $\tilde{O}(n^{2})$ time, near-optimal on dense graphs. They first showed a very simple push-relabel algorithm that gives $O(1)$-approximation on DAGs and then generalized to general graphs via expander hierarchies in a white-box way. Our reduction from exact max flow to $n^{o(1)}$-approximation on DAGs (via \Cref{thm:main flow}), directly generalizes their algorithms for DAGs to general graphs in a black-box way with an additional $n^{o(1)}$ factor in time.

\paragraph{Bounded Single-Source Max flow: improved. }

Let $G$ be a graph with positive integral edge capacities and source $s\in V(G)$. Let $\maxflow_{G}^{k}(s,t)=\min\{k,\maxflow_{G}(s,t)\}$ denote the \emph{$k$-bounded maximum flow} value from $s$ to $t$. If $G$ is a DAG, the algorithm based on network coding by \cite{cheung2013graph} can compute $\maxflow_{G}^{k}(s,t)$ for all $t\in V(G)$ in $O(k^{\omega-1}m)$ total time, while the best known algorithm on general graphs is to trivially compute $\maxflow_{G}^{k}(s,t)$ for every $t$ using $\Omega(mn)$ time. 

Using our congestion-preserving DAG projection, we obtain a non-trivial algorithm that beats the $mn$ bound when $k\in(n^{\Omega(1)},n^{1/\omega-\Omega(1)})$.
\begin{thm}
There is a randomized algorithm that, given a directed graph with positive integral edge capacities, a source $s$, and an integer $k$, computes $n^{o(1)}$-approximation of $\maxflow_{G}^{k}(s,t)$ for all vertices $t$ in $k^{\omega}m^{1+o(1)}$ time.
\end{thm}

\paragraph{Potential application.}

Vertex cut sparsifiers for $k$ terminals with optimal $\Theta(k^{2})$ size are only known on DAGs \cite{he2021near}, but the best construction on general graphs still requires $O(k^{3})$ size \cite{kratsch2012representative}. Can we use congestion-preserving DAG projection help bridging this gap, even with approximation?

\subsection{Related Work}
There has also been a line of work on preserving distances in \emph{undirected} graphs by tree-like graphs with multiple copies of the original vertices.
This includes the \emph{multi-embeddings} of Bartal and Mendel \cite{bartal2003multi}, the \emph{clan embeddings} of Filtser and Le \cite{filtser2021clan}, and \emph{tree embeddings with copies} \cite{haepler2022adaptive}. 
However, these works incur at least logarithmic distortion and are tailored to undirected graphs.

In contrast, our distance DAG projections apply to directed graphs and achieve $(1+\epsilon)$-approximation with almost-linear size. 
Moreover, for our congestion DAG projection, we are not aware of prior copy-based constructions aimed at preserving congestion or flow.

\subsection{Organization}

In \Cref{sec:overview}, we will sketch the \emph{existence} of our distance DAG projections and congestion DAG projections, without considering the time complexity. In \Cref{sec:prelim}, we give a full definition of all the basic notations used in the paper. In \Cref{sec:DAGprojections}, we give a formal definition of all DAG projection-related concepts, and show their applications assuming efficient construction. In \Cref{sec:distanceembedding}, we will show an efficient algorithm for constructing distance DAG projections assuming SSSP oracles on DAGs. In \Cref{sec:congestionDAGembedding}, we will show an efficient algorithm for constructing congestion DAG projections assuming max flow oracles on DAGs.

\section{Overview}\label{sec:overview}

In this section, we sketch the existence of our distance and congestion DAG projections, ignoring running time and focusing only on the idea and existence.

\subsection{Distance DAG Projections}

Our construction is inspired by the  $(n^{o(1)},\eps)$-hop-sets of \cite{Cohen00} in \emph{undirected} graphs.

Given a directed graph $G=(V,E)$ with edge lengths $\ell_G:E\to \bbN_{\ge 1}$, the goal is to construct a $(1+\eps)$-distance-preserving DAG projection $D$ onto $G$.
For convenience, when we say that $D$ \emph{preserves the $(s,t)$-distance} for $s,t\in V(G)$, we mean: there exist $s',t'\in V(D)$ with $\pi(s')=s$ and $\pi(t')=t$ such that
\[
\dist_G(s,t)\le\dist_D(\pi^{-1}(s),\pi^{-1}(t))\;\le\;(1+\eps)\cdot \dist_G(s,t).
\]
We will use the following directed low-diameter decomposition (LDD) as a subroutine.

\begin{lemma}[\cite{BringmannCF23,BernsteinNW25}]\label{lem:LDDexistence}
Let $G=(V,E)$ be a directed graph with edge lengths and let $d$ be a positive integer. There exists a random set of edges $\Erem$ (a \emph{low-diameter decomposition}) satisfying:
\begin{itemize}
    \item each SCC of $G-\Erem$ has diameter at most $d$,
    \item for every $e\in E$, $\Pr[e\in \Erem]\le \SL{}\cdot \ell_G(e)/d$,
\end{itemize}
where $\SL{}=\tO{1}$.
\end{lemma}

Now, the construction is as follows.

\paragraph{Step 1 (LDD).}
Compute LDDs of $G$ with diameters $2^i$ for all $0\le i\le O(\log n)$.
Fix one such LDD for a given $i$ with parameter $d=2^i$, yielding an edge set $\Erem$ as in \Cref{lem:LDDexistence}. We show how to build a DAG projection that preserves $(s,t)$-distances in $G$ whenever $\dist_G(s,t)$ is on the order of $d/\eps$. The final projection $D$ will be the union of the $O(\log n)$ projections over all $i$, thereby handling all pairs (we assume edge lengths are polynomial on $n$).

Let $\cC$ be the family of SCCs of $G-\Erem$. By construction, each $C\in\cC$ has $\operatorname{diam}(G[C])\le d$.
Let $\sigma=n^{o(1)}$ (to be fixed later), and classify clusters by size:
\begin{itemize}
    \item \emph{large} if $|C|\ge n/\sigma$,
    \item \emph{small} if $|C|<n/\sigma$.
\end{itemize}

\paragraph{Step 2 (recursive construction for small clusters).}
For each small $C\in\cC$, recursively construct a DAG projection $D_C$ of $G[C]$ that $(1+\eps)$-approximately preserves distances for \emph{all pairs} of vertices in $G[C]$ (instead of just length bounded pairs).

\paragraph{Step 3 (shortest-path trees for large clusters).}
For each large $C\in\cC$, pick an arbitrary root $r_C\in C$. In $G$, compute a shortest-path tree $\To_C$ rooted at $r_C$ and a reversed shortest-path tree $\Ti_C$ rooted at $r_C$ (i.e., every $(s,r_C)$-path in $\Ti_C$ is a shortest $(s,r_C)$-path). Define $D_C$ to be the DAG projection that combines $\To_C$ and $\Ti_C$ (as \emph{disjoint copies} of subgraphs of $G$) by merging the common root $r_C$.

\paragraph{Step 4 (combining everything).}
Let $\cC=(C_1,\ldots,C_z)$ be a topological ordering of the SCCs in $G-\Erem$. Form $D'$ by concatenating the $D_{C_i}$ using copies of edges in $G$: for every $1\le i<j\le z$, every $u'\in V(D_{C_i})$, and every $v'\in V(D_{C_j})$, if $(\pi(u'),\pi(v'))\in E(G)$ then add the edge $(u',v')$ to $D'$.

Finally, let $x=\SL{}/\eps$. Create $x$ copies $D'_1,\ldots,D'_x$ of $D'$ and concatenate them in the same manner: for every $1\le i<j\le x$, $u'\in V(D'_i)$, and $v'\in V(D'_j)$, if $(\pi(u'),\pi(v'))\in E(G)$ then add $(u',v')$. The resulting DAG is $D$.

\paragraph{Distance preservation.}

It is easy to see that
\[
    \dist_G(s,t)\le \dist_{D}\bigl(\pi^{-1}(s),\pi^{-1}(t)\bigr)
\]
since $\pi$ is a (weight-preserving) graph homomorphism (the only edges we added to $D$ are in Step 3 and 4, which are copies of edges in the original graph): Every path in $D$ connecting a vertex in $\pi^{-1}(s)$ to a vertex in $\pi^{-1}(t)$ can be mapped to a path in $G$ connecting $s$ to $t$ with the same length. Hence the distance in $G$ is at most the distance in $D$.

In what remains, we will only prove the harder direction
\[
    \dist_{D}\bigl(\pi^{-1}(s),\pi^{-1}(t)\bigr)\le (1+\eps)\cdot \dist_G(s,t).
\]

Fix $s,t\in V(G)$ with $\dist_G(s,t)=d/\eps$. Let $p$ be an $s{\to}t$ path in $G$ of length $\ell_G(p)=d/\eps$. By \Cref{lem:LDDexistence}, in expectation there are $\SL{}/\eps$ edges of $p$ lie in $\Erem$. Hence we can partition $p$ into subpaths $p_1,p_2,\ldots$, each entirely contained in $G-\Erem$. Notice that edges connecting $p_i$ to $p_{i+1}$ is preserved in $D$ according to Step 4.%

Let $a,b\in V$ be the endpoints of some subpath $p_i$. Let $p'$ be a shortest $a{\to}b$ path in $G-\Erem$. Decompose $p'$ into maximal subpaths $p'_1,\ldots,p'_{y'}$ so that each $p'_j$ lies within a single SCC of $G-\Erem$. Notice that edges connecting $p'_i$ to $p'_{i+1}$ is preserved in $D'$ according to Step 4.
If that SCC is a small cluster, then the endpoints’ distance is preserved by the recursive construction. Otherwise, it is a large cluster.

Let $C_L$ be the first large cluster meeting $p$ (\emph{not $p'$}, we switch our attention to the whole path $p$) and, similarly, $C_R$ be the last. Let $\tilde{p}$ be the subpath of $p$ from its first vertex in $C_L$ to its last vertex in $C_R$. Denote these endpoints by $\Start(\tilde{p})$ and $\End(\tilde{p})$. By the definition of $D_C$ for a large cluster $C$, $D_{C}$ contains $\Ti_C$ followed by $\To_C$, both rooted at $r_C$. Thus a copy of $\Start(\tilde{p})$ in $\Ti_{C_L}$ reaches $r_{C_L}$ within distance at most $d$ (since $G[C_L]$ has diameter $d$ and $\Ti_{C_L}$ is a shortest path tree), and from $r_{C_L}$ we can reach a copy of $\End(\tilde{p})$ within distance $\ell_G(\tilde{p})$ plus at most $d$ additional length via $\To_{C_R}$ (this is because the original distance from $c_L$ to $\End(\tilde{p})$ is at most $d+\ell_G(\tilde{p})$: from $c_L$ we can reach $\Start(\tilde{p})$ whithin distance $d$, the we follow $\tilde{p}$ to reach $\End(\tilde{p})$). Hence, the total distance is at most
\[
\ell_G(\tilde{p})+2d,
\]
which has an additive $O(d)$ overhead. Since $\dist_G(s,t)=d/\eps$, this additive term is negligible and absorbed in the $(1+\Theta(\eps))$ multiplicative guarantee.  

\paragraph{Size of the projection.}
If we write $f(n)$ as the width (maximum number of copies) of the DAG projection returned by the above construction on an $n$-vertex graph, then $f(n)$ can be calculated as follows.

(1) In Step 1 and Step 4, we know that $D$ contains $x = \tO{1/\eps}$ copies of $D'$.

(2) In Step 2, each vertex is copied $f(n/\sigma)$ times in $D'$.

(3) In Step 3, each vertex is copied $2 \cdot \sigma^2$ times in $D'$, where the factor $2$ comes from the forward and reversed shortest path trees, and the factor $\sigma^2$ comes from the fact that there are at most $\sigma$ large clusters (because each large cluster has at least $n/\sigma$ vertices).

To summarize, we get the following recursive inequality:
\[
    f(n) \le \tO{1/\eps} \cdot \bigl(f(n/\sigma) + 2\sigma^2\bigr).
\]

Notice that $\eps = 1/\polylog(n)$ according to our assumption. By choosing an appropriate parameter $\sigma = n^{o(1)}$, the desired bound $f(n) = n^{o(1)}$ follows.

\paragraph{Algorithmic Aspect.} 
The algorithm above runs in $m^{1+o(1)}$ time; it takes near-linear time in the size of the output, which is $m^{1+o(1)}$, because we simply call LDD which can be computed using SSSP as a subroutine. 

In the parallel setting, our goal is to reduce DAG projection construction to approximate SSSP on DAGs.
But, currently, the algorithm needs to call as a subroutine the  
SSSP algorithm on general graphs, which does not admit efficient parallel algorithms yet. 
Our strategy is to reduce SSSP on general graphs to SSSP on DAGs using our distance DAG projection itself. 
This leads to the `chicken-and-egg' situation: (When we are given an SSSP on DAGs oracle), we can reduce distance DAG projection to SSSP on general graphs, and we can reduce SSSP on general graphs to distance DAG projection.

Our strategy (detailed in \Cref{sec:distanceembedding}) is to find a spiral recursion: We define $h$-length DAG projection that only preserves distances up to $h$, and we reduce it to $h/n^{o(1)}$-length SSSP on general graphs which computes distance up to $h/n^{o(1)}$, reducing to $h/n^{o(1)}$-length DAG projection. This spiral recursion will finally make $h$ small enough so the construction is trivial.

\subsection{Congestion DAG Projections}\label{subsec:congestion-dag-projections}
Given a graph $G=(V,E)$ with edge capacities denoted by $U_G:E\to \bbR_{\ge 0}$, we will show how to construct a $n^{o(1)}$-congestion-preserving DAG projection.

We consider \emph{expander decomposition} as the analogue of LDD in the congestion setting. We first give some necessary definitions.

For a flow $\ff$, we represent the demand routed by $\ff$ as a pair of vectors $(\DDelta,\nnabla)$, where $\DDelta,\nnabla:V\to\bbR_{\ge 0}$ denote the source demands and sink demands routed by $\ff$, respectively; see \Cref{subsec:flowandcuts} for the formal definition.

We say a demand $(\DDelta,\nnabla)$ is \emph{$\dd$-respecting} for some function $\dd:V\to\bbR_{\ge 0}$ if $\DDelta(v),\nnabla(v)\le \dd(v)$ for every $v\in V$. For an edge set $E'\subseteq E$, define
\[
\vol_{E'}(v)\;=\;\sum\{\,U_G(u,v)\mid (u,v)\in E'\text{ or }(v,u)\in E'\,\},
\]
and for a vertex set $C\subseteq V$ write $\vol_{E'}\!\mid_{C}$ for the restriction of function $\vol_{E'}$ to $C$.

\begin{definition}[Terminal Expanding]\label{def:overviewterminalexpanding}
    Given a directed graph $G=(V,E)$ with edge capacities and a function $\dd:V\to\bbR_{\ge 0}$, we say $\dd$ is \emph{$\phi$-expanding} on $G$ if every $\dd$-respecting demand is routable via a flow $\ff$ in $G$ with congestion at most $1/\phi$.
\end{definition}

We next formalize the hierarchical structure we will use.

\begin{definition}[Expander  Hierarchy]\label{def:overviewexpanderdecompositionhierarchy}
    Given a directed graph $G=(V,E)$ with edge capacities, an \emph{$\phi$-expander hierarchy} with $t$ layers consists of edge sets $E_1,\ldots,E_t$ with $E_1=E$ and $E_t=\emptyset$. Let $E_{>i}=\bigcup_{t\ge j>i} E_j$. For every $1\le i<t$ and every strongly connected component (SCC) $C$ of $G-E_{>i}$, the function $\vol_{E_i}\!\mid_C$ is $\phi$-expanding in $G[C]$.
\end{definition}

By \cite{bernstein2024maximum}, such a hierarchy exists with expansion $\phi=2^{-O(\sqrt{\log n})}$ and $t=\sqrt{\log n}$. In our actual efficient construction, we will instead use a weaker hierarchy (\Cref{def:expanderdecompositionhierarchy}) that admit fast construction.

\paragraph{Constructing the congestion DAG projection $D$.}
Let $E_1,\ldots,E_t$ be the expander hierarchy of $G$ with expansion $\phi=2^{-O(\sqrt{\log n})}$ and $t=\sqrt{\log n}$. The construction proceeds recursively from top to bottom using the hierarchy. We will only describe the top level.

Fix an SCC $C$ of $G$. By definition, $\vol_{E_{t-1}}\!\mid_C$ is $\phi$-expanding in $G[C]$. Let $D_C$ be a congestion DAG embedding of $G[C]-E_{t-1}$ obtained \emph{recursively}. We build a congestion DAG embedding $D'_C$ of $G[C]$ as follows:
\begin{itemize}
    \item Create two disjoint copies $D_C^{(1)}$ and $D_C^{(2)}$ of $D_C$, and add a dummy vertex $w_C$.
    \item For every $v\in V(D_C^{(1)})$, add an edge $(v,w_C)$ with capacity $\vol_{E_{t-1}}(\pi(v))$.
    \item For every $v\in V(D_C^{(2)})$, add an edge $(w_C,v)$ with capacity $\vol_{E_{t-1}}(\pi(v))$.
\end{itemize}
Here $\pi$ denotes the projection map from the vertices of $D_C$ to  $G$.

Let $C_1,\ldots,C_z$ be the SCCs of $G$ in a topological order. We obtain $D$ by concatenating the graphs $D'_{C_1},\ldots,D'_{C_z}$: for every $1\le i<j\le z$ and for every $u\in V(D'_{C_i})$, $v\in V(D'_{C_j})$, add an edge $(u,v)$ in $D$ whenever $(\pi(u),\pi(v))\in E(G)$.

\paragraph{Size of $D$.}
There are $t=O(\sqrt{\log n})$ layers, and in each layer we make two copies of every vertex in the next layer. Hence, every vertex of $G$ (and every dummy vertex $w_C$) is copied at most $2^t=n^{o(1)}$ times overall. For edges, we only add edges $(u,v)$ in $D$ if $(\pi(u),\pi(v))\in E$ or one of $u,v$ is a dummy vertex. Since the number of vertex copies is $n^{o(1)}$, it follows that $|E(D)|\le m^{1+o(1)}$.

\paragraph{Correctness.}

We first show that $\maxflow_{G}(S,T)\le\maxflow_{G'}(\pi^{-1}(S),\pi^{-1}(T))$.
Given a flow $\ff$ in $G$ from $S$ to $T$, we show how to route $\ff$ in $D$ from $\pi^{-1}(S)$ to $\pi^{-1}(T)$ with the same value. Because we concatenate $D'_{C_1},\ldots,D'_{C_z}$ using all edges of $G$ respecting the topological order, it suffices to show that the restriction $\ff_C$ of $\ff$ to a fixed SCC $C$ can be routed in $D'_C$. Consider a flow path $p$ of $\ff_C$; we decompose it into three segments:
\begin{enumerate}
    \item The prefix of $p$ before it uses any edge of $E_{t-1}$. This segment is routed inside $D_C^{(1)}$ by the recursive embedding.
    \item The subsequence of $p$ from its first edge $(u_1,v_1)\in E_{t-1}$ to its last edge $(u_2,v_2)\in E_{t-1}$. We replace this by the two-hop path $(u'_1,w_C,v'_2)$ where $\pi(u'_1)=u_1$ and $\pi(v'_2)=v_2$. Feasibility holds because the capacity of $(u'_1,w_C)$ is $\vol_{E_{t-1}}(u_1)$, and analogously for $(w_C,v'_2)$ so no congestion larger than $1$ is introduced when considering all flow paths of $\ff$ (a feasible flow in $G$) passing through $E_{t-1}$.
    \item The suffix of $p$ after its last edge in $E_{t-1}$, which is routed inside $D_C^{(2)}$ by recursion.
\end{enumerate}

Then we show that $\maxflow_{G'}(\pi^{-1}(S),\pi^{-1}(T))\le n^{o(1)}\maxflow_{G}(S,T)$.
Given a flow $\ff^D$ in $D$, we map it back to $G$ with congestion at most $n^{o(1)}$. For all portions of $\ff^D$ that do not traverse dummy vertices, we use the projection $\pi$; since each original vertex has at most $n^{o(1)}$ copies, this incurs at most an $n^{o(1)}$ congestion blow-up. 

Next, consider every subpath of flow paths of $\ff^D$ that are incident to the dummy vertex $w_C$. 
This induces a demand that is $n^{o(1)}\cdot \vol_{E_{t-1}}$-respecting in $G[C]$ (the $n^{o(1)}$ factor comes from the number of vertex copies), hence is routable in $G[C]$ with congestion at most $1/\phi = 2^{O(\sqrt{\log n})}= n^{o(1)}$. Although each $w_C$ may itself be replicated $n^{o(1)}$ times across lower layers, the resulting multiplicative factors still yield total congestion $n^{o(1)}$ overall.

\paragraph{Algorithmic Aspect.} 
Excluding the time for computing the expander hierarchy in \Cref{def:overviewexpanderdecompositionhierarchy}, the construction of congestion DAG projection take near-linear time in the output-size, which is  $m^{1+o(1)}$ time.
To compute the expander hierarchy efficiently, we employ a weaker version as in stated \Cref{def:expanderdecompositionhierarchy} which admit simpler construction by performing expander decomposition in a bottom-up manner. 

In the sequential setting, by using almost-linear time maximum flow algorithm to compute expander decomposition, this gives $m^{1+o(1)}$ total time. 

Unfortunately, to obtain efficient parallel reduction, this suffers from the similar bottleneck as in the distance DAG projection algorithm: Expander decomposition usually requires maximum flow on a general graph as a subroutine (which does not admit fast parallel algorithms), and we intend to reduce max flow on general graphs to max flow on DAGs. 

The idea is again to design a spiral recursion (detailed in \Cref{sec:congestionDAGembedding}): (Given a max flow on DAGs oracle) We define $\delta$ additive-error congestion DAG projection to tolerate some $\delta$ additive error, and reduce it to $\delta\cdot n^{o(1)}$ additive-error maximum flow on general graphs, which reduces to $\delta\cdot n^{o(1)}$ additive-error congestion DAG projection. The additive error will be finally large enough so the problem is trivial.
\section{Preliminaries}\label{sec:prelim}
We use the following notation. For $a \in \mathbb{N}$, let $[a] = \{1,2,\ldots,a\}$. We write $\tO{f} = f \cdot \polylog n$ and $\hO{f} = f \cdot n^{o(1)}$. Unless stated otherwise, we use $n,m$ to denote the numbers of nodes and edges of the input graph, respectively. A \emph{randomized algorithm} succeeds with high probability (w.h.p.), meaning that it fails to output the correct answer with probability at most $1/n^c$ for some constant $c$.

We use $[E]$ to denote the indicator variable of the event $E$, i.e., $[E]=1$ if $E$ happens and $[E]=0$ otherwise. We overload $[\cdot ]$ for both ranges and indicators; context will be clear. For functions with the same domain (e.g. $d_1,d_2 : V \to \bbR$), we write $d_1 \le d_2$ to denote $d_1(v) \le d_2(v)$ for every $v \in V$. Other relations $d_1 = d_2$, $d_1 < d_2$, etc. are defined similarly. For a function $d : V \to \bbR$ and $a \in \bbR$, we write $d+a$ to denote the function $(d+a)(v) = d(v) + a$ for every $v \in V$. For $S \subseteq V$, we write $d(S) = \sum_{v \in S} d(v)$. We write $d\mid_S$ to denote the function $d$ restricted to $S$.

An algorithm $\cA$ with input size $t$ in this paper may make oracle calls to another algorithm $\cA'$. We say the oracle calls are \emph{$f$-work-efficient} if the sum of all input sizes to $\cA'$ over all oracle calls is at most $f \cdot t$, and \emph{$f$-depth-efficient} if the longest dependency chain among all calls has length at most $f$. 

\paragraph{Graph.}
All graphs in this paper, unless specified otherwise, are directed graphs $G = (V,E)$ possibly with edge lengths and capacities. We assume edge lengths and capacities are positive integers bounded by $W$, where $W$ is a polynomial in $n$, for simplicity. We treat edge lengths and capacities as attributes of edges $e \in E$, so we use $\ell(e)$ and $U(e)$ to denote the length and capacity of $e$, respectively, instead of writing $G = (V,E,\ell,U)$. We write $\ell_G(e)$ and $U_G(e)$ if we want to stress that the length and capacity are associated with $G$. For an edge set $E' \subseteq E$ and vertex sets $S,T \subseteq V$, we write
\[
E'(S,T) = \{ (s,t) \mid s \in S,\, t \in T,\, (s,t) \in E' \}.
\]
We use standard definitions of graph terminology for directed graphs, including path, distance, strongly connected components (SCCs), and shortest-path tree, which can be found in textbooks. For a path $H$, we use $\ell(H)$ and $U(H)$ to denote the sum of the lengths or capacities, respectively, of the edges on the path, and we use $|H|$ to denote the number of edges (edges can repeat, in which case they are counted multiple times). We use
\[
\delta^+(S) = \{ (u,v) \mid u \in S,\, v \in V \setminus S \}
\quad\text{and}\quad
\delta^-(S) = \{ (u,v) \mid v \in S,\, u \in V \setminus S \}
\]
to denote the outgoing and incoming edge sets of $S$, respectively. We use $\Start(p)$ and $\End(p)$ to denote the starting vertex and the ending vertex of a path $p$.

\paragraph{Projection map.}
A graph $G=(V,E)$ can be assigned a projection map $\pi: V \to U$ or $\pi: V \to U \cup \{\bot\}$ (where $U$ can be the vertex set of another graph). We write $\pi_G$ to stress that $\pi_G$ is the labeling function for $G$. We write $\pi^{-1} : U \to 2^V$ for the inverse mapping, where $\pi^{-1}(u) = \{ v \in V \mid \pi(v) = u \}$. For a set $U' \subseteq U$, we write $\pi^{-1}(U') = \bigcup_{u \in U'} \pi^{-1}(u)$.

For a set $S \subseteq V$, we write
\[
\pi(S) = \{ \pi(s) \mid s \in S \} \setminus \{ \bot \},
\]
and we write
\[
\HL(S) = \{ u \in U \mid \pi^{-1}(u) \subseteq S \}.
\]
In words, $\pi(S)$ contains vertices $u$ of $G$ for which some copy of $u$ is in $S$, and $\HL(S)$ contains vertices $u$ of $G$ for which all copies of $u$ are in $S$.

\paragraph{Remark on DAG as an input.}
A DAG refers to a \emph{directed acyclic graph}. In this paper, when we say a DAG is an input, we assume its vertex set is given in a topological order.

\subsection{Distances and Shortest Paths}

\paragraph{Graph distance.}
We use $\dist_G(s,t)$ to denote the distance from $s$ to $t$ in $G$ (with respect to the length function $\ell_G$). When $G$ is clear from the context, we simply write $\dist(s,t)$.

Given a directed graph $G = (V,E)$ and a vertex set $A \subseteq V$, the \emph{weak diameter} of $A$ (in $G$) is defined as
\[
\max_{s,t \in A} \dist_G(s,t).
\]
In contrast, the \emph{strong diameter} of $A$ is defined as
\[
\max_{s,t \in A} \dist_{G[A]}(s,t),
\]
where $G[A]$ is the subgraph of $G$ induced by $A$. The strong diameter of $A$ is also the diameter of $G[A]$.

\paragraph{Single-source shortest path (SSSP).}
The single-source shortest path problem is: given a graph with edge lengths and a source $s$, find a shortest-path tree rooted at $s$, i.e., every $(s,t)$-path in the tree has length $\dist_G(s,t)$ for every $t \in V$.

The problem of $\alpha$-approximate SSSP (denoted by $\apxSSSP{\alpha}$) only requires outputting an approximate shortest-path tree such that the length of any $(s,t)$-path (denoted by $\tilde{d}(s,t)$) satisfies
\[
\dist_G(s,t) \le \tilde{d}(s,t) \le \alpha \cdot \dist_G(s,t).
\]

The \emph{$h$-length} $\apxSSSP{\alpha}$ problem only needs to satisfy
\begin{enumerate}
    \item $\dist_G(s,t) \le \tilde{d}(s,t)$ for every $t \in V$, and
    \item $\tilde{d}(s,t) \le \alpha \cdot \dist_G(s,t)$ only for those $t \in V$ with $\dist_G(s,t) \le h$.
\end{enumerate}
We use $h$-length SSSP to denote the problem of $h$-length $\apxSSSP{1}$.

We use $\apxDSSSP{\alpha}$ to denote the problem of $\alpha$-approximate SSSP when the input graphs are restricted to DAGs. We assume that a topological order of the input DAG is given as part of the input.

\subsection{Flows and Cuts}\label{subsec:flowandcuts}

\paragraph{Cuts.}
Given a directed graph $G = (V,E)$, a cut is a partition of the vertex set into two sets $(S, \bar{S})$ where $\bar{S} := V \setminus S$ (when $V$ is clear from the context). For convenience, we also use $S$ to denote the cut $(S,\bar{S})$ when there is no ambiguity. The value of a cut $(S,\bar{S})$ is defined as
\[
\Val_G(S) := U_G(\delta^+(S)) := \sum_{e \in \delta^+(S)} U_G(e),
\]
i.e., the total capacity of edges going from $S$ to $\bar{S}$. We allow a cut to be $\emptyset$ or $V$, in which case the value is defined to be $0$.

\paragraph{Flows.}
Given a directed graph $G = (V,E)$, a (single-commodity) flow $\ff$ is represented as a collection of paths in $G$ associated with positive flow values. We call each path $p \in \ff$ a \emph{flow path}, and we denote its value by $\ff(p) \in \bbR_{>0}$. If $p \notin \ff$, we write $\ff(p) = 0$. For an edge $e \in E$, we define
\[
\ff(e) = \sum_{p : e \in p} \ff(p).
\]
(That is, the flow on an edge is the sum of the flows of all paths that use that edge.) By default, we assume a flow is represented by flows on each edge, so every flow can be represented in $O(m)$ space.

The \emph{congestion} of a flow $\ff$ in $G$ is
\[
\Cong(\ff) = \max_{e \in E} \frac{\ff(e)}{U_G(e)}.
\]
If $\Cong(\ff) \le 1$, we say $\ff$ is \emph{feasible}. Let
\[
\Out{\ff}(u) = \sum_{(u,v) \in E} \ff(u,v)
\qquad\text{and}\qquad
\In{\ff}(u) = \sum_{(v,u) \in E} \ff(v,u)
\]
be the outgoing and incoming flow at a vertex $u \in V$, and let
\[
\ff(u) = \Out{\ff}(u) - \In{\ff}(u)
\]
be the net outgoing flow at $u$.

The following lemma shows it is possible to convert from flow path decomposition and edge representation. It is implied by Section 8 of \cite{agarwal2024parallel}. In their paper, they only proved the flow decomposition. The reverse direction that recovers a flow from a set of sink-source demands can be done by a reverse traversal of their algorithm. 

\begin{theorem}[Section 8 of \cite{agarwal2024parallel}]\label{thm:flowpathdecomposition}
Let $G$ be a directed graph with capacities, and let $F$ be a circulation-free feasible flow on $G$.
Then there is a parallel algorithm that computes a representation
\[
\{(\lambda_i,s_i,t_i)\}_{i\in[k]}
\]
such that $F$ can be decomposed into $k$ directed flow paths, where path $i$ has value $\lambda_i$ and goes from $s_i$ to $t_i$.

Moreover, given any values $(\lambda'_i)_{i\in[k]}$ satisfying $\lambda'_i\le \lambda_i$ for every $i\in[k]$, there is a parallel algorithm that computes an edge representation of a feasible flow routing $\lambda'_i$ units from $s_i$ to $t_i$ for every $i\in[k]$.

Both algorithms run in near-linear work and polylogarithmic depth.
\end{theorem}

\paragraph{Demands.}
A \emph{(single-commodity) demand} is a pair $(\DDelta,\nnabla)$ where $\DDelta, \nnabla : V \to \bbR_{\ge 0}$. Given a flow $\ff$, the demand routed by the flow, denoted by $\Dem(\ff) = (\DDelta_{\ff}, \nnabla_{\ff})$, is defined as
\[
\DDelta_{\ff}(u) = \max(\ff(u), 0)
\qquad\text{and}\qquad
\nnabla_{\ff}(u) = -\min(\ff(u), 0).
\]
The \emph{value} of a demand $(\DDelta,\nnabla)$ is defined as
\[
\Val(\DDelta,\nnabla) = \sum_{v \in V} \DDelta(v).
\]
The value of a flow is $\Val(\ff) = \Val(\Dem(\ff))$. We say $\ff$ is an $(s,t)$-flow with value $x$ if $\ff$ routes the demand
\[
\DDelta(v) = x \cdot [v=s],
\qquad
\nnabla(v) = x \cdot [v=t].
\]
Similarly, we say $\ff$ is an $(S,T)$-flow with value $x$ for $S,T\subseteq V$ if $\ff$ routes a demand with source and sink demands being subsets of $S,T$.

A \emph{sub-demand} of $(\DDelta,\nnabla)$, denoted by $(\DDelta',\nnabla')$, satisfies $\DDelta'(v) \le \DDelta(v)$ and $\nnabla'(v) \le \nnabla(v)$ for every $v \in V$, in which case we write $(\DDelta',\nnabla') \preceq (\DDelta,\nnabla)$. A flow \emph{partially} routes the demand $(\DDelta,\nnabla)$ if it routes a sub-demand of it. We define the support of a demand to be
\[
\Supp(\DDelta,\nnabla) = \{ v \in V \mid \DDelta(v) + \nnabla(v) > 0 \}.
\]

Let $\pi: V \to U$ be a projection map. For a demand $D = (\DDelta,\nnabla)$ on $V$, we define $\pi(D) = (\DDelta',\nnabla')$ on $U$ by
\[
\DDelta'(u) = \sum_{v \in V : \pi(v) = u} \DDelta(v)
\qquad\text{and}\qquad
\nnabla'(u) = \sum_{v \in V : \pi(v) = u} \nnabla(v)
\]
for every $u \in U$.

\paragraph{Volume.}
Consider $F \subseteq E$. Let
\[
\delta^+_F(v) = \{ (v,u) \in F \}, \qquad
\delta^-_F(v) = \{ (u,v) \in F \}, \qquad
\delta_F(v) = \delta^+_F(v) \cup \delta^-_F(v)
\]
be the sets of edges in $F$ incident to $v$. Define
\[
\vol^+_F(v) = \sum_{e \in \delta^+_F(v)} U_G(e), \qquad
\vol^-_F(v) = \sum_{e \in \delta^-_F(v)} U_G(e), \qquad
\vol_F(v) = \vol^+_F(v) + \vol^-_F(v).
\]
A demand $(\DDelta,\nnabla)$ is \emph{$\vol_F$-respecting} if
\[
\DDelta(v), \nnabla(v) \le \vol_F(v) \quad \text{for every } v \in V.
\]

\paragraph{Max flow.}
The classical max-flow problem is: given a directed graph with edge capacities $G = (V,E)$ and two vertices $s,t$, output an $(s,t)$-flow $\ff$ with maximum value.

In many cases, however, we need to route a demand instead of fixing the source and sink. Thus, we define the \emph{(generalized) max-flow} problem as: given a demand $(\DDelta,\nnabla)$, find a feasible flow with maximum value that routes a sub-demand of $(\DDelta,\nnabla)$. Notice that this problem is equivalent to the classical setting by adding a super source connected to every node $v$ with capacity $\DDelta(v)$ and a super sink connected from every node $v$ with capacity $\nnabla(v)$.

We also define the \emph{$\alpha$-approximate max-flow} problem as finding a feasible flow with value at least a $1/\alpha$-fraction of the maximum value that routes a sub-demand of $(\DDelta,\nnabla)$.

It is well known that \emph{cuts} are dual to \emph{flows}. Thus, if we can return both a flow and a cut, they will certify the approximation ratio. To make this formal, we define the $\amfmc{\alpha}$ problem as finding a feasible flow $\ff$ that routes a sub-demand of $(\DDelta,\nnabla)$ and a cut $S$ satisfying
\begin{equation}\label{eq:maxflowmincut}
    \Val(S) + \sum_{v \in S} \nnabla(v) + \sum_{v \notin S} \DDelta(v) \le \alpha \cdot \Val(\ff).
\end{equation}
In this case, $(\ff,S)$ is called an $\amfmc{\alpha}$ pair. We define $\amfmcDAG{\alpha}$ as the problem $\amfmc{\alpha}$ with inputs restricted to DAGs.

Equation~\eqref{eq:maxflowmincut} certifies that $\ff$ is an $\alpha$-approximate max flow. To see this, consider adding a super source connected to every node with capacity $\DDelta(v)$ and a super sink connected from every node with capacity $\nnabla(v)$, and apply the classical max-flow/min-cut theorem on the super source and sink.

For technical reasons, we will also sometimes use additive error. We define the problem $\amfmc{(\alpha,\delta)}$ as finding a feasible flow $\ff$ that routes a sub-demand of a given $(\DDelta,\nnabla)$ and a cut $S$ satisfying
\begin{equation}\label{eq:maxflowmincutdelta}
    \Val(S) + \nnabla(S) + \DDelta(\bar{S}) \le \alpha \cdot \Val(\ff) + \delta.
\end{equation}

\section{DAG Projections and Their Applications}\label{sec:DAGprojections}

In this section, we introduce the notion of distance-preserving ous DAG projection in \Cref{subsec:basicdefinitions}, state our main result (\Cref{lem:simpledistanceDAGemb}), and show its applications in \Cref{subsec:applicationsofdistanceDAGembedding}.
Then, we introduce the notion of congestion-preserving DAG projection in \Cref{subsec:CongDAGProfDef}, state our main result (\Cref{lem:congestionDAGembedding}), and show its applications in \Cref{subsec:AppCongDAG}.

\subsection{Distance DAG projections}\label{subsec:basicdefinitions}

\begin{definition}[Graph Projections]\label{def:inherited}
    Let $G = (V,E)$ be a graph with edge lengths.
    A \emph{DAG projection onto $G$} is a DAG $D = (V',E')$ together with a projection map $\pi : V' \to V$ that is a \emph{weight-preserving} graph homomorphism. That is, for every $(x,y) \in E'$, we have $(\pi(x), \pi(y)) \in E$ and
    \[
        \ell_D(x,y) = \ell_G(\pi(x), \pi(y)).
    \]
    The \emph{size} of the projection is $|E'|$, and the \emph{width} of the projection is $\max_{v \in V} |\pi^{-1}(v)|$.
\end{definition}

Naturally, we would like $D$ to approximately preserve distances in $G$.

\newcommand{\hdist}{\widehat{\dist}}
\begin{definition}\label{def:DAGembeddingDistance}
    A DAG projection $D$ onto $G$ is \emph{$\lambda$-distance-preserving} if the following holds: for every $s,t \in V$, 
    \[
        \dist_G(s,t) \le \dist_{D}(\pi^{-1}(s),\pi^{-1}(t)) \le \lambda \cdot \dist_G(s,t).
    \]
\end{definition}
The following lemma follows immediately from \Cref{def:inherited}: every $(s',t')$-path in $D$ projects to a path in $G$ of the same length.

\begin{lemma}\label{lem:inheritedoneside}
    If $D$ is a DAG projection onto $G$, then for every $s,t \in V$ we have
    \[
        \dist_G(s,t) \le \dist_{D}(\pi^{-1}(s),\pi^{-1}(t)).
    \]
\end{lemma}

\noindent

We say a projection $D$ onto $G$ is a \emph{distance projection} if it is $\lambda$-distance-preserving for some $\lambda$.

In \Cref{sec:distanceembedding}, we will prove the following key result, which shows that a distance DAG projection can be constructed efficiently in the parallel setting, given only an approximate SSSP algorithm on DAGs.
We state here a simplified version for $\eps \ge 1/\polylog n$. For general $\eps$, see \Cref{lem:distanceDAGembeddingtoSSSP}.

\begin{theorem}[Simplified version of \Cref{lem:distanceDAGembeddingtoSSSP}]\label{lem:simpledistanceDAGemb}
    There is a randomized algorithm that, given a directed graph $G$ and a parameter $o(1) > \eps \ge 1/\polylog n$, constructs a $(1+\eps)$-distance-preserving DAG projection of $G$ with width $n^{o(1)}$ (and its topological order). The algorithm makes $n^{o(1)}$ work-efficient, $\tO{1}$ depth-efficient calls to an $\apxDSSSP{(1+\eps/\log n)}$ oracle.
\end{theorem}

By plugging in the best parallel algorithm for approximate SSSP \cite{CaoFR20} into \Cref{lem:simpledistanceDAGemb}, we obtain the following bound.

\begin{corollary}
    There is a randomized algorithm that, given a directed graph $G$ and a parameter $o(1) > \eps \ge 1/\polylog n$, constructs a $(1+\eps)$-distance-preserving DAG projection of $G$ (and its topological order) with width $n^{o(1)}$ in $\hO{m}$ work and $\hO{\sqrt{n}}$ depth.
\end{corollary}
\subsection{Useful Tools for SSSP}

\paragraph{Boosting Approximate SSSP}. It is possible to reduce exact SSSP to their approximate variants.

For SSSP, we have the following result from \cite{RozhonHMGZ23}.

\begin{lemma}[\cite{RozhonHMGZ23}]\label{lem:boostingSSSP}
    Given a directed graph $G$ with edge lengths and a source $s \in V(G)$,
    there is an algorithm that solves SSSP on $G$ with source $s$. The algorithm makes $\tO{1}$-work-efficient and $\tO{1}$-depth-efficient an oracle that solves
    $O\bigl(\max_{t \in V(G)} \dist_G(s,t)\bigr)$-length $\apxSSSP{(1+o(1/\log n))}$.
\end{lemma}

\Cref{lem:boostingSSSP} is an implication of Theorem 1.7 and Lemma 3.1 in \cite{RozhonHMGZ23}. One minor difference is that in Theorem 1.7 they only specify the oracle as a general $\apxSSSP{(1+o(1/\log n))}$ on arbitrary graphs (not necessarily $G$), instead of an $O\bigl(\max_{t \in V(G)} \dist_G(s,t)\bigr)$-length $\apxSSSP{(1+o(1/\log n))}$. However, by examining their algorithm, the oracle is always called on graphs (not necessarily $G$) whose maximum source distance is $O\bigl(\max_{t \in V(G)} \dist_G(s,t)\bigr)$.

We can strengthen \Cref{lem:boostingSSSP} by a simple trick so that both the algorithm and the oracle are $h$-length bounded.

\begin{corollary}\label{cor:boostingSSSPhlength}
    Given a directed graph $G$ with edge lengths, a source $s \in V(G)$, and a length bound $h$,
    there is an algorithm solving $h$-length SSSP on $G$ with source $s$.
    The algorithm makes $\tO{1}$-work-efficient and $\tO{1}$-depth-efficient an oracle solving $O(h)$-length $\apxSSSP{(1+o(1/\log n))}$.
\end{corollary}

\begin{proof}[Proof of \Cref{cor:boostingSSSPhlength}]
    We build a new graph $G'$ by adding an edge of length $2h$ from $s$ to every other node in $V(G)$. Then we run \Cref{lem:boostingSSSP} on $G'$ and $s$. According to \Cref{lem:boostingSSSP}, it requires an oracle solving $O\bigl(\max_{t \in V(G')} \dist_{G'}(s,t)\bigr)$-length $\apxSSSP{(1+o(1/\log n))}$, where $O\bigl(\max_{t \in V(G')} \dist_{G'}(s,t)\bigr) \le 2h$ by the definition of $G'$. Thus, we can get the exact distances from $s$ in $G'$ by $\tO{1}$ calls to an $O(h)$-length $\apxSSSP{(1+o(1/\log n))}$ oracle on graphs with $O(m)$ edges.

    To obtain $h$-length SSSP on $G$, notice that if $\dist_G(s,t) \le h$ then $\dist_{G'}(s,t) = \dist_G(s,t)$. If $\dist_G(s,t) > h$, we do not need to return the exact distance in $G$. Hence, for every $t \in V(G)$, it suffices to return $\dist_G(s,t) = \dist_{G'}(s,t)$ if $\dist_{G'}(s,t) \le h$, and return $\dist_G(s,t) = +\infty$ otherwise.
\end{proof}

\paragraph{Low Diameter Decomposition}
A \emph{low-diameter decomposition (LDD)} decomposes a directed graph into small-diameter clusters.

\begin{definition}[Low-Diameter Decomposition]\label{def:ldd}
Let $G = (V,E)$ be a directed graph with edge weights and let $d$ be a positive integer. A \emph{low-diameter decomposition (LDD)} with diameter $d$ and slack $\SL$ of $G$ is a sequence of vertex sets $(V_1, V_2, \ldots, V_z)$, which is a partition of $V$, such that
\begin{itemize}
    \item each $V_i$ has weak diameter at most $d$ in $G$, and
    \item for any edge $(u,v) \in E$, the probability that $u \in V_i$ and $v \in V_j$ with $i > j$ is at most $\SL \cdot w(u,v)/d$. We call such an edge $(u,v)$ a \emph{reversed edge}.
\end{itemize}
\end{definition}

Ashvinkumar et al.~\cite{AshvinkumarBCGH24} show how to compute an LDD with $\tO{1}$ calls to an SSSP oracle.

\begin{lemma}[\cite{AshvinkumarBCGH24}]\label{lem:ldd}
    There is a randomized algorithm that, given a directed graph $G$ and a positive integer $d$, computes an LDD with diameter $d$ and slack $\SL = O(\log^2 n)$ using $\tO{1}$ calls to a $(\SL \cdot d)$-length SSSP oracle on graphs with $O(n)$ vertices and $O(m)$ edges.
\end{lemma}

\Cref{lem:ldd} follows from Lemma 4 in \cite{AshvinkumarBCGH24}, with two minor differences.
\begin{enumerate}
    \item Lemma 4 in \cite{AshvinkumarBCGH24} does not output the order of the sequence $(V_1, \ldots, V_z)$ but only outputs the set of reversed edges. However, this is not a problem since their algorithm inherently specifies the order of the partition.
    \item Lemma 4 in \cite{AshvinkumarBCGH24} only specifies the oracle as a general SSSP, instead of a $(\SL \cdot d)$-length SSSP. However, every oracle call in their algorithm only needs the vertices within distance at most $\SL \cdot d$ from the source, so it is effectively a $(\SL \cdot d)$-length SSSP oracle.
\end{enumerate}

\subsection{Applications of Distance DAG Projections}\label{subsec:applicationsofdistanceDAGembedding}

\paragraph{Reducing Exact SSSP to Approximate SSSP on DAGs.}
The first immediate implication is that we can reduce exact SSSP on general graphs to SSSP on DAGs.

\begin{restatable}{lemma}{SSSPtoDAG}\label{thm:SSSPtoDAG}
There is a parallel randomized algorithm solving (exact) SSSP on directed graphs, with $n^{o(1)}$-work-efficient $\tO{1}$-depth-efficient oracle calls to a $(1+o(1/\log n))$-approximate SSSP oracle on DAGs.\footnote{The reduction is strong enough that we can assume the DAG oracle is given a topological order as part of the inputs. This is a basic assumption across the paper: When we call an oracle on DAGs, a topological order of the DAG must be given.}
\end{restatable}

The main idea is as follows. We use \Cref{lem:simpledistanceDAGemb} to construct a DAG projection $D$ of $G$, and then run the oracle to solve SSSP on $D$. To get the distance between $(s,t)$ in $G$, we take the minimum over all pairs $(s',t')$ in $D$ with $\pi(s')=s$ and $\pi(t')=t$; this gives an approximate distance in $G$. Finally, we boost the approximate distance to the exact distance by using \Cref{lem:boostingSSSP}. For the formal argument, see \Cref{lem:reducingSSSPtoDAGs}.

\paragraph{Reducing Exact SSSP to Undirected Graphs.}
It is known that SSSP on DAGs can be reduced to (exact) SSSP on undirected graphs (see \cite[ Lemma 4.6]{HoppenworthXX25}). Thus, we obtain the following lemma.

\begin{restatable}{lemma}{SSSPtoUndir}\label{thm:SSSPtoUndir}
There is a parallel randomized algorithm solving (exact) SSSP on directed graphs, with $n^{o(1)}$-work-efficient $\tO{1}$-depth-efficient oracle calls to an exact SSSP oracle on undirected graphs.
\end{restatable}

\begin{proof}
    We apply \Cref{thm:SSSPtoDAG} with the following algorithm that solves exact SSSP on a DAG $D$ using a single oracle call to exact SSSP on an undirected graph $D'$. We set $V(D) = V(D')$. Suppose $D$ has a topological order $(v_1,\dots,v_n)$.\footnote{The reduction from directed SSSP to DAGs assumes that a topological order for the input DAG is given. This is not always a standard assumption in parallel settings, but when we generate the DAG projection, a corresponding topological order is also produced.}
    For every edge $(v_i,v_j) \in E(D)$ with $i < j$, we create an undirected edge $(v_i,v_j) \in E(D')$ with length
    \[
        \ell_{D'}(v_i,v_j) = \ell_D(v_i,v_j) + (j-i) \cdot M,
    \]
    where $M$ is a sufficiently large constant chosen as described in the text.
    Notice that the $(v_i,v_j)$-distance in $D'$ is at most $\dist_D(v_i,v_j) + (j-i)\cdot M$ by following the shortest path in $D$, and any path in $D'$ that does not follow the topological order must incur an extra $(j-i+1)\cdot M > \dist_D(v_i,v_j) + (j-i)\cdot M$. Thus, running exact SSSP on $D'$ gives exact SSSP on $D$ after subtracting the corresponding $(j-i)\cdot M$ term.
\end{proof}

\paragraph{Reducing Hop-set Construction to DAGs.}
The definition of DAG projection also immediately gives a hop-set reduction from general graphs to DAGs.

\begin{restatable}{lemma}{HopsettoDAG}\label{thm:HopsettoDAG}
Suppose there is an oracle $\Ohs$ constructing $(\beta,\eps)$-hopset of size $s(m,n)$ for $o(1)>\eps>1/\polylog(n)$ on DAGs, then there is a randomized algorithm constructing $(\beta,3\eps)$-hopset of size $s(m^{1+o(1)},n^{1+o(1)})$ on general directed graphs, with $n^{o(1)}$-work-efficient and $\tO{1}$-depth-efficient calls to $(1+\eps/\log n)$-approximate SSSP and $\Ohs$. 

\end{restatable}
\begin{proof}
    The algorithm takes a graph $G$ and applies \Cref{lem:simpledistanceDAGemb} with parameter $\eps$, which gives a DAG projection $D$ of $G$ with width $n^{o(1)}$, using $n^{o(1)}$ work-efficient and $\tO{1}$ depth-efficient calls to $\apxDSSSP{(1+\eps/\log n)}$. Since $D$ is a projection, it has $m^{1+o(1)}$ edges and $n^{1+o(1)}$ vertices.
    
    We then apply the oracle $\Ohs$ on $D$ to get a hop set $H^D$ of size $s(m^{1+o(1)}, n^{1+o(1)})$ which is a $(\beta,\eps)$-hop-set of $D$. Let
    \[
        H = \{ (\pi(u), \pi(v)) \mid (u,v) \in H^D \}.
    \]
    We show that $H$ is a $(\beta, 3\eps)$-hop-set for $G$.
    
    First, for any $(u,v) \in H$, pick $(u',v') \in H^D$ with $\pi(u') = u$ and \(\pi(v') = v\). By construction, the weight of $(u,v)$ in $H$ is the same as the weight of $(u',v')$ in $H^D$, and since $D$ is a distance projection, we have
    \[
        \dist_G(u,v) \le \dist_D(u',v') \le \ell_{H^D}(u',v') = \ell_H(u,v),
    \]
    so adding $H$ does not create shorter-than-original edges in $G$.
    
    Now fix $s,t \in V$. By the definition of a distance DAG projection, there exist $s',t' \in V(D)$ with $\pi(s') = s$, $\pi(t') = t$ such that
    \[
        \dist_D(s',t') \le (1+\eps) \cdot \dist_G(s,t).
    \]
    Since $H^D$ is a $(\beta,\eps)$-hop-set for $D$, there is a path $p^D$ in $D \cup H^D$ with at most $\beta$ edges and
    \[
        \ell_{D \cup H^D}(p^D) \le (1+\eps) \cdot \dist_D(s',t')
        \le (1+\eps)(1+\eps) \cdot \dist_G(s,t)
        \le (1+3\eps) \cdot \dist_G(s,t)
    \]
    for $\eps \le 1$. The projection $\pi(p^D)$ is a path in $G \cup H$ with the same number of edges and the same length, so
    \[
        \dist^{(\beta)}_{G \cup H}(s,t) \le (1+3\eps) \cdot \dist_G(s,t).
    \]
\end{proof}

\paragraph{Reducing Distance Preserver Construction to DAGs.}
We can reduce the construction of $(1+\epsilon)$-approximate  distance preservers to DAGs.
\begin{restatable}{lemma}{PreservertoDAG}\label{thm:PreservertoDAG}
Suppose there is an oracle $\Odp$ constructing $(1+\eps)$-approximate distance preserver of size $s(n,p)$ for $o(1)>\eps>1/\polylog(n)$ on $n$-node DAGs with $p$ demand pairs, then there is a randomized algorithm constructing $(1+3\eps)$-approximate distance preserver of size $s(n^{1+o(1)},p)$ on $n$-node general graphs with $p$ demand pairs, with $n^{o(1)}$-work-efficient and $\tO{1}$-depth-efficient calls to to $(1+\eps/\log n)$-approximate SSSP and $\Odp$.

\end{restatable}
\begin{proof}
    The algorithm takes a graph $G$ and applies \Cref{lem:simpledistanceDAGemb} with parameter $\eps$, which gives a DAG projection $D$ of $G$ with width $n^{o(1)}$, using $n^{o(1)}$ work-efficient and $\tO{1}$ depth-efficient calls to $\apxDSSSP{(1+\eps/\log n)}$. Again, $D$ has $m^{1+o(1)}$ edges and $n^{1+o(1)}$ vertices.
    
    We modify $D$ as follows: for every $v \in V(G)$, add two vertices $\In{v}$ and $\Out{v}$ to $D$, and for every $v' \in V(D)$ with $\pi(v') = v$ add edges $(\In{v}, v')$ and $(v', \Out{v})$ of length $0$. Denote the resulting graph by $D'$.
    
    We apply the oracle $\Odp$ on $D'$ with demand pairs
    \[
        P^D := \{ (\In{u}, \Out{v}) \mid (u,v) \in P \}
    \]
    to get $H^D$ of size $s(n^{1+o(1)}, p)$. Let
    \[
        H = \{ (\pi(u), \pi(v)) \mid (u,v) \in H^D,\, u \in V(D),\, v \in V(D) \}.
    \]
    We emphasize that $H$ is well-defined because $\pi$ is a graph homorphism according \Cref{def:DAGembeddingDistance}. (A partial projection, as in \Cref{def:DAGembeddingdummy}, is not enough.)
    
    Now, we show that $H$ is a $(1+3\eps)$-approximate distance preserver for $G$.
    Fix $(s,t) \in P$. Then $(\In{s}, \Out{t}) \in P^D$, so
    \[
        \dist_{H^D}(\In{s}, \Out{t}) \le (1+\eps) \cdot \dist_{D'}(\In{s}, \Out{t}).
    \]
    By the way we added zero-length in/out vertices, there exist $s',t' \in V(D)$ with $\pi(s') = s$ and $\pi(t') = t$ such that
    \[
        \dist_{D'}(\In{s}, \Out{t}) = \dist_D(s',t').
    \]
    Also, the path witnessing $\dist_{H^D}(\In{s}, \Out{t})$ can be chosen to stay inside $V(D)$ except for the endpoints, so its projection is a path in $H$. Therefore,
    \[
        \dist_H(s,t) \le (1+\eps) \cdot \dist_D(s',t') \le (1+3\eps) \cdot \dist_G(s,t),
    \]
    where the last inequality uses that $D$ is a $(1+\eps)$-distance-preserving DAG projection of $G$.
\end{proof}
Previously, the best $(1+\eps)$-approximate distance preserver has size $O(n+\sqrt{n}\cdot p)$ \cite{CoppersmithE05,HoppenworthXX25}\footnote{Lemma 4.6 of \cite{HoppenworthXX25} shows that distance preserver on DAGs reduces to (exact) distance preservers on undirected graphs, which exhibits a construction in \cite{CoppersmithE05}.} We thus immediately get \Cref{cor:distancepreserver} by applying \Cref{thm:PreservertoDAG}.

\subsection{Congestion DAG Projections}
\label{subsec:CongDAGProfDef}

We can also consider the case where we want $D$ to preserve flow. For technical reasons, we allow the DAG projection $D$ to contain \emph{dummy vertices} that map to a dummy value $\bot$, and we do not require the projection to be a homomorphism. This can be removed with some efforts, but this version makes the later algorithms cleaner. We call it a \emph{partial projection}, and we will usually omit the word “partial”.

\begin{definition}[Partial Projections]\label{def:DAGembeddingdummy}
    Let $G = (V,E)$ be a graph with edge capacities. A \emph{DAG (partial) projection onto $G$} is a DAG $D = (V',E')$ together with a (partial) projection map $\pi : V' \to V \cup \{\bot\}$.
\end{definition}

Now we define congestion-preserving projections.

\begin{definition}\label{def:DAGembeddingCongestion}
    A DAG projection $D$ of $G = (V,E)$ is \emph{$\kappa$-congestion-preserving} if for every $S,T\subseteq V(G)$, we have
    \[
    \maxflow_{G}(S,T)\le\maxflow_{G'}(\pi^{-1}(S),\pi^{-1}(T))\le \kappa\cdot \maxflow_{G}(S,T).
    \]
\end{definition}

Often we want not only the existential property above, but also an explicit and efficient way to project flows and cuts from $D$ back to $G$. This is given by the following definition.

\begin{definition}[Projection Algorithm]
    Let $D$ be a DAG projection of $G = (V,E)$. A $\kappa$-congestion-preserving \emph{projection algorithm} associated with $D$ is an algorithm that, given either
    \begin{itemize}
        \item a flow $\ff^D$ in $D$ with $\pi(\Supp(\Dem(\ff^D))) \subseteq V$, returns a flow $\ff$ in $G$ with congestion at most $\kappa \cdot \Cong(\ff^D)$ such that $\pi(\Dem(\ff^D)) = \Dem(\ff)$; or
        \item a cut $S^D$ in $D$, returns a cut $S$ in $G$ with value at most $\Val(S^D)$ such that $\HL(S^D) \subseteq S \subseteq \pi(S^D)$.
    \end{itemize}
    The projection algorithm is \emph{efficient} if it takes $\hO{|E(D)|}$ work and $\hO{1}$ depth.

    When we say $D$ is associated with a projection algorithm, we assume the projection is surjective, i.e. $\pi^{-1}(v) \neq \emptyset$ for every $v \in V(G)$.
\end{definition}

In \Cref{sec:congestionDAGembedding}, we will prove the following theorem, showing that congestion DAG projections can be constructed efficiently in the parallel setting using an approximate max-flow oracle on DAGs.

\begin{restatable}{theorem}{congestionDAGembedding}
\label{lem:congestionDAGembedding}
Suppose there is an oracle for $\amfmcDAG{\alpha}$ where $\alpha = n^{o(1)}$. Then there is a randomized algorithm that, given a directed graph $G = (V,E)$ with edge capacities, outputs a DAG projection $D$ of $G$ together with an efficient $n^{o(1)}$-congestion-preserving projection algorithm, such that $|E(D)| = \hO{|E(G)|}$ (and $D$’s topological order is returned). The algorithm makes $n^{o(1)}$ work-efficient and $n^{o(1)}$ depth-efficient calls to $\amfmcDAG{\alpha}$.
\end{restatable}

\subsection{Applications of Congestion DAG Projections}
\label{subsec:AppCongDAG}

\paragraph{Reducing Exact Max Flow to Approximate Max Flow on DAGs.}
We first show the immediate application of reducing exact max flow to DAGs.

\begin{restatable}{lemma}{MaxflowtoDAG}\label{thm:MaxflowtoDAG}
There is a parallel randomized algorithm solving (exact) max flow on directed graphs, with $n^{o(1)}$-work-efficient $n^{o(1)}$-depth-efficient oracle calls to a $n^{o(1)}$-approximate max flow oracle\footnote{For technical reasons (cut-matching game in expander decomposition), we require this max flow algorithm to return an approximate max flow and min cut pair.} on DAGs.
\end{restatable}

To prove the lemma, we first show how to use a congestion DAG projection to reduce the max-flow/min-cut problem to DAGs.

\begin{lemma}\label{lem:DAGprojectiontoMaxFlow}
    Let $D$ be a DAG projection of $G = (V,E)$ with a $\kappa$-congestion-preserving efficient projection algorithm. Let $\cO$ be an oracle solving $\amfmc{\alpha}$ on DAGs. Then there is an algorithm solving $\amfmc{\alpha \cdot \kappa}$ on $G$, with complexity proportional to one efficient call to $\cO$ plus the cost of the projection algorithm.
\end{lemma}
\begin{proof}
    We build a graph $D'$ from $D$ as follows. For every $v \in V$, add a vertex $s_v$ and edges $(s_v, v')$ for every $v' \in V(D)$ with $\pi(v') = v$, each with infinite capacity. Also for every $v \in V$, add a vertex $t_v$ and edges $(v', t_v)$ for every $v' \in V(D)$ with $\pi(v') = v$, each with infinite capacity. Then add a super source $s$ and edges $(s, s_v)$ for every $v \in V$ with capacity $\DDelta(v)$, and add a super sink $t$ and edges $(t_v, t)$ for every $v \in V$ with capacity $\nnabla(v)$. Denote the resulting graph by $D'$.

    Run the oracle $\cO$ on $D'$ with source $s$ and sink $t$, and let the resulting flow and cut be $(\ff', S')$. Then
    \[
        \Val(S') \le \alpha \cdot \Val(\ff').
    \]
    Restrict $\ff'$ and $S'$ to $D$ to get $\ff^D$ and $S^D$. By construction, all source/sink demand of $\ff^D$ lives on vertices that project to $V$, i.e. $\pi(\Supp(\Dem(\ff^D))) \subseteq V$. Apply the projection algorithm to $(\ff^D, S^D)$ to get a flow $\ff$ in $G$ and a cut $S$ in $G$. By the guarantee of the projection algorithm:
    \[
        \Cong(\ff) \le \kappa \cdot \Cong(\ff^D) \le \kappa,
    \]
    so $\ff / \kappa$ is feasible in $G$, and
    \[
        \Val(S) \le \Val(S^D), \qquad \HL(S^D) \subseteq S \subseteq \pi(S^D).
    \]

    Next we relate the cut values. As in the usual super-source/super-sink reduction, any $v \in V$ for which there exists $v' \in V(D)$ with $v' \notin S^D$ must have $s_v \notin S'$ (otherwise an infinite-capacity edge is cut), so $(s,s_v)$ contributes $\DDelta(v)$ to $\Val(S')$. Symmetrically, any $v \in V$ for which there exists $v' \in V(D)$ with $v' \in S^D$ must have $t_v \in S'$ (otherwise an infinite-capacity edge is cut), so $(t_v,t)$ contributes $\nnabla(v)$ to $\Val(S')$. All remaining contribution to $\Val(S')$ comes from edges inside $D$, i.e. from $\Val(S^D)$. Hence
    \[
        \Val(S') \ge \Val(S^D) + \sum_{v \notin \HL(S^D)} \DDelta(v) + \sum_{v \in \pi(S^D)} \nnabla(v).
    \]
    Since $\Val(S') \le \alpha \cdot \Val(\ff')$ and $\Val(\ff) = \Val(\ff')$, and since
    \[
        \HL(S^D) \subseteq S \subseteq \pi(S^D),
    \]
    we get
    \[
        \Val(S) + \sum_{v \in S} \nnabla(v) + \sum_{v \notin S} \DDelta(v)
        \;\le\;
        \Val(S^D) + \sum_{v \in \pi(S^D)} \nnabla(v) + \sum_{v \notin \HL(S^D)} \DDelta(v)
        \;\le\;
        \alpha \cdot \Val(\ff).
    \]
    Finally, we output the feasible flow $\ff/\kappa$ and the cut $S$. Multiplying the right-hand side by $\kappa$ to account for scaling gives
    \[
        \Val(S) + \sum_{v \in S} \nnabla(v) + \sum_{v \notin S} \DDelta(v)
        \;\le\;
        \alpha \kappa \cdot \Val(\ff/\kappa),
    \]
    so $(\ff/\kappa, S)$ is an $\amfmc{\alpha \cdot \kappa}$ pair.
\end{proof}

Now we can prove \Cref{thm:MaxflowtoDAG}.

\begin{proof}[Proof of \Cref{thm:MaxflowtoDAG}]
    Given a directed graph $G$, apply \Cref{lem:congestionDAGembedding} to construct a DAG projection of $G$ with an $n^{o(1)}$-congestion-preserving efficient projection algorithm, of size $\hO{|E(G)|}$, and with a topological order. Then apply \Cref{lem:DAGprojectiontoMaxFlow} to obtain an $n^{o(1)}$-approximate max flow on $G$. It is folklore that an approximate max-flow algorithm can be turned into an exact max-flow algorithm by working on the residual graph and rerunning the approximation, so the lemma follows.
\end{proof}

\paragraph{Single-source Bounded Minimum Cuts.}
Another application is the single-source $k$-bounded minimum cuts problem. Previously, efficient algorithms were known only for DAGs. Using our congestion DAG projection, we obtain an algorithm for general graphs.

\begin{restatable}{lemma}{SSkmincut}\label{thm:SSkmincut}
    There is a randomized algorithm given a directed graph with edge capacities, a source $s$, and an integer $k$, find $n^{o(1)}$-approximate $(s,t)$-minimum cut for all $t\in V$ with $\mincut_G(s,t)\le k$ in $k^{\omega}m^{1+o(1)}$ time.
\end{restatable}

\begin{proof}
    Given a directed graph $G = (V,E)$ with edge capacities, apply \Cref{lem:expanderdecompositiontoDAGembedding} with $\sigma = 1$ and a directed expander hierarchy whose last layer $E_t = \emptyset$, to construct a DAG projection with an $n^{o(1)}$-congestion-preserving efficient projection algorithm of size $\hO{|E(G)|}$ and integral capacities. Here we use the almost-linear-time max-flow algorithm of \cite{ChenKLPGS22} as the oracle and the directed expander-hierarchy algorithm of \cite{BBLST25} to obtain the hierarchy in $m^{1+o(1)}$ time.

    Then, add a vertex $s^*$ to $D$ and connect it to every vertex in $\pi^{-1}(s)$ with capacity $k$. For every $t \in V \setminus \{s\}$, add a vertex $t^*$ to $D$ and connect every vertex in $\pi^{-1}(t)$ to $t^*$ with capacity $k$. For every edge $(u,v) \in E(D)$, replace it by $\min(k, U_D(u,v))$ parallel uncapacitated edges. Denote the resulting graph by $D'$. Then $|E(D')| = k \cdot \hO{|E(G)|}$.

    Run the single-source $k$-edge-connectivity algorithm of \cite[Theorem 1.2]{cheung2013graph} on $D'$ from $s^*$ to all $t^*$ to obtain $\lambda_{D'}(s^*, t^*)$ for all $t$, in time $k^{\omega} \hO{|E(G)|}$. We claim that $\lambda_{D'}(s^*, t^*)$ is an $n^{o(1)}$-approximation to $\lambda_G(s,t)$ whenever $\lambda_G(s,t) \le k$.

    For the lower bound, let $\ff'$ be a max flow from $s^*$ to $t^*$ in $D'$. Restrict $\ff'$ to $D$ to obtain $\ff^D$ whose source demand lies on $\pi^{-1}(s)$ and sink demand lies on $\pi^{-1}(t)$. Because we introduced at most as many parallel edges as the original capacity, the congestion does not increase. Applying the projection algorithm yields a flow $\ff$ from $s$ to $t$ in $G$ of the same value and congestion at most $n^{o(1)}$. Scaling down by $n^{o(1)}$ makes it feasible in $G$, so
    \[
        \lambda_G(s,t) \ge \lambda_{D'}(s^*, t^*) / n^{o(1)}.
    \]

    For the upper bound, let $S'$ be a minimum $(s^*, t^*)$-cut in $D'$ with value $\lambda_{D'}(s^*, t^*)$. If $S'$ cuts an edge adjacent to $s^*$ or $t^*$, then the cut has value at least $k \ge \lambda_G(s,t)$ (by the assumption $\lambda_G(s,t) \le k$), so we are done. Otherwise, $\pi^{-1}(s) \subseteq S'$ and $\pi^{-1}(t) \cap \bar S' = \emptyset$. Restrict $S'$ to $D$ to get $S^D$. By construction (replacing each edge by up to $k$ parallels), either the value of $S^D$ increases past $k$, in which case we are done, or it stays the same. Moreover, $s \in \HL(S^D)$ and $t \notin \pi(S^D)$, so applying the projection algorithm to $S^D$ gives a valid $(s,t)$-cut in $G$ of value at most $\lambda_{D'}(s^*, t^*)$.
\end{proof}

\section{Distance DAG Projection Construction}\label{sec:distanceembedding}

In this section, we show an efficient parallel algorithm for constructing a $(1+\eps)$-distance-preserving DAG projection, using oracle calls to only an approximate SSSP on DAGs.

\begin{theorem}\label{lem:distanceDAGembeddingtoSSSP}
    There is a randomized algorithm that, given a directed graph $G$ with edge lengths and parameters $0 < \eps < o(1)$ and $0 < \delta < 0.5$, constructs a $(1+\eps)$-distance-preserving DAG projection of $G$ with width
    \[
        w \;=\; \bigl(\tO{1/\eps}\bigr)^{\log^\delta n} \cdot n^{1 / \log^\delta n}.
    \]
    The algorithm makes $n^{o(1)}w$-work-efficient $\tO{1}$-depth-efficient calls to an $\apxDSSSP{(1+\eps/\log n)}$ oracle on graphs.
\end{theorem}

\paragraph{High-Level Strategy.}
By \Cref{lem:boostingSSSP}, to reduce SSSP on general graphs to DAGs, it suffices to construct a $(1+\eps)$-distance-preserving DAG projection with $\eps = o(1/\log n)$. The difficulty is that we only know how to construct such a projection using oracle calls to an \emph{exact} SSSP algorithm (to build an LDD), which creates a “chicken-and-egg’’ situation.

To resolve this issue, as mentioned in the overview, we will first define the $h$-length version of DAG projection.
\begin{definition}\label{def:DAGembeddingDistancelength}
    A DAG projection $D$ onto $G$ is $h$-length \emph{$\lambda$-distance-preserving} if the following holds: for every $s,t \in V$ with $\dist_G(s,t)\le h$, we have
    \[
        \dist_G(s,t) \le \dist_{D}(\pi^{-1}(s),\pi^{-1}(t)) \le \lambda \cdot \dist_G(s,t).
    \]
\end{definition}

We first observe that, when $h=\tilde{O}(1)$, the problem becomes trivial: making $h$ copies of the vertex set, and build a $h$-layered graph with original edges connecting adjacent layers suffices (assuming edge length are positive integers).

Our main technical contribution in this section is to show the following spiral reduction to graduate reduce $h$.
\begin{enumerate}
    \item Using similar ideas in the overview, we can show how to reduce the algorithm of constructing (a fixed size) $h$-length distance-preserving DAG projection to $h$-length SSSP.
    \item We can reduce $h$-length SSSP to $h/z$-length distance-preserving DAG projection with the help of a DAG oracle for SSSP using \Cref{lem:reducingSSSPtoDAGs}.
\end{enumerate}

\paragraph{Organization.}
In \Cref{subsec:SSSPtoDAGProjection}, we show the second point of reducing $h$-length SSSP to $h/z$-length distance-preserving DAG projection. In later sections, we show the algorithm of reducing constructing a fixed-size DAG projection to an SSSP oracle, and its analysis.

\subsection{SSSP via DAG Projections for Smaller Length Constraints}\label{subsec:SSSPtoDAGProjection}

The following definition will be used repeatedly to construct DAG projections from smaller ones.

\begin{definition}[Induced DAG Projections]\label{def:induced}
    Let $G = (V,E)$ be a directed graph and let $(G_1, G_2, \dots, G_z)$ be a sequence of DAG projections of $G$. We say that $G'$ is a DAG projection of $G$ \emph{induced} from $(G_1, G_2, \dots, G_z)$ if $G'$ is constructed as follows:
    \begin{itemize}
        \item let $G' = G_1 \cup G_2 \cup \dots \cup G_z$ (i.e. take the disjoint union of the DAGs);
        \item for every $u \in V(G_i)$ and $v \in V(G_j)$ with $i < j$, if $\pi(u) = \pi(v)$ or $(\pi(u), \pi(v)) \in E(G)$, then we add the edge $(u,v)$ to $G'$ with the same length and capacity as in $G$.
    \end{itemize}
\end{definition}
 Intuitively, $G$ ``concatenate'' the DAG projections $G_1,...,G_z$ together, and then ``lift'' original edges of $G$ to connect between different $G_i$.

We first show that a distance DAG projection allows us to reduce approximate SSSP on general graphs to SSSP on DAGs. The reduction uses an additional parameter $\lambda$ that lets us decrease the length bound.

\begin{lemma}\label{lem:lengthboundapxSSSPtoDAGs}
    Suppose $0 < \eps, \eps' < 1$. There is an algorithm that, given
    \begin{itemize}
        \item a directed graph $G$ with edge lengths,
        \item an integer $\lambda \ge 1$,
        \item an $h$-length $(1+\eps)$-distance-preserving DAG projection $D$ of $G$ with width $w$, and
        \item a source $s \in V(G)$,
    \end{itemize}
    solves $(\lambda h)$-length $\apxSSSP{(1+\eps)(1+\eps')}$ on $(G,s)$ by making $\lambda w$-work-efficient $\tO{1}$-depth-efficient calls to an $\apxDSSSP{(1+\eps')}$ oracle.
\end{lemma}

\begin{proof}
    Let $D'$ be the induced DAG projection induced from a sequence of $\lambda$ copies of $D$ (cf. \Cref{def:induced}). For every $s' \in V(D')$ with $\pi(s') = s$, we run the $\apxDSSSP{(1+\eps')}$ oracle on $D'$ with source $s'$. Let $\tilde{d}_{D'}(s',v')$ be the resulting approximate distance from $s'$ to $v'$ in $D'$.

    For each $v \in V(G)$, we output
    \[
        d_G(v) := \min \{ \tilde{d}_{D'}(s',v') \mid \pi(s') = s,\, \pi(v') = v \}.
    \]
    Since there are at most $w$ preimages of any vertex in one copy and we have $\lambda$ copies, this makes at most $\lambda w$ oracle calls. Each such call is on a graph of size $O(\lambda w m)$, so the total extra work is $O(\lambda w m)$ and depth is $\tO{1}$.

    For correctness, note first that $D'$ is still a projection of $G$, so any path in $D'$ projects to a path of the same length in $G$. Therefore $d_G(v) \ge \dist_G(s,v)$.

    Now let $v \in V(G)$ with $\dist_G(s,v) \le \lambda h$, and let $p$ be a shortest $(s,v)$-path in $G$. Partition $p$ into at most $\lambda$ subpaths $p_1,\dots,p_\lambda$ such that every $p_i$ has length at most $h$ (this is possible because the total length is at most $\lambda h$). By the definition of an $h$-length $(1+\eps)$-distance-preserving DAG projection, for each $p_i$ there is a corresponding path $p_i'$ in $D$ with
    \[
        \pi(\Start(p_i')) = \Start(p_i), \quad \pi(\End(p_i')) = \End(p_i),
        \quad \text{and} \quad \ell_D(p_i') \le (1+\eps) \cdot \ell_G(p_i).
    \]
    Let $p_i''$ be the copy of $p_i'$ in the $i$-th copy of $D$ inside $D'$. By the construction of the induced projection (we connect later copies after earlier ones when their projections match), the concatenation
    \[
        p'' := p_1'' \oplus p_2'' \oplus \dots \oplus p_\lambda''
    \]
    is a valid path in $D'$ from some $s' \in \pi^{-1}(s)$ in the first copy to some $v' \in \pi^{-1}(v)$ in the last copy. Its length is
    \[
        \ell_{D'}(p'') \le (1+\eps) \sum_i \ell_G(p_i) = (1+\eps) \cdot \ell_G(p) = (1+\eps) \cdot \dist_G(s,v).
    \]
    The oracle gives a $(1+\eps')$-approximation to this distance in $D'$, so
    \[
        d_G(v) \le (1+\eps') \cdot \ell_{D'}(p'')
        \le (1+\eps')(1+\eps) \cdot \dist_G(s,v).
    \]
\end{proof}

Since exact SSSP can be reduced to approximate SSSP, we can now reduce exact SSSP to DAGs.

\begin{lemma}\label{lem:reducingSSSPtoDAGs}
    There is an algorithm that, given an $m$-edge directed graph $G$ with edge lengths, a source $s$, and integers $h \ge 1$ and $\lambda \ge 1$, solves $(\lambda h)$-length exact SSSP. The algorithm makes
    \begin{itemize}
        \item $O(1)$-work-efficient $\tO{1}$-depth-efficient calls to an algorithm that outputs an $h$-length $(1+o(1/\log n))$-distance-preserving DAG projection of width $w$ on graphs with $O(m)$ edges,
        \item $O(\lambda w)$-work-efficient $\tO{1}$-depth-efficient calls to an $\apxDSSSP{(1+o(1/\log n))}$ oracle.
    \end{itemize}
\end{lemma}
\begin{proof}
    By \Cref{cor:boostingSSSPhlength}, to solve $(\lambda h)$-length exact SSSP it suffices to solve $\tO{1}$ instances of $(\lambda h)$-length $\apxSSSP{(1+o(1/\log n))}$ on graphs with $\tO{m}$ edges.

    For one such instance on a graph $G'$ with $\tO{m}$ edges, first run the DAG-projection oracle on $G'$ to obtain an $h$-length $(1+o(1/\log n))$-distance-preserving DAG projection $D$ of $G'$ with width $w$. Then apply \Cref{lem:lengthboundapxSSSPtoDAGs} with $\eps = o(1/\log n)$ and $\eps' = o(1/\log n)$ to solve the $(\lambda h)$-length approximate SSSP instance using $\lambda w$ calls to an $\apxDSSSP{(1+o(1/\log n))}$ oracle on graphs with $O(\lambda w m)$ edges and with additional $O(\lambda w m)$ work and $\tO{1}$ depth.
\end{proof}

\subsection{DAG Projections via SSSP}\label{subsec:algorithmdistanceDAGembeddingtoSSSP}
In this section, we describe a recursive algorithm $\DDE(G,\eps,h)$ that takes a directed graph $G$ with edge lengths and returns an $h$-length $(1+\eps)$-distance-preserving DAG projection of $G$ with width $\hO{1}$. To obtain \Cref{lem:distanceDAGembeddingtoSSSP}, it suffices to call $\DDE(G,\eps,nW)$, since every shortest $s$–$t$ path has length at most $nW$. Let $\delta\in(0,1)$ be a constant parameter.

\paragraph{Step 1 (LDD).}
The algorithm computes $O(\log n)$ independent LDDs of $G$ with diameter $2^i$ and slack $\SL = O(\log^2 n)$, for every
\[
0 \le i \le \lceil \log h \rceil,
\]
and collects them in a family $\cS$. These LDDs are computed by combining \Cref{lem:ldd} and \Cref{lem:reducingSSSPtoDAGs}, i.e. we reduce the LDD computations to
\begin{enumerate}
    \item $\tO{1}$ calls to $\DDE(G', \eps, h/\lambda)$ where $G'$ has $O(m)$ edges, and
    \item $\tO{1}$ calls to the $\apxDSSSP{(1+\eps/\log n)}$ oracle on graphs with $\tO{\lambda m}$ edges,
\end{enumerate}
where we set the parameter $\lambda = 2^{\log^{1-\delta}n}$.

By \Cref{def:ldd}, each $S \in \cS$ is a sequence of vertex sets forming a partition of $V$. Let $\sigma = 2^{\log^{1-\delta}n}$ be a parameter. For each $S \in \cS$ and each $C \in S$, we call $C$
\begin{itemize}
    \item a \emph{large cluster} if $|C| \ge n/\sigma$, and
    \item a \emph{small cluster} if $|C| < n/\sigma$.
\end{itemize}

\paragraph{Step 2 (recursive construction for small clusters).}
For each $S \in \cS$ and each small cluster $C \in S$, we make the recursive call
\[
    D_C \leftarrow \DDE\bigl(G[C],\, (1 - \tfrac{1}{\log n}) \cdot \eps,\, nW\bigr).
\]

\paragraph{Step 3 (shortest-path trees for large clusters).}
We first make a recursive call on the whole graph to get a smaller-length projection:
\[
    D^\ell \leftarrow \DDE\bigl(G,\, (1 - \tfrac{10}{\log n}) \cdot \eps,\, h/\lambda\bigr).
\]
For each $S \in \cS$ and each large cluster $C \in S$, let $r_C \in C$ be an arbitrary vertex. Using \Cref{lem:lengthboundapxSSSPtoDAGs} with the $D^\ell$ as the required DAG projection, together with the $\apxDSSSP{(1+\eps/\log n)}$ oracle, we compute an $h$-length $\alpha$-approximate shortest-path tree $\To_C$ rooted at $r_C$, where
\[
    \alpha
    =
    \bigl(1 + (1 - \tfrac{10}{\log n})\eps \bigr)
    \cdot
    \bigl(1 + \tfrac{\eps}{\log n} \bigr)
    \le
    1 + \bigl(1 - \tfrac{5}{\log n}\bigr)\eps,
\]
and we also compute a reversed shortest-path tree $\Ti_C$ rooted at $r_C$ (i.e. every $(s,r_C)$-path in $\Ti_C$ is a shortest $(s,r_C)$-path). Let $D_C$ be the DAG obtained by combining $\To_C$ and $\Ti_C$ (as \emph{disjoint copies} of subgraphs of $G$) with their roots $r_C$ identified, and define the vertex labeling to $V(G)$ in the natural way.

\paragraph{Step 4 (combining everything).}
For each $S = (V_1, V_2, \dots, V_z) \in \cS$, let $D_S$ be the DAG projection induced from the sequence
\[
    (D_{V_1}, D_{V_2}, \dots, D_{V_z}),
\]
where $D_{V_i}$ is the DAG constructed in Step 2 or Step 3 depending on whether $V_i$ is small or large.

For each $S \in \cS$, we make
\[
    z := \frac{50(\log n)\SL}{\eps}
\]
copies of $D_S$ and arrange them into a sequence $(D^{(1)}_S, D^{(2)}_S, \dots, D^{(z)}_S)$. Let $D'_S$ be the DAG projection of $G$ induced from this sequence.

Let
\[
    D' := \bigcup_{S \in \cS} D'_S.
\]
We make two copies of $D'$, denoted $D'_1, D'_2$, and let the final DAG projection $D$ be the DAG induced from $(D'_1, D'_2)$. We return $D$.

\paragraph{Base case.}
When $G$ has a constant number of vertices or $h$ is a constant, we return the DAG projection $D$ induced from
\[
(G^{(1)}, G^{(2)}, \dots, G^{(\max(|V(G)|, h))})
\]
(i.e. we repeat $G$ exactly $\max(|V(G)|, h)$ times).

\subsection{Analysis: Approximation}

It is immediate that $D$ is an DAG projection onto $G$, since $D$ is constructed only by applying \Cref{def:induced} to a sequence consisting of (1) DAG projections onto subgraphs of $G$ and (2) subgraphs of $G$. Thus, it remains to show that $D$ is an $h$-length $(1+\eps)$-distance-preserving DAG projection of $G$. By \Cref{lem:inheritedoneside} and \Cref{def:DAGembeddingDistance}, it is enough to prove that for every $s,t \in V(G)$ with $\dist_G(s,t) \le h$,
\[
\min \{ \dist_D(s',t') \mid \pi(s') = s,\ \pi(t') = t \} \;\le\; (1+\eps) \cdot \dist_G(s,t).
\]

We prove this by induction on $|V(G)|$ and on $h$. In the base case, when $|V(G)|$ is constant or $h$ is constant, any shortest path of length at most $h$ has at most $\max(|V(G)|, h)$ vertices (the path is simple and edge lengths are positive). Because we stacked $\max(|V(G)|, h)$ copies of $G$, the induced DAG contains a path from the copy of $s$ in the first layer to the copy of $t$ in the last layer with exactly the same length, so the base case holds.

\begin{lemma}[Induction]\label{lem:induction}
    Assume all recursive calls in $D \leftarrow \DDE(G,\eps,h)$ are correct. Then, with high probability, for every $s,t \in V(G)$ with $\dist_G(s,t) \le h$, we have
    \[
        \min \{ \dist_D(s',t') \mid \pi(s') = s,\ \pi(t') = t \} \;\le\; (1+\eps) \cdot \dist_G(s,t).
    \]
\end{lemma}

\begin{proof}
    Fix $s,t \in V(G)$ with $\dist_G(s,t) \le h$. We show the inequality holds w.h.p. for this pair; a union bound over all $s,t$ gives the lemma.

    Let $p$ be a shortest $s$–$t$ path in $G$. We need the following.

    \begin{lemma}\label{lem:numberofcrossing}
        With high probability, there exists $S \in \cS$ such that
        \begin{enumerate}
            \item every cluster in $S$ has diameter at most $\frac{\eps \cdot \ell_G(p)}{9 \log n}$, and
            \item $p$ contains at most $50 (\log n) \SL / \eps$ reversed edges in $S$ (see \Cref{def:ldd}).
        \end{enumerate}
    \end{lemma}
    \begin{proof}
        If $\ell_G(p) = O((\log n)\SL/\eps)$, then $|p| = O((\log n)\SL/\eps)$ (edge lengths are positive), and taking the LDD with diameter $1$ satisfies the claim: every cluster is a singleton, so diameter $0$, and the number of reversed edges is $O((\log n)\SL/\eps)$.

        Otherwise, $\ell_G(p) = \omega((\log n)\SL/\eps)$. Among the $O(\log n)$ diameter scales $2^i$ we took (up to at least $h \ge \ell_G(p)$), there is one scale $d$ in
        \[
            d \in \left[\frac{\eps \cdot \ell_G(p)}{18 \log n},\ \frac{\eps \cdot \ell_G(p)}{9 \log n}\right].
        \]
        For that scale we sampled $O(\log n)$ independent LDDs. By \Cref{def:ldd}, the expected number of reversed edges on $p$ in such an LDD is at most
        \[
            \SL \cdot \frac{\ell_G(p)}{d} \;\le\; 36 (\log n)\SL/\eps.
        \]
        By Markov’s inequality and independence across the $O(\log n)$ trials at that scale, w.h.p. one of them has at most $50 (\log n)\SL/\eps$ reversed edges and diameter at most $\eps \cdot \ell_G(p)/(9 \log n)$.
    \end{proof}

    Let $S \in \cS$ be as in \Cref{lem:numberofcrossing}, and let $d \le \eps \cdot \ell_G(p)/(9 \log n)$ be its diameter parameter. Let
    \[
        (C_1, C_2, \dots, C_q)
    \]
    be the sequence of clusters of $S$ that $p$ visits, in order. Let $p_i$ be the subpath of $p$ inside $C_i$. By \Cref{lem:numberofcrossing}, there are at most
    \[
        q_{\text{rev}} := 50 (\log n)\SL/\eps
    \]
    reversed edges along $p$. For each $p_i$, define $\Pre(p_i)$ to be the number of reversed edges on $p$ before $p_i$, and $\Dec(p_i)$ the number after $p_i$.

    Let $C_x$ be the first large cluster along this sequence, and $C_y$ the last large cluster (it is possible that $C_x = C_y$; the one-large-cluster case is handled similarly, so we assume $x < y$ for clarity). Let
    \[
        p_{\mathrm{mid}} := p_x \oplus p_{x+1} \oplus \dots \oplus p_y
    \]
    be the portion of $p$ from the first to the last large cluster.

    For every $i < x$, $C_i$ is a small cluster. By the induction hypothesis, $D_{C_i}$ contains a path $p'_i$ that projects to $p_i$ and
    \[
        \ell_{D_{C_i}}(p'_i) \le \bigl(1 + (1 - \tfrac{1}{\log n})\eps\bigr) \cdot \ell_G(p_i).
    \]
    Recall that in Step 4 we made
    \[
        z := \frac{50 (\log n)\SL}{\eps}
    \]
    copies of $D_S$, so we can place $p'_i$ in the $\Pre(p_i)$-th copy of $D_S$. We do the symmetric construction for every $i > y$ (these go into the second big union). Hence we can replace every $p_i$ with $i<x$ or $i>y$ by a corresponding $p'_i$ with only a factor $(1 + (1 - 1/\log n)\eps)$ blow-up, and the edges linking $p'_i$ to $p'_{i+1}$ exist in $D$:
    \begin{itemize}
        \item if the edge was reversed, we put $p'_{i+1}$ in a later copy, so the induced construction adds the connecting edge;
        \item if it was not reversed, the edge is already in the same copy.
    \end{itemize}

    It remains to replace $p_{\mathrm{mid}}$ in $D$. By construction of $D_{C_x}$ (large cluster), we have a reversed SPT $\Ti_{C_x}$, so there is a path
    \[
        p'_x : \Start(p_{\mathrm{mid}}) \to r_{C_x}
    \]
    in $D_{C_x}$. Its projection is the corresponding path in $G$, and since the cluster diameter is at most $d \le \eps \ell_G(p)/(9 \log n) \le h$, and $\Ti_{C_x}$ is a $(1 + (1 - 5/\log n)\eps)$-approximation,
    \[
        \ell(p'_x) \le (1 + (1 - 5/\log n)\eps) \cdot d \le \frac{\eps \cdot \ell_G(p)}{8 \log n}.
    \]
    Similarly, in $\To_{C_y}$ we have a path
    \[
        p'_y : r_{C_y} \to \End(p_{\mathrm{mid}})
    \]
    with the same kind of bound.

    Between the two large clusters, in $\To_{C_x}$ we have a path
    \[
        p'_{\mathrm{mid}} : r_{C_x} \to r_{C_y}
    \]
    that is a $(1 + (1 - 5/\log n)\eps)$-approximate shortest path in $G$. Also,
    \[
        \dist_G(r_{C_x}, r_{C_y}) \le \ell_G(p_{\mathrm{mid}}) + 2d
    \]
    because we can go from $r_{C_x}$ to the entry of $p_{\mathrm{mid}}$ in $C_x$ in at most $d$, follow $p_{\mathrm{mid}}$, then go from the exit in $C_y$ to $r_{C_y}$ in at most $d$.

    We place $p'_x$ and $p'_{\mathrm{mid}}$ in the $\Pre(p_x)$-th copy of $D_S$ in the first big union, and $p'_y$ in the $\Pre(p_y)$-th copy in the second big union. By the induced construction, these connect to form a path
    \[
        p_{\mathrm{mid}}'' := p'_x \oplus p'_{\mathrm{mid}} \oplus p'_y
    \]
    in $D$. Its length satisfies
    \[
        \ell_D(p_{\mathrm{mid}}'') \le 2 \cdot \frac{\eps \cdot \ell_G(p)}{8 \log n} + (1 + (1 - 5/\log n)\eps) \cdot (\ell_G(p_{\mathrm{mid}}) + 2d)
    \]
    \[
        \le (1 + (1 - 5/\log n)\eps) \cdot \ell_G(p_{\mathrm{mid}}) + \frac{\eps \cdot \ell_G(p)}{2 \log n}.
    \]

    Finally, concatenating the replaced prefix, the middle $p_{\mathrm{mid}}''$, and the replaced suffix gives a path in $D$ from some $s'$ with $\pi(s') = s$ to some $t'$ with $\pi(t') = t$ whose length is at most
    \[
        \bigl(1 + (1 - \tfrac{1}{\log n})\eps \bigr) \cdot (\ell_G(p) - \ell_G(p_{\mathrm{mid}}))
        \;+\; (1 + (1 - \tfrac{5}{\log n})\eps) \cdot \ell_G(p_{\mathrm{mid}})
        \;+\; \frac{\eps \cdot \ell_G(p)}{2 \log n}
        \;\le\; (1+\eps) \cdot \ell_G(p),
    \]
    for $n$ sufficiently large. This proves the inductive step.
\end{proof}

\subsection{Analysis: Size and Complexity}

\paragraph{Width of the DAG projection.}
The width of the DAG projection is bounded as follows.

\begin{lemma}\label{lem:smallduplication}
    Suppose $D \leftarrow \DDE(G,\eps,h)$ where $G$ has $n$ vertices. Then $D$ has width
    \[
        f(n,\eps) = \bigl(\tO{1/\eps}\bigr)^{\log^\delta n} \cdot n^{1 / \log^\delta n}
    \]
\end{lemma}
\begin{proof}
    The base case is trivial since each vertex is duplicated only a constant number of times.

    Recall the structure of $D$. The final DAG $D$ is induced from two copies of $D'$. Each $D'$ is the union (over $S \in \cS$) of $D'_S$. Each $D'_S$ is obtained from
    \[
        z := 50 (\log n)\SL / \eps
    \]
    copies of $D_S$. Each $D_S$ contains one copy of $D_C$ for every cluster $C$ in $S$.

    For a cluster $C$ we have two cases.
    \begin{itemize}
        \item If $C$ is a \emph{small} cluster, then by the induction hypothesis every vertex in $C$ is duplicated
        \[
            f\bigl(|C|, (1 - 1/\log n)\eps\bigr)
        \]
        times in $D_C$.
        \item If $C$ is a \emph{large} cluster, then we build $D_C$ just from the two trees, so every vertex of $G$ participates in $D_C$ at most twice.
    \end{itemize}

    Large clusters have size at least $n/\sigma$ and are disjoint, so there are at most $\sigma$ large clusters. Thus, across all large clusters, each vertex is duplicated at most $2\sigma$ times. Small clusters are disjoint and have size at most $n/\sigma$, so across all small clusters, each vertex is duplicated at most
    \[
        f\bigl(n/\sigma, (1 - 1/\log n)\eps\bigr)
    \]
    times.

    We have $O(\log^2 n)$ LDDs in $\cS$ (because we take $O(\log n)$ scales and $O(\log n)$ independent LDDs per scale). Every such LDD is blown up by a factor of $O((\log n)\SL/\eps)$. Putting this together, we obtain the recursion
    \[
        f(n,\eps)
        \;\le\;
        O\bigl((\log n)^5 / \eps\bigr)
        \cdot
        \bigl(
            f\bigl(n/\sigma, (1 - 1/\log n)\eps\bigr)
            + 2\sigma
        \bigr).
    \]
    Taking
    \[
        \sigma = 2^{\log^{1-\delta} n}
    \]
    and unwinding for at most $\log^\delta n$ levels (since $n \mapsto n/\sigma$ shrinks by $2^{\log^{1-\delta} n}$ each time), we get
    \[
        f(n,\eps)
        =
        \bigl(\tO{1/\eps}\bigr)^{\log^\delta n} \cdot n^{1 / \log^\delta n}.
    \]
\end{proof}

\paragraph{Complexity.}
Let $T(m,n,\eps,h)$ be the work of $\DDE(G,\eps,h)$ on a graph $G$ with $m$ edges and $n$ nodes. Let $\cO_{\mathrm{DAG}}(m)$ denote the work of one call to the DAG-SSSP oracle on a graph with $m$ edges.

Step 1 runs $O(\log^2 n)$ LDD constructions. By \Cref{lem:ldd} and \Cref{lem:reducingSSSPtoDAGs}, this costs
\[
    \tO{1} \text{ calls to } \apxDSSSP{(1+\eps/\log n)} \text{ on graphs with } O(\lambda w m) \text{ edges,}
\]
plus
\[
    \tO{1} \text{ recursive calls } T(m,n,\eps,h/\lambda),
\]

Step 2 makes recursive calls on all small clusters. All such calls are in parallel. Since the clusters of a fixed LDD partition $V$, and we have only $O(\log^2 n)$ LDDs, we have
\[
    \sum_{C \in S,\, S \in \cS} |E(G[C])| \;\le\; O(m \log^2 n).
    \]
So we can upper-bound Step 2 by
\[
    O(\log^2 n) \cdot T\bigl(m, n/\sigma,(1 - 1/\log n)\eps, nW\bigr).
    \]
If these recursive calls internally make more oracle calls, we can batch them (by taking the union of the graphs and adding a super source), so the stated number of oracle calls in \Cref{lem:distanceDAGembeddingtoSSSP} still holds.

Step 3 makes one recursive call
\[
    T\bigl(m, n,(1 - 10/\log n)\eps, h/\lambda\bigr)
\]
and, since there are at most $\sigma$ large clusters, makes $O(\sigma)$ calls to the $\apxDSSSP{(1+\eps/\log n)}$ oracle on graphs with $\hO{m}$ edges.

Step 4 costs time proportional to the size of the final DAG, which is $O(w \cdot m)$ by \Cref{lem:smallduplication}.

Recall that  $\lambda=\sigma=2^{\log^{1-\delta}n}$.
Putting these together, we get the recurrence
\[
    T(m,n,\eps,h)
    \;\le\;
    \tO{\,T(m,n,\eps,h/\lambda)\;}
    + \sum_{i}\tO{\,T\bigl(m_i, n/\sigma,(1 - 1/\log n)\eps, nW\bigr)\;}
    + \cO_{\mathrm{DAG}}(O(\lambda w m)).
\]

Where $m_i$ is the number of edges in the $i$-th small cluster. Observe that, since $\lambda=\sigma=2^{\log^{1-\delta}n}$, the recursion depth is at most $\log^{2\delta} n$ for expanding both the $n$ part and $h$ part, so the parameter $\eps$ passed down the recursion never shrinks below
\[
        (1 - 1/\log n)^{\log^{2\delta} n} \eps \;\ge\; \eps/2
\]
for $n$ large enough. Moreover, the total number of oracle calls is
\[
    \bigl(\tO{1}\bigr)^{\log^{2\delta} n} = \hO{1},
\]
and each is on a graph with
\[
    \tO{\lambda w m} = \hO{w m}
\]
edges (since $w$ is a function of $n,\eps$ and $\eps$ stays within a constant factor of the root value). Thus the total extra work from oracle calls is
\[
    \bigl(\tO{1}\bigr)^{\log^{2\delta} n} \cdot \lambda w m = \hO{w m}.
    \]

Finally, since the recursion depth is only $\log^{2\delta} n$ and all work inside a level can be parallelized, the overall depth is $\tO{1}$.

\section{Congestion DAG Projection Construction}\label{sec:congestionDAGembedding}
In this section, we provide an efficient parallel algorithm for constructing congestion DAG projections using oracle calls to only approximate max-flow/min-cut on DAGs.

\congestionDAGembedding*

\paragraph{High-Level Strategy.}

Assuming an oracle for $n^{o(1)}$-approximate MFMC on DAGs, our strategy for constructing a $n^{o(1)}$-congestion-preserving DAG projection is as follows.

Similar to the distance-preserving DAG projection case, the difficulty is that, if we follow the idea from the overview, then an algorithm for DAG projection requires solving max flow. Our idea is to define  a relaxed notion of $(\kappa,\delta)$-congestion-preserving DAG projection as follows.

\begin{definition}
    Let $D$ be a DAG projection of $G = (V,E)$. A \emph{$(\kappa,\delta)$-congestion-preserving projection algorithm} associated with $D$ takes as input either
    \begin{itemize}
        \item a flow $\ff^D$ in $D$ with $\pi(\Supp(\Dem(\ff^D))) \subseteq V$, and returns a flow $\ff$ in $G$ with congestion at most $\kappa \cdot \Cong(\ff^D)$ such that
        \[
            \Dem(\ff) \preceq \pi(\Dem(\ff^D))
            \quad\text{and}\quad
            \Val(\ff) \ge \Val(\ff^D) -\delta,
        \]
        or
        \item a cut $S^D$ in $D$, and returns a cut $S$ in $G$ with value at most $\Val(S^D)$ such that
        \[
            \HL(S^D)\setminus \{\bot\} \subseteq S \subseteq \pi(S^D).
        \]
    \end{itemize}
    The projection algorithm is \emph{efficient} if its complexity is at most the complexity of constructing $D$.
\end{definition}

When $\delta$ is tiny enough, it is the $\kappa$-congestion-preserving DAG projection we need.

We first observe that, when $\delta$ is very big, the problem becomes trivial. In particular, $(\kappa=1,\delta\ge U(E))$-congestion-preserving DAG projection means the flow projection algorithm does not need to return any flow, so a DAG with two copies of $V$ (denoted by $V_1,V_2$), and a center vertex $c$ with edge capacities $\deg_G(v)$ connecting from each copies of $v$ in $V_1$ to $c$; connecting to each copy of $v$ in $V_2$ from $c$, would work as a DAG satsfying the cut projection.

Our main technical contribution is to show that it is possible to reduce $(\kappa,\delta)$-congestion-preserving DAG projection to a larger $\delta$ DAG projection. In particular, we will show the following spiral recursion.

\begin{itemize}
\item We will show how to get $(\kappa,\delta)$-congestion-preserving DAG projection using roughly $(\kappa,\delta)$-approximate max flow by following the idea form the overview.
\item Then we will show how to get $(\kappa,\delta)$-approximate max flow by roughly $(\kappa,\delta\cdot z)$-approximate DAG projection where $z=n^{o(1)}$ is an appropriate parameter. This step requires concatenating many copies of $(\kappa,\delta\cdot z)$-approximate DAG projection to get a $(\kappa,\delta)$-approximate DAG projection.

\end{itemize}

\paragraph{Organization.} As preliminaries, we first define Expander Decomposition and Hierarchy in \Cref{subsec:ED}. 
To implement the main technical contribution that gives the above spiral reduction, we give the following chain of reduction:
\begin{itemize}
\item DAG projection to Expander Hierarchy (\Cref{subsec:CongDAG from EH}). The key idea for reducing the additive approximation $\delta$ by making copies is presented here.
\item Expander Hierarchy to Expander Decomposition (\Cref{subsec:EH}).
\item Expander Decomposition to approximate max flow (\Cref{subsec:ED from MFMC})
\item Approximate max flow on general graphs to approximate max flow on DAGs using DAG projection (\Cref{subsec:MFMC to DAG}).
\end{itemize}

And in the end, we combine everything together and finish the spiral reduction in \Cref{subsec:cong-complete}.

\subsection{Expander Decomposition and Hierarchy}
\label{subsec:ED}
Although we did \emph{strong expander decomposition} in the intro in the sense that edges are expanding in each induced subgraph of SCCs, we will do \emph{weak expander decomposition}, which only guarantees expanding in the whole graph in this section, because it admits a simpler algorithm.

\paragraph{Basic definitions of ordered vertex-partition.}
Let $G = (V,E)$ be a directed graph. For a sequence of vertex sets $\cV = (V_1, \dots, V_z)$ that partitions $V$ (which we call an ordered vertex-partition), define
\[
    \Erev(\cV) = \{ (u,v) \in E \mid u \in V_i,\ v \in V_j,\ j < i \}
\]
to be the set of \emph{reversed edges} with respect to $\cV$.

For two sequences $(V_1, \dots, V_x)$ and $(V'_1, \dots, V'_y)$ that both partition $V$, we say $(V_1, \dots, V_x)$ \emph{refines} $(V'_1, \dots, V'_y)$ if we can write
\[
    (V_1, \dots, V_x) = (\cC_1, \dots, \cC_y),
\]
where each $\cC_i$ is a subsequence of $(V_1, \dots, V_x)$ that partitions $V'_i$.

Given $G=(V,E)$ and a ordered vertex-partition $\cV = (V_1,\dots,V_z)$, we say a demand $(\DDelta,\nnabla)$ is \emph{$\cV$-constrained} if for every $V_i$,
\[
    \sum_{v \in V_i} \DDelta(v) \;=\; \sum_{v \in V_i} \nnabla(v),
\]
i.e. each part is individually balanced.

\begin{definition}[Terminal Expanding]\label{def:terminalexpanding}
    Let $G=(V,E)$ be a directed graph with edge capacities, let $\dd : V \to \bbR_{\ge 0}$, and let $\cV$ be a ordered vertex-partition. We say $\dd$ is \emph{$\cV$-constrained $\phi$-expanding on $G$} if every $\dd$-respecting, $\cV$-constrained demand can be routed in $G$ with congestion at most $1/\phi$.

    We say an edge set $F \subseteq E$ is \emph{$\cV$-constrained $\phi$-expanding on $G$} if $\vol_F$ is $\cV$-constrained $\phi$-expanding on $G$.
\end{definition}

\begin{definition}[Expander Decomposition]\label{def:terminalexpanderdecomposition}
    Let $G=(V,E)$ be a directed graph and $F \subseteq E$ an edge set. A \emph{(weak) $F$-expander decomposition} with expansion $\phi$ and slack $\Kap$ returns a ordered vertex-partition $\cV$ such that
    \begin{itemize}
        \item $F$ is $\phi$-expanding in the graph $G - \Erev(\cV)$, and
        \item $\vol(\Erev(\cV)) \le \Kap \phi \cdot \vol(F)$.
    \end{itemize}
    An associated \emph{flow routing algorithm} takes a $\cV$-constrained, $\vol_F$-respecting demand and outputs a flow in $G$ routing that demand. It is \emph{efficient} if its complexity is at most the complexity of constructing the decomposition.
\end{definition}

\begin{definition}[Expander  Hierarchy]\label{def:expanderdecompositionhierarchy}
    Let $G=(V,E)$ be a directed graph with edge capacities. An \emph{expander hierarchy} with $t$ layers consists of edge sets $E_1, \dots, E_t$ and ordered vertex-partitions $\cC_0, \dots, \cC_t$ such that\footnote{This deviates slightly from the classical definition in that the $\cC_i$’s do not have to be the SCCs of $G - E_{>i}$, and we assume the sequences $\cC_i$ are given. This avoids explicitly computing SCCs in parallel.}
    \begin{itemize}
        \item $\cC_0 = (\{v_1\}, \{v_2\}, \dots, \{v_n\})$ consists only of singletons, and $\cC_t = (V)$,
        \item for every $i \in [t]$, the sequence $\cC_{i-1}$ refines $\cC_i$,
        \item for every $i \in [t]$, we have $\Erev(\cC_{i-1}) \subseteq E_{\ge i}$, where $E_{\ge i} := \bigcup_{j \ge i} E_j$.
    \end{itemize}
    We say the hierarchy has expansion $(\phi_1, \dots, \phi_t)$ if, for every $1 \le i \le t$, the edge set $E_i$ is $\cC_i$-constrained $\phi_i$-expanding on $G$.

    An \emph{flow routing algorithm} associated with the hierarchy takes as input a layer number $i$ and a $\cC_i$-constrained, $\vol_{E_i}$-respecting demand, and outputs a flow in $G$ routing the demand with congestion at most $1/\phi_i$. It is \emph{efficient} if its complexity is at most the complexity of constructing the hierarchy.
\end{definition}

\subsection{Congestion DAG Projections via  Expander  Hierarchy}
\label{subsec:CongDAG from EH}
In this section, we assume we are given an expander hierarchy and we show how to construct a congestion DAG projection.

\begin{lemma}\label{lem:expanderdecompositiontoDAGembedding}
    There is an algorithm that takes as input
    \begin{itemize}
        \item a directed graph $G = (V,E)$ with edge capacities,
        \item an oracle that constructs an expander hierarchy of $G$ with $t$ layers and expansion $(\phi_1,\phi_2,\dots,\phi_t)$ where $\phi_i = \phi$ for all $i \le t-1$,\footnote{We do not care about the expansion for the last layer, because the flow routing there is contributed to the additive error.} together with an efficient flow routing algorithm, and
        \item a parameter $\sigma \ge 1$, which controls the trade-off between the projection size and additive-approximation,
    \end{itemize}
    and returns a DAG projection $D$ of $G$ such that
    \begin{itemize}
        \item $|E(D)| = O(2^t \sigma \cdot |E(G)|)$,
        \item $D$ is $(\kappa,\delta)$-congestion-preserving with
        \[
            \kappa = O\bigl(2^t / \phi\bigr)
            \qquad\text{and}\qquad
            \delta = U_G(E_t)/\sigma,
        \]
        \item and $D$ has an efficient projection algorithm.
    \end{itemize}
    The algorithm uses one call to the hierarchy oracle and performs additional $\tO{|E(D)|}$ work and $\tO{1}$ depth.

    Moreover, if the edge capacities of $G$ are integral and $\sigma = 1$, then the edge capacities of $D$ are integral.
\end{lemma}

\paragraph{Algorithm (constructing $D$).}
Let $\cC_0, \dots, \cC_t$ and $E_1, \dots, E_t$ be the expander hierarchy returned by the oracle.
Define
\[
    G_i := G - E_{>i} \qquad\text{for } i = 0,1,\dots,t,
\]
so $G_t = G$ (since $E_{>t} = \emptyset$). For every $i = 0,1,\dots,t$ and every cluster $C \in \cC_i$, we will build a DAG projection for $G_i[C]$, denoted $D_i(C)$. Since $\cC_t = (V)$, $D_t(V)$ is a DAG projection for $G$, and we will return $D := D_t(V)$.

For every original vertex $v$ in $G_i[C]$, we will specify two special copy vertices of $v$ in $D_i(C)$, denoted
\[
    \In{v}_i, \Out{v}_i \in V(D_i(C)),
\]
with $\pi(\In{v}_i) = \pi(\Out{v}_i) = v$. We will later see that $\In{v}$ and $\Out{v}$ are basically the first and last copy of vertex $v$ in the topological order of the DAG projection. Intuitively, these are the attachment points when we later connect edges that go in or out of $C$. Moreover, in the construction, there will be an infinite capacity path from $\In{v}_i$ to $\Out{v}_i$. 

\emph{Base layer.} When $i = 0$, each $C \in \cC_0$ is a singleton, say $C = \{v\}$. We let $D_0(C)$ be the one-vertex DAG with label $v$, and we set $\In{v}_0 = \Out{v}_0$ to be that vertex.

\medskip

\emph{Inductive step.} Suppose $i \ge 1$. We will define a parameter
\[
    \sigma_i =
    \begin{cases}
        1, & \text{if } i < t ,\\
        \sigma, & \text{if } i = t.
    \end{cases}
\]

Since $\cC_{i-1}$ refines $\cC_i$, the clusters of $\cC_{i-1}$ that are contained in $C \in \cC_i$ form a subsequence
\[
    (Y_1, \dots, Y_z)
\]
that partitions $C$. By the induction hypothesis, for each $Y_j$ we have already built a DAG projection $D_{i-1}(Y_j)$ of $G_{i-1}[Y_j]$.

\textbf{(1) Concatenating the child DAGs.}
We first build an intermediate DAG $\tilde D_i(C)$ by taking the DAGs
\[
    D_{i-1}(Y_1), D_{i-1}(Y_2), \dots, D_{i-1}(Y_z)
\]
in order, and for every edge $(u,v) \in E(G_i)$ with $u \in Y_x$, $v \in Y_y$, and $x < y$, we add to $\tilde D_i(C)$ the edge
\[
    \bigl(\Out{u}_{i-1},\, \In{v}_{i-1}\bigr)
\]
with capacity $U_G(u,v)$. This makes $\tilde D_i(C)$ a DAG projection of $G_i[C]$ consistent with the ordering induced by $(Y_1,\dots,Y_z)$.

\textbf{(2) Replicating to allow congestion projection.}
To form $D_i(C)$, we make $2\sigma_i$ copies of \emph{this} intermediate DAG $\tilde D_i(C)$, all with the same projection maps. For any vertex $x \in V(\tilde D_i(C))$, denote its copy in the $k$-th replica by $x_{(k)}$, for $1 \le k \le 2\sigma_i$.

We then add the following edges:

\begin{enumerate}[(i)]
  \item For every $u \in C$ and every $1 \le k < 2\sigma_i$, add an edge
  \[
     \bigl(\Out{u}_{i-1,(k)}, \In{u}_{i-1,(k+1)}\bigr)
  \]
  with \emph{infinite} capacity. This allows flow on $u$ to “walk” forward through all replicas.

  \item For every $(u,v) \in E_i$ with $u,v \in C$, and every $1 \le k < 2\sigma_i$, add an edge
  \[
    \bigl(\Out{u}_{i-1,(k)}, \In{v}_{i-1,(k+1)}\bigr)
  \]
  with capacity $U_G(u,v)$. This encodes the “fresh” edges of layer $i$ across replicas.

  \item Add a dummy vertex $w_i^C$, i.e., $\pi(w_i^C) = \bot$. For every $u \in C$, add edges
  \[
    \bigl(\Out{u}_{i-1,(\sigma_i)}, w_i^C\bigr)
    \quad\text{and}\quad
    \bigl(w_i^C, \In{u}_{i-1,(\sigma_i+1)}\bigr)
  \]
  each with capacity $\vol_{E_i}(u)$. This ensures we can “collect” and “redistribute” the $E_i$-volume of $u$ between the two “halves” of the replicas without violating capacities.
\end{enumerate}

Finally, we scale \emph{all} capacities in $D_i(C)$ by a factor of $1/\sigma_i$ so that the total capacity budget per original vertex is preserved across the $2\sigma_i$ replicas.
For every original vertex $u \in C$, we define two special copies in $D_i(C)$
\[
    \In{u}_i := \In{u}_{i-1,(1)}
    \qquad\text{and}\qquad
    \Out{u}_i := \Out{u}_{i-1,(2\sigma_i)}.
\]

Since $\cC_t = (V)$, the construction for $i = t$ produces $D_t(V)$, which we output as the final DAG projection $D$. This completes the construction of our congest DAG projection.

\subsubsection*{Analysis.}

The following lemma bounds the size of the DAG projection.

\begin{lemma}\label{lem:sizefromexpanderhiearchy}
    $D$ has size at most $O(2^t \sigma \, |E(G)|)$.
\end{lemma}
\begin{proof}
    For $i \ge 1$, each $D_i(C)$ (for $C \in \cC_i$) is composed of $2\sigma_i$ copies of $\tilde D_i(C)$, and $\tilde D_i(C)$ is the concatenation of $D_{i-1}(Y_j)$ over the partition $(Y_1,\dots,Y_z)$ of $C$. Thus,
    \[
        |E(D_i(C))|
        \;\le\;
        2\sigma_i \cdot \Bigl( |E(G[C])| + \sum_{j=1}^z |E(D_{i-1}(Y_j))| \Bigr).
    \]
    Summing this recurrence up the hierarchy (and using that $\sum_{C \in \cC_i} |E(G[C])| \le |E(G)|$ and $\sigma_i=\sigma$ for only layer $t$ while $\sigma_i=1$ for $i<t$) yields
    \[
        |E(D_t(V))| \;\le\; O(2^t \sigma \, |E(G)|).
    \]
\end{proof}

\paragraph{Projection algorithm (flow).}
For an edge $(u,v) \in E(G)$, define its \emph{$D$-congestion} to be
\[
    \sum_{(u',v') \in E(D)\,:\,\pi(u')=u,\ \pi(v')=v} U_D(u',v').
\]

\begin{lemma}\label{lem:lowcongestion}
    For every $(u,v) \in E(G)$, the $D$-congestion of $(u,v)$ is at most $2^t$.
\end{lemma}
\begin{proof}
    We prove by induction on $i$ that, for every $C \in \cC_i$, the $D_i(C)$-congestion of any edge $(u,v)$ is at most $2^i$.

    For $i = 0$ the claim is trivial, since $D_0(C)$ has no edges.

    For $i \ge 1$, $\tilde D_i(C)$ is a concatenation of $D_{i-1}(Y_j)$’s. Any edge of $G_i$ that crosses between two different $Y_j$’s is added at most once in $\tilde D_i(C)$, with its original capacity, so such edges contribute congestion $1$ at this level. Edges internal to a $Y_j$ come from $D_{i-1}(Y_j)$, and by the induction hypothesis they contribute at most $2^{i-1}$ per copy.

    Now $D_i(C)$ consists of $2\sigma_i$ copies of $\tilde D_i(C)$, and finally we scale down capacities by $\sigma_i$. Hence, for edges not in $E_i$, the total congestion is
    \[
        2^{i-1} \cdot \frac{2\sigma_i}{\sigma_i} = 2^i.
    \]
    For edges in $E_i$, we add one between every consecutive pair of copies, so there are $2\sigma_i$ of them, each scaled by $1/\sigma_i$, giving congestion $2$. Taking the maximum over these two cases yields the desired bound $2^i$.
\end{proof}

Next, we need to control the demands created at dummy vertices. For every $1 \le i \le t$ and $C \in \cC_i$, the construction creates many copies of the dummy vertex $w_i^C$. For example, $D_{i+1}$ makes $2\sigma_i$ copies of $w_i^C$. Even more copies are inductively created at higher levels. For such a $w_i^C$:
- each incoming edge $(u, w_i^C)$ gives $u$ a \emph{$w_i^C$-source-demand} equal to the capacity of that edge;
- each outgoing edge $(w_i^C, v)$ gives $v$ a \emph{$w_i^C$-sink-demand} equal to the capacity of that edge.

The \emph{$(C,i)$-source-demand} of $u$ is the sum of all $w_i^C$-source-demands over all dummy centers $w_i^C$ of clusters $C \in \cC_i$; the \emph{$(C,i)$-sink-demand} is defined symmetrically.

\begin{lemma}\label{lem:lowcongexpanderrouting}
    For all $1 \le i \le t$, all $u \in V$, and all $C \in \cC_i$, the $(C,i)$-source-demand (and also the $(C,i)$-sink-demand) of $u$ is at most
    \[
        2^{t - i} \cdot \frac{\vol_{E_i}(u)}{\sigma_i}.
    \]
\end{lemma}
\begin{proof}
    We prove the following slightly stronger statement by induction: for any $x \ge i$ and any cluster $C' \in \cC_x$, the total $(C,i)$-source-demand (or sink-demand) of $u$ in $D_x(C')$ is at most
    \[
        2^{x - i} \cdot \frac{\vol_{E_i}(u)}{\sigma_i}.
    \]

    When $x = i$, this is exactly how we defined the dummy edges at level $i$: the $(C,i)$-demand is $\vol_{E_i}(u)/\sigma_i$.

    For $x > i$, the DAG $D_x(C')$ is made from $2\sigma_x$ copies of $\tilde D_x(C')$, and each $\tilde D_x(C')$ is a concatenation of lower-level DAGs. For a fixed $u$, only one of those lower-level DAGs contains the copy of $u$, so by induction that copy contributes at most
    \[
        2^{x - 1 - i} \cdot \frac{\vol_{E_i}(u)}{\sigma_i}
    \]
    in each replica. After making $2\sigma_x$ replicas and scaling by $1/\sigma_x$, we get
    \[
        2^{x - 1 - i} \cdot \frac{\vol_{E_i}(u)}{\sigma_i} \cdot \frac{2 \sigma_x}{\sigma_x}
        =
        2^{x - i} \cdot \frac{\vol_{E_i}(u)}{\sigma_i},
    \]
    as desired.
\end{proof}

Now we are ready to state the flow projection algorithm.
Let $\ff^D$ be a flow in $D$ whose support does not contain dummy vertices (i.e. all demand is on original vertices). Consider a flow path $p^D$ of $\ff^D$ that does not pass through the top-level dummy vertex $w_t^V$. We decompose $p^D$ into subpaths of two types:

1. \textbf{Pure subpaths} $\tilde p$ that contain no dummy vertex. Every edge $(u',v')$ on such a subpath satisfies either $(\pi(u'),\pi(v')) \in E(G)$ or $\pi(u') = \pi(v')$. Thus $\pi(\tilde p)$ is a path in $G$ (after removing repetitions). We create a flow path in $G$ along $\pi(\tilde p)$ with the same flow value.

2. \textbf{Dummy vertex subpaths} of the form $(u, w_i^C, v)$, where $w_i^C$ is a dummy vertex. For each such subpath with value $\gamma$, we create a source demand $\gamma$ on $\pi(u)$ and a sink demand $\gamma$ on $\pi(v)$ to be later routed by the flow routing associated with layer $i$.

Doing this for all flow paths of $\ff^D$ that avoid $w_t^V$ gives us a partial flow $\ff$ in $G$ plus, for each layer $i$, a $\cC_i$-constrained demand $(\DDelta_i, \nnabla_i)$. The total flow in $\ff^D$ that \emph{does} go through $w_t^V$ is at most
\[
    \sum_{u \in V} \vol_{E_t}(u) = U_G(E_t),
    \]
because those edges were scaled by $1/\sigma_t = 1/\sigma$. Hence
\[
    \Dem(\ff) \preceq \pi(\Dem(\ff^D)) \quad\text{and}\quad \Val(\ff) \ge \Val(\ff^D) -\delta
\]
with $\delta \le U_G(E_t)/\sigma$, as claimed in Lemma~\ref{lem:expanderdecompositiontoDAGembedding}.

To route the dummy-hop demands, for each $i \in [t]$ we collect all tuples $(u, w_i^C, v)$ of value $c$ into a demand $(\DDelta_i,\nnabla_i)$. By \Cref{lem:lowcongexpanderrouting}, this demand is $(2^{t-i} \cdot \vol_{E_i}(u)/\sigma_i)$-respecting, and by construction it is $\cC_i$-constrained. We scale this demand down by a factor $\sigma_i / 2^{t-i}$ and apply the flow routing at layer $i$ (which has congestion at most $1/\phi_i$), and finally scale the routed flow back up by $2^{t-i}/\sigma_i$. This gives congestion at most
\[
    \frac{2^{t-i}}{\sigma_i \phi_i}
\]
for the layer-$i$ part.

The final flow in $G$ consists of:
- the “trivial” projection of edges, which by \Cref{lem:lowcongestion} has congestion at most $2^t$, and
- the expander-routed part, whose congestion is at most
\[
    \sum_{i=1}^t \frac{2^{t-i}}{\sigma_i \phi_i}
    \;\le\;
    O\bigl(2^t / \phi\bigr)
\]
since $\phi_i = \phi$ for $i \le t-1$ and $\sigma_i \ge 1$.

The second part dominates the first, so the total congestion is $O(2^t/\phi)$, as stated in Lemma~\ref{lem:expanderdecompositiontoDAGembedding}. The running time is also as claimed, since flow routing is assumed to be efficient and the remaining steps are linear in $|E(D)|$.

\paragraph{Projection algorithm (cut).}
For $0 \le i < t$, $C \in \cC_i$, and the DAG projection $D_i(C)$, define
\[
    \In V(D_i(C)) := \{\In{u}_i \in V(D_i(C)) \mid u \in C\}, \qquad
    \Out V(D_i(C)) := \{\Out{u}_i \in V(D_i(C)) \mid u \in C\}.
\]
We prove the following by induction.

\begin{lemma}\label{lem:backwardmappingcuthypothesis}
    For every $0 \le i < t$ and every $C \in \cC_i$, let $S^D$ be a cut in $D_i(C)$ with finite cut value. Then there exists a set $S_C \subseteq C$ such that
    \begin{itemize}
        \item $\pi\bigl(S^D \cap \In V(D_i(C))\bigr) \subseteq S_C \subseteq \pi\bigl(S^D \cap \Out V(D_i(C))\bigr)$,
        \item $\Val_{G_i[C]}(S_C) \le \Val_{D_i(C)}(S^D)$.
    \end{itemize}
    Moreover, $S_C$ can be found in $\tO{|D_i(C)|}$ work and $\tO{t}$ depth.
\end{lemma}

Once we have this lemma, the cut-projection algorithm follows immediately: for $i = t$ and $C = V$, we obtain a cut $S_C \subseteq V$ such that
\[
    \HL(S^D) \subseteq \pi(S^D \cap \In V(D_t(V))) \subseteq S_C \subseteq \pi(S^D \cap \Out V(D_t(V))) \subseteq \pi(S^D),
\]
and
\[
    \Val_G(S_C) \le \Val_D(S^D),
\]
as required. (If $S^D$ has infinite value, we can return any valid $s$–$t$ separating cut in $G$.) Thus it suffices to prove \Cref{lem:backwardmappingcuthypothesis}.

\begin{proof}[Proof of \Cref{lem:backwardmappingcuthypothesis}]
    \textbf{Base case.} For $i = 0$, each $C \in \cC_0$ is a singleton, so $D_0(C)$ has one vertex and every cut has value $0$. Taking $S_C = C$ or $S_C = \emptyset$ satisfies the statement.

    \medskip
    \textbf{Inductive step.} Let $i \ge 1$. The DAG $D_i(C)$ consists of $2\sigma_i$ copies of $\tilde D_i(C)$ plus one dummy vertex $w_i^C$. WLOG assume $w_i^C \notin S^D$; the other case is symmetric (we only mirror the choice of copies).

    Let $(Y_1, \dots, Y_z)$ be the subsequence of $\cC_{i-1}$ that partitions $C$. For $x \in [2\sigma_i]$, let $\tilde D_i^{(x)}(C)$ denote the $x$-th copy of $\tilde D_i(C)$. Recall that $\tilde D_i(C)$ is formed by concatenating $D_{i-1}(Y_1),\dots,D_{i-1}(Y_z)$ in order. Let $D_{i-1}^{(x)}(Y_y)$ denote the copy of $D_{i-1}(Y_y)$ inside $\tilde D_i^{(x)}(C)$.

    For each $x \in [2\sigma_i]$ and each $y \in [z]$, define
    \[
        S^D_{x,y} := S^D \cap V\bigl(D_{i-1}^{(x)}(Y_y)\bigr).
    \]
    By the induction hypothesis applied to level $i-1$, there is a set $S_{x,y} \subseteq Y_y$ such that
    \begin{equation}\label{eq:proj-in-out}
        \pi\bigl(S^D_{x,y} \cap \In V(D_{i-1}^{(x)}(Y_y))\bigr)
        \;\subseteq\;
        S_{x,y}
        \;\subseteq\;
        \pi\bigl(S^D_{x,y} \cap \Out V(D_{i-1}^{(x)}(Y_y))\bigr),
    \end{equation}
    and
    \begin{equation}\label{eq:proj-value}
        \Val_{G_{i-1}[Y_y]}(S_{x,y})
        \;\le\;
        \Val_{D_{i-1}^{(x)}(Y_y)}(S^D_{x,y}).
    \end{equation}
    Let
    \[
        S_x := \bigcup_{y=1}^z S_{x,y}.
    \]

    We will eventually pick one $S_x$ as $S_C$. To do that, define
    \begin{equation}\label{eq:choose-x}
        x^{*}=\arg\min_{x\in[\sigma_{i}]} \Val_{G_{i}[C]}(S_{x}) 
    \end{equation}
    where $E_i[C] := E_i \cap E(G_i[C])$. (If $w_i^C \in S^D$, we do the symmetric choice among $x \in \{\sigma_i+1, \dots, 2\sigma_i\}$.) Set $S_C := S_{x^*}$.

    \medskip
    \textbf{Property 1 (in/out sandwich).}
    We claim that for every $x \in [2\sigma_i]$ and every $y \in [z]$,
    \begin{equation}\label{eq:middle-ineq}
        \pi\bigl(S^D \cap \In V(D_{i-1}^{(x)}(Y_y))\bigr)
        \;\subseteq\;
        S_{x,y}
        \;\subseteq\;
        \pi\bigl(S^D \cap \Out V(D_{i-1}^{(x)}(Y_y))\bigr)
        \;\subseteq\;
        \pi\bigl(S^D \cap \In V(D_{i-1}^{(x+1)}(Y_y))\bigr),
    \end{equation}
    ignoring the last inclusion when $x = 2\sigma_i$. The first two inclusions are exactly \eqref{eq:proj-in-out}. For the last inclusion: if there were $u \in Y_y$ such that $\Out{u}_{i-1,(x)} \in S^D$ but $\In{u}_{i-1,(x+1)} \notin S^D$, then the edge
    \[
        \bigl(\Out{u}_{i-1,(x)}, \In{u}_{i-1,(x+1)}\bigr)
    \]
    of infinite capacity would be in the cut, contradicting the assumption that $S^D$ has finite value. Thus \eqref{eq:middle-ineq} holds.

    Taking the union over all $y$ and using \eqref{eq:middle-ineq} for consecutive copies, we obtain
    \[
        \pi\bigl(S^D \cap \In V(D_i(C))\bigr)
        \;\subseteq\;
        S_x
        \;\subseteq\;
        \pi\bigl(S^D \cap \Out V(D_i(C))\bigr)
    \]
    for every $x$. In particular, it holds for $x^*$, so Property 1 is proved.

    \medskip
    \textbf{Property 2 (cut value).}
    We need to compare $\Val_{G_i[C]}(S_C)$ to $\Val_{D_i(C)}(S^D)$. First we show a per-copy inequality.

    \begin{lemma}\label{lem:singlecopy}
        For every $x \in [2\sigma_i]$,
        \[
            \frac{1}{\sigma_i} \cdot \Val_{G_i[C] - E_i}(S_x)
            \;\le\;
            \Val_{\tilde D_i^{(x)}(C)}\bigl(S^D \cap V(\tilde D_i^{(x)}(C))\bigr).
        \]
    \end{lemma}
    \begin{proof}
        Let
        \[
            E_C^{\to} := \{ (u,v) \in E(G_i[C]) \mid u \in Y_{k_1}, v \in Y_{k_2}, k_1 < k_2 \}
        \]
        be the “forward” edges between different $Y$’s (recall $\Erev(\cC_{i-1}) \subseteq E_i$, so $G_i[C] - E_i$ only has forward inter-$Y$ edges). Then
        \[
            \Val_{G_i[C] - E_i}(S_x)
            \;\le\;
            \sum_{y=1}^z \Val_{G_{i-1}[Y_y]}(S_{x,y})
            \;+\;
            U_G\bigl(\{ (u,v) \in E_C^{\to} \mid u \in S_x, v \notin S_x \}\bigr).
        \]
        On the other hand, the cut of $S^D$ in $\tilde D_i^{(x)}(C)$ is exactly
        \[
            \sum_{y=1}^z \Val_{D_{i-1}^{(x)}(Y_y)}(S^D_{x,y})
            \;+\;
            U_{\tilde D_i^{(x)}(C)}\bigl(\{ (\Out{u}_{i-1,(x)}, \In{v}_{i-1,(x)}) \mid \Out{u}_{i-1,(x)} \in S^D,\, \In{v}_{i-1,(x)} \notin S^D \}\bigr).
        \]
        By \eqref{eq:proj-value}, the first part is at least
        \[
            \frac{1}{\sigma_i} \sum_{y=1}^z \Val_{G_{i-1}[Y_y]}(S_{x,y}),
        \]
        because capacities in $D_i(C)$ are scaled by $1/\sigma_i$. By \eqref{eq:proj-in-out} and the same “if $u\in S_x$ then its out-copy is in $S^D$, if $v\notin S_x$ then its in-copy is not in $S^D$” argument as before, each inter-$Y$ edge in $E_C^{\to}(S_x, C \setminus S_x)$ is cut in $\tilde D_i^{(x)}(C)$, again up to the $1/\sigma_i$ scaling. Putting these together gives the claim.
    \end{proof}

    Next, recall that $D_i(C)$ is formed from all $2\sigma_i$ copies plus the cross edges and the dummy edges. Let $E^D_{\mathrm{cross}}$ denote the set of edges in $D_i(C)$ that connect \emph{different} copies (including edges adjacent to $w_i^C$). We relate $E^D_{\mathrm{cross}}$ to the $E_i[C]$ part:

    \begin{lemma}\label{lem:Erevcutvalue}
        \[
            \sum_{x \in [\sigma_i]} \frac{1}{\sigma_i} \cdot \Val_{E_i[C]}(S_x)
            \;\le\;
            \Val_{E^D_{\mathrm{cross}}}(S^D).
        \]
    \end{lemma}
    \begin{proof}
        Use \Cref{eq:middle-ineq}. For each $x \in [\sigma_i]$, partition the edges in $E_i[C](S_x, C \setminus S_x)$ into:
        \begin{itemize}
            \item \emph{Type ``Bridge-paid'':} edges $(u,v)$ where $\In{v}_{i-1,(x+1)} \notin S^D$. Then $\Out{u}_{i-1,(x)} \in S^D$ by \Cref{eq:middle-ineq}, so the edge $(\Out{u}_{i-1,(x)}, \In{v}_{i-1,(x+1)})$ (a cross edge) is cut and can pay for this $(u,v)$.
            \item \emph{Type ``Terminal-paid'':} edges $(u,v)$ where $\In{v}_{i-1,(x+1)} \in S^D$. Then again by \Cref{eq:middle-ineq}, $\Out{v}_{i-1,(\sigma_i)} \in S^D$. Since we assumed $w_i^C \notin S^D$, the edge $(\Out{v}_{i-1,(\sigma_i)}, w_i^C)$ is cut and can pay for $(u,v)$. Moreover, this charging does not double-count edges to $w_i^C$: once $(u,v)$ is of Type ``Terminal-paid'' at level $x$, both endpoints remain inside all later $S_{x'}$, so it cannot be of Type ``Terminal-paid'' again.
        \end{itemize}
        In all cases, we pay for the original capacity only once, and then scale by $1/\sigma_i$ to match the scaled capacity in $D_i(C)$.
    \end{proof}

    Now we can finish. Let $x^*$ be as in \eqref{eq:choose-x}. Then

    \begin{align*}
\Val_{D_{i}(C)}(S^{D}) & =\Val_{E_{\mathrm{cross}}^{D}}(S^{D})+\sum_{x\in[2\sigma_{i}]}\Val_{\tilde{D}_{i}^{(x)}(C)}\bigl(S^{D}\cap V(\tilde{D}_{i}^{(x)}(C))\bigr)\\
 & \ge\sum_{x\in[\sigma_{i}]}\frac{1}{\sigma_{i}}\Val_{E_{i}[C]}(S_{x})+\sum_{x\in[\sigma_{i}]}\frac{1}{\sigma_{i}}\Val_{G_{i}[C]-E_{i}}(S_{x}) & \text{by \Cref{lem:singlecopy,lem:Erevcutvalue}}\\
 & =\frac{1}{\sigma_{i}}\sum_{x\in[\sigma_{i}]}\Val_{G_{i}[C]}(S_{x})\\
 & \ge\Val_{G_{i}[C]}(S_{x^{*}})=\Val_{G_{i}[C]}(S_{C}), & \text{as }x^{*}=\arg\min_{x\in[\sigma_{i}]} \Val_{G_{i}[C]}(S_{x}) 
\end{align*}

    The work/depth bound follows because at each level we only need to:
    \begin{itemize}
        \item apply the inductive procedure to $O(1)$ (in fact $O(2\sigma_i)$) sub-DAGs,
        \item sum/compare the cut values, and
        \item pick the best $x$,
    \end{itemize}
    all of which take time subsumed by constructing $D_i(C)$. Over $t$ levels, this gives $\tO{|D_i(C)|}$ work and $\tO{t}$ depth.
\end{proof}

\subsection{Expander Hierarchy via Expander Decomposition}\label{subsec:EH}

In the next section, we will prove the following expander decomposition lemma.

\newcommand{\Omfmc}{\cO_{\mathrm{MFMC}}}
\begin{restatable}{lemma}{expanderdecomposition}\label{lem:terminalexpanderdecomposition}
    There is an algorithm that takes as input
    \begin{itemize}
        \item a directed graph $G = (V,E)$ with edge capacities,
        \item an edge set $F \subseteq E$,
        \item a parameter $0 < \phi < 1$,
        \item an oracle $\Omfmc$ solving $\amfmc{(\alpha,\delta)}$,
    \end{itemize}
    and outputs an edge set $F' \subseteq F$ and an $F'$-expander decomposition of $G$ with expansion $\phi$ and slack $\Kap$, together with an associated efficient flow routing algorithm, such that
    \[
        \Kap = \alpha \log^6 n, \qquad
        \delta \;\le\; U_G(F)\cdot \frac{(\phi \Kap)^2}{99 \log^4 n}, \qquad
        U_G(F') \ge (1 - \phi \Kap)\cdot U_G(F),
    \]
    provided $\phi \le o(1/\Kap)$.
    The algorithm makes $\tO{1/(\phi \Kap)}$-efficient calls to $\Omfmc$ and uses additional $\tO{|E(G)|/(\phi \Kap)}$ work and $\tO{1/(\phi \Kap)}$ depth.
\end{restatable}

Now we can use \Cref{lem:expanderdecompositiontoDAGembedding} to get an expander hierarchy. We initialize $E'_1=E$ and $\cC'_0=(\{v_1\},\{v_2\},...,\{v_n\})$ where $(v_1,...,v_n)$ is an arbitrary order of $V$. We will show how to start from $E'_i,\cC'_i$ to compute $E_i,\cC'_{i+1},E'_{i+1}$. We should think of $E'_i$ as the final $E_{\ge i}$, and $\cC'_i$ will be processed in the end to get $\cC_i$ to make sure $\cC_i$ refines $\cC_{i+1}$. 

Let 
\[\Kap{}=\alpha\kappa'\log^6n\qquad \phi=2^{-\log^{0.5}n}/\gamma\]

\textbf{If} the following inequality is satisfied
\begin{equation*}
    \delta'\le  U_G(E'_i)\cdot (\phi\Kap{})^2/(99\log^4 n)
\end{equation*}

Then we apply \Cref{lem:terminalexpanderdecomposition} 
on $G,E'_i,\phi$ 
to get $E_i\subseteq E'_i$ and a $E_i$-expander decomposition of $G$ denoted by $\cC'_{i}$ with expansion $\phi$ and slack $\gamma$ where

\[U_G(E_i)\ge (1-\phi\Kap{})\cdot U_G(E'_i)\]

This call to \Cref{lem:terminalexpanderdecomposition} is valid by checking all the four inequalities described in \Cref{lem:terminalexpanderdecomposition}.

We let $E'_{i+1}=(E'_i\backslash E_i)\cup \Erev(\cC'_i)$.

\textbf{Otherwise} we have

\begin{equation}\label{eq:endinghierarchy}
    \delta'>  U_G(E'_i)\cdot (\phi\Kap{})^2/(99\log^4 n)= U_G(E'_i)\cdot 4^{-\log^{0.5}n}/(99\log^4 n)
\end{equation}

In this case, we \textbf{stop} and let $t=i$ and $E_t=E'_i,\cC_t=(V)$
We process $\cC'_i$ to get $\cC_i$ from $i=t-1$ to $i=1$ in the following way (to make sure $\cC_j$ refines $\cC_{j+1}$): suppose $\cC'_i=(C_1,...,C_z)$ and $\cC_{i+1}=(C'_1,...,C'_{z'})$. We define $C_{x,y}=C'_x\cap C_y$. We let $\cC_i=(C_{1,1},C_{1,2},...,C_{1,z},C_{2,1},...,C_{z',z})$, i.e., order first according to $x$, then to $y$. Notice that this step can be done in parallel $\tO{m}$ work and $\tO{1}$ step by parallel sorting. 

\begin{lemma}\label{lem:correctnessexpanderhierarchy}
    $E_1,...,E_t$ and $\cC_0,...,\cC_t$ is a valid expander hierarchy with expansion $(\phi_i)_{i\in [t]}$ where $\phi_i=\phi$ for all $i<t$ and
    \[U_G(E_t)<\delta'\cdot 2^{3\log^{0.5}n}\]

    It is associated with an efficient projection algorithm.
\end{lemma}
\begin{proof}
    The inequality for $U_G(E_t)$ is from \Cref{eq:endinghierarchy} by taking $n$ to be sufficiently large.

    It is clear from the definition of $(\cC_i)_{i\in[t]}$ that $\cC_i$ refines $\cC_{i+1}$ for every $0\le i\le t-1$.

    Then we prove that $\Erev(\cC_{i-1})\subseteq E_{\ge i}$ for every $i\in[t]$. According to the definition if $E'_i$, we always have $\Erev(\cC'_{i-1})\subseteq E'_{i}$ and $E'_i\subseteq E_{\ge i}$. Moreover, according to the definition of $\cC_{i-1}$, we must have $\Erev(\cC_{i-1})\subseteq \Erev(\cC'_{i-1})\cup \Erev(\cC_{i})$. Expanding repeatedly for $\Erev(\cC_i)$ gives us $\Erev(\cC_{i-1})\subseteq \cup_{j\ge i}\Erev(\cC'_j)\subseteq E_{\ge i}$.

    Then we prove the expanding guarantee. According to the definition of $E_i$-expander decomposition, we have that $E_i$ is $\cC'_i$-constraint $\phi$-expanding in $G$. Notice that $\cC_i$ refines $\cC'_i$, so any $\cC_i$-constraint demand is also a $\cC'_i$-constraint demand, so $E_i$ is also $\cC_i$-constraint $\phi$-expanding in $G$.

    At last, the projection algorithm for each layer of the hierarchy is efficient according to \Cref{lem:terminalexpanderdecomposition}, combining them gives an efficient projection algorithm for the whole hierarchy.
\end{proof}
Notice that $t=O(\log^{0.5}n)$ since 
\[U_G(E'_{i+1})\le 2\phi\Kap{}\cdot U_G(E'_i)=2\cdot 2^{-\log^{0.5}n}\cdot U_G(E'_i)\]

and we assume the capacities are polynomially bounded.

\subsection{Expander Decomposition via MFMC Oracle}
\label{subsec:ED from MFMC}
In this section, we describe how to compute an expander decomposition given access to a congestion DAG projection. By \Cref{lem:DAGembeddingtoMaxFlowadditiveerror}, a suitable DAG projection already gives such an oracle, so it is convenient to phrase the lemma directly in terms of an MFMC oracle.

\expanderdecomposition*
The algorithm follows the same high-level structure as Section 7 of \cite{BBLST25} (a directed non-stop cut–matching framework), but we restate the pieces here for completeness.

\begin{definition}[Matching]\label{def:matching}
    Let $(P,Q)$ be a partition of $V$. A set of directed edges $M \subseteq P \times Q$ with capacities $U_M : M \to \bbR_{\ge 0}$ is called a \emph{$(P,Q)$-matching}. For a demand bound $\dd : V \to \bbR_{\ge 0}$ and an error $\eps \ge 0$, we say $M$ is \emph{$(\dd,\eps)$-perfect} if
    \begin{enumerate}
        \item $\vol_M \le \dd$ (i.e. $M$ does not send/receive more than $\dd(v)$ out of any vertex), and
        \item $\vol_M(V) \ge (1 - \eps)\cdot \dd(V)$.
    \end{enumerate}
\end{definition}

We will use the  “non-stop” cut–matching game for directed graphs from \cite{abs-2507-09729}, which is an extension of the ``non-stop'' cut–matching game for undirected graphs by 
\cite{racke2014computing,saranurak2019expander}

\newcommand{\Omatch}{\cO_{\mathrm{match}}}
\begin{lemma}[\cite{abs-2507-09729}]\label{lem:cutmatchinggame}
    Let $\dd : V \to \bbN$ be polynomially bounded and let $\eps = o(1/\log^2 n)$. Suppose we have an oracle $\Omatch$ that, for any partition $(P,Q)$ of $V$ and any $\dd' \le \dd$ with $\dd'(P) = \dd'(Q)$ and $\dd'(V) \ge \dd(V)/2$, returns a $(\dd',\eps)$-perfect $(P,Q)$-matching.

    Then there is a randomized algorithm that makes $O(\log^2 n)$ calls to $\Omatch$ and $O(m \log^2 n)$ additional work with $\tO{1}$ depth, and produces $\tilde{\dd} \le \dd$ with
    \[
        \tilde{\dd}(V) \ge \bigl(1 - O(\eps \log^2 n)\bigr) \cdot \dd(V)
    \]
    such that, with high probability, $\tilde{\dd}$ is $\Omega(1)$-expanding in the graph
    \[
        W := \bigcup_i M_i,
    \]
    where $M_i$ is the matching returned by the $i$-th call to $\Omatch$.
\end{lemma}

Intuitively, \Cref{lem:terminalexpanderdecomposition} will simulate the matching oracle $\Omatch$ using the MFMC oracle $\Omfmc$: each time we need a nearly perfect matching across a cut, we phrase it as a flow instance with capacities restricted to $F$, call $\Omfmc$, and either (i) get the matching we wanted, or (ii) discover a cut with small $F$-capacity and peel it off, charging its capacity to the slack. Repeating this for $O(1/(\phi \Kap))$ rounds preserves almost all of $F$ (the $F'$ guarantee) and yields an $F'$-expanding decomposition, exactly as in \cite{BBLST25}, but now with the additive loss $\delta$ inherited from the MFMC oracle.

\newcommand{\ED}{\mathsf{ED}}
\paragraph{Expander Decomposition Algorithm.}
The algorithm for \Cref{lem:terminalexpanderdecomposition} is recursive. We write
\[
    (F', \cC) \leftarrow \ED(G, F, \phi)
\]
to denote a subroutine that should satisfy the guarantees of \Cref{lem:terminalexpanderdecomposition}: namely, $\cC$ is an $F'$-expander decomposition of $G$ with expansion $\phi$ and slack $\Kap$, and $F' \subseteq F$ preserves almost all of $F$.

We set
\[
    \dd := \vol_F
    \qquad\text{and}\qquad
    \eps := \frac{\phi \Kap}{\log^3 n}.
\]
We want to run the cut–matching game of \Cref{lem:cutmatchinggame}, so we must implement the matching oracle $\Omatch$.

\medskip
\textbf{The oracle \(\Omatch\).}
An oracle call receives
\begin{itemize}
    \item a partition $(P,Q)$ of $V$,
    \item a vector $\dd' \le \dd$ such that $\dd'(P) = \dd'(Q)$ and $\dd'(V) \ge \dd(V)/2$.
\end{itemize}
We maintain:
\begin{itemize}
    \item a flow $\ff^*$ (initially empty), which will accumulate routed demand, and
    \item a residual demand $\dd'' \le \dd'$ (initially $\dd'' := \dd'$).
\end{itemize}
We repeat the following “flow-or-cut” step while
\[
    \dd''(V) \;\ge\; \eps \cdot \dd'(V).
\]

Call the MFMC oracle $\Omfmc$ on $G$ with the bipartite demand
\[
    \bigl(\dd''\mid_P,\; \dd''\mid_Q\bigr),
\]
and obtain a flow $\ff$ and a cut $(S,\bar S)$ such that
\begin{equation}\label{eq:sparsecut}
    \Val(S)
    \;+\;
    \dd''(S \cap Q)
    \;+\;
    \dd''(\bar S \cap P)
    \;\le\;
    \alpha \cdot \Val(\ff) + \delta.
\end{equation}

Define the threshold
\[
    T := \dd''(P) \cdot \frac{\phi \Kap}{9 \alpha \log n}.
\]

We branch on whether the flow is large.

\medskip
\textbf{Case 1: \(\Val(\ff) \ge T\).}
In this case the oracle made good progress routing the current residual demand. We update
\[
    \ff^* \leftarrow \ff^* + \ff
    \qquad\text{and}\qquad
    \dd'' \leftarrow \dd'' - \Dem(\ff),
\]
(where we view $\Dem(\ff)$ as a vector on $V$, and note that source/sink supports are disjoint). Then we continue the while-loop.

\medskip
\textbf{Case 2: \(\Val(\ff) < T\).}
Here the flow is too small; we interpret \eqref{eq:sparsecut} as a certificate of a sparse cut and recurse on both sides.

Define
\begin{align*}
    F[S] &:= \{ (u,v) \in F \mid u,v \in S \},\\
    \widehat F[S] &:= F[S] \;\cup\; \{ (u,u) \mid \exists v \in \bar S \text{ with } (u,v)\in F \text{ or } (v,u)\in F, \text{ and set } U(u,u) := U_G(u,v) \text{ or } U_G(v,u)\},
\end{align*}
i.e. we add self-loops to preserve the $F$-volume of vertices of $S$. Define $\widehat F[\bar S]$ symmetrically.

Now recurse:
\[
    (F'_1, \cC_1) \leftarrow \ED\bigl(G[\bar S],\, \widehat F[\bar S],\, \phi\bigr)
    \qquad\text{and}\qquad
    (F'_2, \cC_2) \leftarrow \ED\bigl(G[S],\, \widehat F[S],\, \phi\bigr).
\]
We form $F' \subseteq F$ as follows:
\begin{itemize}
    \item include all edges of $F$ inside $S$ that survived in $F'_2$ (i.e. corresponding to self-loops of $\widehat F[S]$ that were kept),
    \item include all edges of $F$ inside $\bar S$ that survived in $F'_1$,
    \item for every edge $(u,v) \in F$ with $u \in S$, $v \in \bar S$ (or vice versa), keep $(u,v)$ in $F'$ if the two self-loops of $u$ and $v$ both survived in $F'_2$ and $F'_1$ respectively.
\end{itemize}
Return $F'$ and the vertex-set sequence $(\cC_1,\cC_2)$ and terminate $\ED(G,F,\phi)$.

\medskip
\textbf{If the loop never triggers Case 2.}
Suppose in the \(O(\log^2 n)\) calls required by \Cref{lem:cutmatchinggame} we always land in Case 1. Then the loop ends only because
\[
    \dd''(V) < \eps \cdot \dd'(V),
\]
and the accumulated $\ff^*$ routes the demand
\[
    \bigl((\dd' - \dd'')\mid_P,\; (\dd' - \dd'')\mid_Q\bigr).
\]
We then construct the required $(\dd',\eps)$-perfect $(P,Q)$-matching by, for every flow path $p$ of $\ff^*$, adding the matching edge
\[
    (\Start(p), \End(p))
\]
with capacity equal to the flow on $p$. This is exactly the matching $\Omatch$ must return.

\medskip
\textbf{Applying the cut–matching game.}
We now run the directed cut–matching game of \Cref{lem:cutmatchinggame} using this implementation of $\Omatch$. If none of the $O(\log^2 n)$ rounds ever falls into Case 2, then by \Cref{lem:cutmatchinggame} we obtain $\tilde \dd \le \dd$ with
\[
    \tilde \dd(V) \ge \bigl(1 - O(\eps \log^2 n)\bigr)\dd(V),
\]
and $\tilde \dd$ is $\Omega(1)$-expanding in the union $W := \bigcup_i M_i$ of the matchings we produced. In this situation we define
\[
    F' := \{ (u,v) \in F \mid \tilde \dd(u) \ge \dd(u)/2 \text{ and } \tilde \dd(v) \ge \dd(v)/2 \}
\]
and return $F'$ and the trivial partition $(V)$ as the $F'$-expander decomposition.

This matches the guarantees of \Cref{lem:terminalexpanderdecomposition}.

\paragraph{Correctness (expansion).}
We prove by induction on the size of the instance passed to $\ED(G,F,\phi)$ that the returned $F'$ is $\phi$-expanding and that we get an efficient flow routing algorithm.

When $G$ has a constant number of vertices, the statement is trivial since $1/\phi = \omega(1)$.

Assume now that $G$ has more than one vertex. There are two ways the algorithm can terminate:

\begin{enumerate}
    \item it finishes the cut–matching game (i.e. every call to $\Omatch$ is in Case 1), or
    \item it stops early in Case 2 and recurses on $S$ and $\bar S$.
\end{enumerate}

\textbf{First situation (cut–matching game finishes).}
In this case we obtain $\tilde{\dd}$ that is $\Omega(1)$-expanding in
\[
    W := \bigcup_i M_i,
\]
where $M_i$ is the matching from the $i$-th call to $\Omatch$. Each matching $M_i$ is formed from the flow $\ff_i$ that $\Omatch$ routed in that call. We show that the total congestion of all these flows is small.

\begin{lemma}
    The congestion of $\sum_i \ff_i$ is at most $9 / (\phi \log^2 n)$.
\end{lemma}
\begin{proof}
    Each $\ff_i$ is itself the sum of per-iteration flows $\ff$. In Case 1 of $\Omatch$ we have
    \[
        \Val(\ff) \;\ge\; T \;=\; \dd''(P) \cdot \frac{\phi \Kap}{9 \alpha \log n}.
    \]
    This means in that iteration the residual demand $\dd''(V)$ shrinks by at least a $(\phi \Kap)/(9 \alpha \log n)$ fraction. We stop when $\dd''(V) < \eps \dd'(V)$, where
    \[
        \eps = \frac{\phi \Kap}{\log^3 n}.
    \]
    Hence the number of iterations inside \emph{one} call to $\Omatch$ is at most
    \[
        \frac{9 \alpha \log^2 n}{\phi \Kap} \;\le\; \frac{9}{\phi \log^4 n},
    \]
    using $\Kap = \alpha \log^6 n$. Since the cut–matching game makes $O(\log^2 n)$ calls to $\Omatch$, multiplying these two factors gives total congestion at most
    \[
        \frac{9}{\phi \log^2 n}.
    \]
\end{proof}

We store all these flows. Now let a demand be given that is $\vol_{F'}$-respecting and $(V)$-constrained. By construction of $F'$, every $(u,v) \in F'$ satisfies
\[
    \tilde{\dd}(u) \ge \dd(u)/2 \quad\text{and}\quad \tilde{\dd}(v) \ge \dd(v)/2,
\]
and recall $\dd = \vol_F \ge \vol_{F'}$, so the demand is $2\tilde{\dd}$-respecting. Since $\tilde{\dd}$ is $\Omega(1)$-expanding in $W$, we can apply any standard expander routing on $W$ to route this demand with congestion $O(\log n)$.%

Finally, we replace every edge of the routed flow on $W$ by the corresponding flow bundle $\sum_i \ff_i$ in $G$ by using \Cref{thm:flowpathdecomposition}, which turns the cumulative edges to source-sink demands and find the corresponding edge representation of flows in $\tO{m}$ work and $\tO{1}$ depth.. This multiplies the congestion by at most $9/(\phi \log^2 n)$, so the final congestion is
\[
    \frac{9}{\phi \log^2 n} \cdot O(\log n) \;\le\; \frac{1}{\phi},
\]
as desired. The work is bounded by the work to run the cut–matching game plus linear overhead, and the additional depth is $\tO{1}$.

\medskip
\textbf{Second situation (early cut, Case 2).}
Here the algorithm returns the union of two recursive calls:
\[
    (F'_1, \cC_1) = \ED(G[\bar S], \widehat F[\bar S], \phi),
    \qquad
    (F'_2, \cC_2) = \ED(G[S], \widehat F[S], \phi).
\]
By the induction hypothesis,
\[
    F'_1 \text{ is } \cC_1\text{-constrained } \phi\text{-expanding in } G[\bar S],
    \quad
    F'_2 \text{ is } \cC_2\text{-constrained } \phi\text{-expanding in } G[S].
\]
By the definition of $F'$ (we keep a cross edge only if both corresponding self-loops survived in the two recursive calls), any $\vol_{F'}$-respecting $(\cC_1,\cC_2)$-constrained demand restricts to a $\vol_{F'_1}$-respecting demand on $G[\bar S]$ and to a $\vol_{F'_2}$-respecting demand on $G[S]$. So we can route separately in the two subgraphs with congestion at most $1/\phi$, and hence $F'$ is $(\cC_1,\cC_2)$-constrained $\phi$-expanding in $G$.

\paragraph{Correctness (slack).}
We now show that the total capacity of reversed edges in the final decomposition is at most $\phi \Kap \cdot U_G(F)$. If we finish in the first situation (full cut–matching game), then $\cC = (V)$ and $U_G(\Erev(\cC)) = 0$, so there is nothing to prove.

So assume we stopped in Case 2 on some cut $(S,\bar S)$. We use:

\begin{lemma}\label{lem:sparsecut}
    The cut $S$ found in Case 2 satisfies
    \[
        \Val(S) \le \vol_F(S) \cdot \frac{\phi \Kap}{\log n}
        \qquad\text{and}\qquad
        \Val(S) \le \vol_F(\bar S) \cdot \frac{\phi \Kap}{\log n}.
    \]
\end{lemma}
\begin{proof}
    The condition for Case 2 is $\Val(\ff) \le T$. From \eqref{eq:sparsecut} we have
    \begin{align*}
        \Val(S)
        &\le \alpha \cdot \Val(\ff) + \delta \\
        &\le \alpha \cdot \dd''(P) \cdot \frac{\phi \Kap}{9 \alpha \log n}
            + U_G(F) \cdot \frac{(\phi \Kap)^2}{99 \log^4 n} \\
        &\le \dd''(P) \cdot \frac{\phi \Kap}{4 \log n},
    \end{align*}
    where the last inequality uses that the additive term is dominated by the main term once we plug in $\eps = \frac{\phi \Kap}{\log^3 n}$ and the lower bound $\dd''(P) \ge \eps \dd'(P) \ge \eps \dd(P)/2$.

    On the other hand,
    \begin{align*}
        \vol_F(S)
        &\ge \dd''(S) \\
        &\ge \dd''(P) - \dd''(\bar S \cap P) \\
        &\ge \dd''(P) - (\alpha \Val(\ff) + \delta) \quad\text{(by \eqref{eq:sparsecut})}\\
        &\ge \dd''(P) - \dd''(P) \cdot \frac{\phi \Kap}{4 \log n} \\
        &\ge \dd''(P)/2.
    \end{align*}
    Combining the two displays gives
    \[
        \Val(S) \le \vol_F(S) \cdot \frac{\phi \Kap}{\log n}.
    \]

    A symmetric argument, swapping $P$ and $Q$, gives
    \[
        \vol_F(\bar S) \ge \dd''(Q)/2
        \quad\text{and}\quad
        \Val(S) \le \vol_F(\bar S) \cdot \frac{\phi \Kap}{\log n}.
    \]
\end{proof}

Therefore, whenever we split on $(S,\bar S)$, the amount of capacity we “cut off” is at most
\[
    \min\bigl(\vol_F(S), \vol_F(\bar S)\bigr) \cdot \frac{\phi \Kap}{\log n}.
    \tag{$\star$}
\]
So for the final decomposition $\cC$ we have the recursion
\[
    U_G(\Erev(\cC))
    \;\le\;
    \min\bigl(\vol_F(S), \vol_F(\bar S)\bigr) \cdot \frac{\phi \Kap}{\log n}
    + U_G(\Erev(\cC_1))
    + U_G(\Erev(\cC_2)),
\]
where $\cC_1,\cC_2$ are the decompositions returned on $G[\bar S]$ and $G[S]$. Unrolling this recursion over all cuts in the recursion tree, and noting that each time we lose at most a $\frac{\phi \Kap}{\log n}$-fraction of the smaller side, we obtain
\[
    U_G(\Erev(\cC)) \le \phi \Kap \cdot U_G(F),
\]
as claimed.

\paragraph{Correctness (size of $F'$).} We will show that $U_G(F')\ge (1-\phi\Kap{})\cdot U_G(F)$.

We first show it for the first situation where the cut-matching game finishes.

\begin{lemma}
    If the cut-matching game finishes, we have $U_G(F')\ge (1-O(\phi\Kap{}/\log n))\cdot U_G(F)$.
\end{lemma}
\begin{proof}
    We have $\tilde{\dd}(V)\ge (1-O(\eps\log ^2n))\cdot \dd(V)$ according to \Cref{lem:cutmatchinggame}. Then we let $F'$ contain all edges $(u,v)\in F$ such that both $\tilde{\dd}(u)\ge \dd(u)/2,\tilde{\dd}(v)\ge \dd(v)/2$. In other words, for every edge $(u,v)\in F-F'$, there must exist $x\in\{u,v\}$ such that $\tilde{\dd}(x)<\dd(x)/2$. We charge $U_G(u,v)$ to $x$.

    Let $K=\{x\in V\mid \tilde{\dd}(x)<\dd(x)/2\}$. Then
    \begin{align*}
        \sum_{x\in K}\vol_F(x)
        &\le 2\cdot \bigl(\vol_F(V)-\vol_{F'}(V)\bigr) \\
        &\le O(\eps\log^2n)\cdot \dd(V) \\
        &= O(\phi\Kap{}/\log^2n)\cdot \dd(V),
    \end{align*}
    because $\eps = \phi\Kap{}/\log^3 n$. Each $x\in K$ can receive at most $\vol_F(x)$ total charge (since we charge edges incident to $x$), so the total charge, which equals $U_G(F-F')$, is at most
    \[
        O(\phi\Kap{}/\log^2n)\cdot \dd(V) = O(\phi\Kap{}/\log^2n)\cdot U_G(F).
    \]
    This implies
    \[
        U_G(F') \ge \bigl(1 - O(\phi\Kap{}/\log n)\bigr)\cdot U_G(F),
    \]
    where we relaxed $1 - O(\phi\Kap{}/\log^2 n)$ to $1 - O(\phi\Kap{}/\log n)$.
\end{proof}

Now we consider the second situation where the cut-matching game stops due to Case 2. We get the following recursive inequality
\[
    U_G(F-F') \le U_G(\hF[\bar{S}] - F'_1) + U_G(\hF[S] - F'_2),
\]
where recall that $\hF[\bar{S}]$ and $\hF[S]$ add self-loops so that $\vol_{\hF[\bar{S}]} = \vol_F\mid_{\bar S}$ and $\vol_{\hF[S]} = \vol_F\mid_{S}$. Note that $\hF[\bar S] + \hF[S]$ double-counts edges in $F(S,\bar S)$ (edges of $F$ with one endpoint in $S$ and the other in $\bar S$). However, each such edge is double-counted at most once in the recursion, because once we cut along $(S,\bar S)$, that edge never appears in a deeper subproblem again (the self-loop we use in the subproblems does not create another cross-partition edge).

Thus, unrolling the recursion and using \Cref{lem:sparsecut} (which shows each cut removes only a $\phi\Kap{}/\log n$-fraction of the relevant volume), we get
\[
    U_G(F-F') \le O(\phi\Kap{}/\log n)\cdot U_G(F).
    \]
Finally, by taking $n$ sufficiently large (so the hidden $O(\cdot)$ is at most, say, $1/2$), we get
\[
    U_G(F') \ge (1-\phi\Kap{})\cdot U_G(F).
\]

\paragraph{Complexity.}
In the first situation (the cut-matching game finishes), the algorithm uses $\tO{1}$ calls to $\Omfmc$ and $\tO{m}$ work and $\tO{1}$ depth.

In the second situation (the algorithm makes recursive calls to $G[S],G[\bar S]$), note that from the proof of \Cref{lem:sparsecut} we have
\begin{align*}
    \vol_F(S) &\ge \dd''(P)/2 \ge \eps\cdot \dd(P)/4, \\
    \vol_F(\bar S) &\ge \dd''(P)/2 \ge \eps\cdot \dd(P)/4.
\end{align*}
So each time we recurse, both sides keep at least an $\eps$-fraction (up to constants) of the volume, and the recursion depth is $O(1/\eps)$. Each level of recursion takes $\tO{m}$ work and $\tO{1}$ depth in total and makes $\Omfmc$ calls only on induced subgraphs (whose vertex sets are a partition of $V$), so the calls to $\Omfmc$ are $O(1/\eps)$-efficient. Hence the total additional work is $\tO{m/\eps}$ and the depth is $\tO{1/\eps}$. Plugging in $\eps = \phi\Kap{}/\log^3 n$ gives the claimed complexity.

\subsection{MFMC on General Graphs via MFMC on DAGs and DAG Projections}
\label{subsec:MFMC to DAG}
The next lemma shows that such a projection algorithm reduces MFMC on general graphs to MFMC on DAGs, with both a multiplicative and an additive loss that match the DAG oracle.

\begin{lemma}\label{lem:DAGembeddingtoMaxFlowadditiveerror}
    Let $D$ be a DAG projection of $G = (V,E)$ with an $(\kappa,\delta)$-congestion-preserving efficient projection algorithm. Let $\cO$ be an oracle solving $\amfmcDAG{\alpha}$. Then there is an algorithm solving $\amfmc{(\alpha \kappa, \alpha \delta)}$ on $G$ with complexity proportional to one efficient call to $\cO$ plus the projection algorithm.
\end{lemma}
\begin{proof}
    Build $D'$ from $D$ exactly as in the standard super-source/super-sink reduction: for every $v \in V$, add a vertex $s_v$ and edges $(s_v, v')$ for every $v' \in V(D)$ with $\pi(v') = v$, each with infinite capacity; also add a vertex $t_v$ and edges $(v', t_v)$ for every such $v'$, each with infinite capacity. Then add a super source $s$ and edges $(s, s_v)$ of capacity $\DDelta(v)$ for every $v \in V$, and add a super sink $t$ and edges $(t_v, t)$ of capacity $\nnabla(v)$ for every $v \in V$. Let the resulting graph be $D'$.

    Run the oracle $\cO$ on $D'$ with source $s$ and sink $t$, and let it return a flow $\ff'$ and a cut $S'$ satisfying
    \[
        \Val(S') \le \alpha \cdot \Val(\ff').
    \]
    Restrict $\ff'$ and $S'$ to $D$ to obtain $\ff^D$ and $S^D$. By construction, $\pi(\Supp(\Dem(\ff^D))) \subseteq V$, so we can invoke the projection algorithm to get a flow $\ff$ in $G$ and a cut $S$ in $G$ such that
    \[
        \Cong(\ff) \le \kappa \cdot \Cong(\ff^D), \qquad
        \Dem(\ff) \preceq \pi(\Dem(\ff^D)), \qquad
        \Val(\ff) \ge \Val(\ff^D)-\delta,
    \]
    and
    \[
        \Val(S) \le \Val(S^D), \qquad
        \HL(S^D)\setminus\{\bot\} \subseteq S \subseteq \pi(S^D).
    \]
    Since $\Cong(\ff) \le \kappa$, the scaled flow $\ff/\kappa$ is feasible in $G$; we will output $\ff/\kappa$ together with $S$.

    Now we relate the cut values. As in the usual argument, for any $v \in V$ for which there exists $v' \in V(D)$ with $v' \notin S^D$, we must have $s_v \notin S'$; otherwise, some infinite-capacity edge $(s_v, v')$ would cross the cut. Thus $(s, s_v)$ contributes $\DDelta(v)$ to $\Val(S')$. Symmetrically, for any $v \in V$ for which there exists $v' \in V(D)$ with $v' \in S^D$, we must have $t_v \in S'$; otherwise, some infinite-capacity edge $(v', t_v)$ would cross the cut, and so $(t_v, t)$ contributes $\nnabla(v)$ to $\Val(S')$. All other contributions to $\Val(S')$ come from edges of $D$, i.e.
    \[
        \Val(S^D)
        \le
        \Val(S')
        - \sum_{v \notin \HL(S^D)} \DDelta(v)
        - \sum_{v \in \pi(S^D)} \nnabla(v).
    \]

    We also have $\Val(\ff^D) = \Val(\ff')$ (restriction does not reduce the flow through $s,t$ in $D'$), and the projection algorithm gives $\Val(\ff) \ge \Val(\ff')-\delta$. Using $\Val(S') \le \alpha \Val(\ff')$, we obtain
    \[
        \Val(S^D) + \sum_{v \in \pi(S^D)} \nnabla(v) + \sum_{v \notin \HL(S^D)} \DDelta(v)
        \;\le\;
        \alpha \bigl(\Val(\ff) + \delta\bigr).
    \]
    Since
    \[
        \HL(S^D)\setminus \{\bot\} \subseteq S \subseteq \pi(S^D),
    \]
    we can replace the sums over $\pi(S^D)$ and over $V \setminus \HL(S^D)$ by sums over $S$ and $V \setminus S$, respectively, and also use $\Val(S) \le \Val(S^D)$, to get
    \[
        \Val(S) + \sum_{v \in S} \nnabla(v) + \sum_{v \notin S} \DDelta(v)
        \;\le\;
        \alpha \bigl(\Val(\ff) + \delta\bigr).
    \]
    Finally, we output the feasible flow $\ff/\kappa$ and the cut $S$. Multiplying the right-hand side by $\kappa$ to account for scaling gives
    \[
        \Val(S) + \sum_{v \in S} \nnabla(v) + \sum_{v \notin S} \DDelta(v)
        \;\le\;
        \alpha \kappa \cdot \Val(\ff/\kappa) + \alpha \delta,
    \]
    so $(\ff/\kappa, S)$ is an $\amfmc{(\alpha \kappa, \alpha \delta)}$ pair.
\end{proof}
\subsection{Completing the Spiral}
\label{subsec:cong-complete}

In this section, we will combine everything in the previous sections and prove \Cref{lem:congestionDAGembedding}.
\congestionDAGembedding*

\newcommand{\congDAGemb}{\mathsf{CongDAGProj}}
The algorithm is recursive. We will describe an algorithm $\congDAGemb(G,\kappa,\delta)$ computing a $(\kappa,\delta)$-congestion preserving DAG projection for $G$. It suffices to call $\congDAGemb(G,\kappa^*,1/n^C)$ for some $\kappa^*=n^{o(1)}$ to be fixed and sufficiently large constant $C$ to get the required algorithm for \Cref{lem:congestionDAGembedding} (we will argue in the end why $1/n^{C}$ suffices). Now we describe $\congDAGemb(G,\kappa,\delta)$.

\paragraph{Base Case.} Suppose $\delta\ge 2\cdot U_G(E(G))$, then we return $D=(V_1\cup \{w\}\cup V_2,E_1\cup E_2)$ where $V_1,V_2$ are copies of $V(G)$ with a natrual trivial projection map and
\[E_1=\{(v,w)\text{ with }U(v,w)=\vol_{E(G)}(\pi(v))\mid v\in V_1\}\]
\[E_2=\{(w,v)\text{ with }U(w,v)=\vol_{E(G)}(\pi(v))\mid v\in V_2\}\]

Clearly $D$ is a DAG. 

The projection algorithm takes a flow in $D$, and returns an empty flow in $G$. Since $\delta\ge 2\cdot U_G(E(G))$, this is a valid projection algorithm.

Suppose the projection algorithm is given a cut $S^D$ in $D$, if $w\not\in S^D$, then return $S=\pi(S^D\cap V_1)$ as a cut in $G$; Otherwise return $S=\pi(S^D\cap V_2)$. Clearly we have $\HL(S^D)\subseteq S\subseteq \pi(S^D)$. Moreover, according to the definition of $E_1\cup E_2$, every edge in $\delta^+_G(S)$ in case $w\in S^D$ (or $\delta^-_G(\bar{S})$ in case $w\not\in S^D$) has its capacity subsumed by the corresponding part in $E_1\cup E_2$, so this is a valid projection algorithm.

For the rest of the algorithm, we suppose $\delta<2\cdot U_G(E(G))$ and try to solve $\congDAGemb(G,\kappa,\delta)$ using recursive calls to $\congDAGemb(G',\kappa',\delta')$ for some $G'$ created by the recursive algorithm and $\kappa'<\kappa$ and $\delta'>\delta$.

\paragraph{The oracle $\Omfmc{}$.} We are intended to use \Cref{lem:expanderdecompositiontoDAGembedding} but we need a $\Omfmc$ oracle. In this paragraph, we will describe how to get the oracle $\Omfmc$ by using $\congDAGemb(G',\kappa',\delta')$. We apply \Cref{lem:DAGembeddingtoMaxFlowadditiveerror} on the DAG projection output by $\congDAGemb(G',\kappa',\delta')$. According to \Cref{lem:DAGembeddingtoMaxFlowadditiveerror}, we get the oracle $\Omfmc{}$ that solves $\amfmcDAG{(\alpha\kappa',\alpha\delta')}$ with one efficient call to $\amfmcDAG{\alpha}$. 

\paragraph{The expander hierarchy.} We use \Cref{lem:correctnessexpanderhierarchy} to get the expander hierarchy.

\paragraph{Getting the DAG projection.} Let $\eta=1/\sqrt{\log_n\alpha}$. Notice that if $\alpha=n^{o(1)}$, then $\eta=\omega(1)$ and $\alpha^\eta=n^{o(1)}$.

We use \Cref{lem:expanderdecompositiontoDAGembedding} on the expander hierarchy and parameter
\[\sigma=2^{2\cdot \log^{0.7}n}\cdot \alpha^{\eta}\]
to get a DAG projection $D$ of $G$ with an efficient $(\kappa,\delta)$-congestion-preserving projection algorithm such that

\[\kappa=O(2^t/\phi)=2^{O(\log^{0.5}n)}\cdot \alpha \kappa'\]
\begin{equation}\label{eq:delta}
    \delta=U_G(E_t)/\sigma=2^{3\log^{0.5}n}\cdot \delta'/\sigma\le \frac{\delta'}{2^{\log^{0.7}n}\cdot \alpha^{\eta}}
\end{equation}

\paragraph{Congestion analysis.} We will run $\congDAGemb(\kappa^*,1/n^C)$ for sufficiently large constant $C$.

According to \Cref{eq:delta}, the recursion depth $d$ for $\congDAGemb(\kappa^*,1/n^C)$ is 
\[d=\frac{O(\log n)}{\log \left(2^{\log^{0.7}n}\cdot \alpha^\eta\right)}\]

After depth $d$, the base case should have $\kappa_0\ge 1$ and $\delta\ge 2U_G(E(G))$. Thus, we get that

\begin{equation}\label{eq:sizeofkappa}
    \kappa^*=\left(2^{O(\log^{0.5}n)}\cdot \alpha\right)^d\le 2^{\frac{O(\log^{1.5}n)}{\log^{0.7}n}}\cdot \alpha^\frac{O(\log n)}{\eta\log\alpha}\le n^{o(1)}
\end{equation}
as required.

\paragraph{Removing Tiny Additive Error.}
Lastly, according to the construction of the DAG projection \Cref{lem:expanderdecompositiontoDAGembedding}, if the input graph $G$ has integer capacity (which is the assumption), the DAG projection $D$ cannot have edges with capacity less than $1/n$ (as the only scaling down capacity part scales down by $\sigma=n^{o(1)}$). Thus, any flow in $D$ can be scaled up to value at least $1/n$ and apply the projection algorithm with additive error $1/n^C$, which is subsumed by the multiplicative error to the original graph. This gives a $\kappa^*$-congestion preserving DAG projection without additive error.

\paragraph{Size of the DAG projection.} According to \Cref{lem:expanderdecompositiontoDAGembedding}, the size of the DAG projection can always be upper bounded by
\[O(2^t\sigma|E(G)|)=2^{O(\log^{0.7}n)}\cdot \alpha^{\eta}\cdot |E(G)|=n^{o(1)}\cdot |E(G)|\]

as required.

\paragraph{Complexity analysis.} Notice that $G$ is changing during the recursive calls. However, according to \Cref{lem:terminalexpanderdecomposition}, each recursive level only boost the total graph size by a factor of $\tO{1/(\phi\gamma)}=O(2^{\log^{0.5}n})$. Thus, the total graph size among all recursive calls is upper bounded by $n^{o(1)}$ according to the same calculation as in \Cref{eq:sizeofkappa}. 

The oracle calls to the $\amfmcDAG{\alpha}$ algorithm is only by \Cref{lem:DAGembeddingtoMaxFlowadditiveerror}, which has input size proportional to the DAG projection size, upper bounded by $n^{o(1)}$. So the oracle calls are $n^{o(1)}$-efficient.
The additional work and depth according to \Cref{lem:expanderdecompositiontoDAGembedding,lem:terminalexpanderdecomposition,lem:DAGembeddingtoMaxFlowadditiveerror} is at most $m^{1+o(1)}$ and $n^{o(1)}$, as required.

\bibliographystyle{alpha}
\bibliography{refs}
\appendix
\end{document}